\documentclass[reqno,a4paper]{amsart}
\usepackage[utf8]{inputenc}
\usepackage{amsmath, amssymb, amsthm, amscd}
\usepackage{bbm}
\usepackage{mathrsfs,mathtools}
\usepackage[text={5.8in,8.8in},centering]{geometry}
\usepackage[usenames,dvipsnames,svgnames,table]{xcolor}%Introduced by VG
\usepackage{tikz}
\usepackage{caption}

\usepackage{paralist}

\makeatletter\@addtoreset{footnote}{page}\makeatother
\renewcommand{\thefootnote}{\ifcase\value{footnote}\or(*)\or
(**)\or(***)\or(****)\or(\#)\or(\#\#)\or(\#\#\#)\or(\#\#\#\#)\or($\infty$)\fi} 

\usepackage[pagebackref,hypertexnames=true]{hyperref}
\hypersetup{colorlinks=true,linkcolor=RoyalBlue,%
urlcolor=RoyalBlue,citecolor=RoyalBlue}
\newtheorem{theorem}{Theorem}[section]
\newtheorem{proposition}[theorem]{Proposition}
\newtheorem{lemma}[theorem]{Lemma}
\newtheorem{definition}[theorem]{Definition}
\newtheorem{corollary}[theorem]{Corollary}
\newtheorem{remark}[theorem]{Remark}

\newtheorem{example}[theorem]{Example}

\def\bbbone{{\mathchoice {\rm 1\mskip-4mu l} {\rm 1\mskip-4mu l}
{\rm 1\mskip-4.5mu l} {\rm 1\mskip-5mu l}}}

\def\one{\bbbone}
\newcommand{\per}{\mathrm{perp}}
\newcommand{\sperp}{{\,\s\!\perp}}
\newcommand{\Span}{\mathrm{Span}}
\renewcommand{\i}{\mathrm{i}}

\newcommand{\cH}{\mathcal{H}}
\newcommand{\cB}{\mathcal{B}}
\newcommand{\cZ}{\mathcal{Z}}
\newcommand{\cC}{\mathcal{C}}

\newcommand{\cI}{\mathcal{I}}

\newcommand{\cU}{\mathcal{U}}
\newcommand{\cK}{\mathcal{K}}
\newcommand{\cL}{\mathcal{L}}
\newcommand{\cV}{\mathcal{V}}

\newcommand{\cW}{\mathcal{W}}

\newcommand{\cG}{\mathcal{G}}

\newcommand{\cY}{\mathcal{Y}}
\newcommand{\cD}{\mathcal{D}}
\newcommand{\cb}{\mathcal{B}}

\newcommand{\Ran}{\mathcal{R}}
\newcommand\supp{\mathrm{supp}}
 \newcommand{\Ker}{\mathcal{N}}

\def\braket#1#2{\langle{#1}|{#2}\rangle}
\def\bracet#1#2{({#1}|{#2})}
\def\ketbra#1#2{|{#1}\rangle\langle{#2}|}

\newcommand{\M}{{\max}}
\newcommand{\m}{{\min}}
\newcommand{\wlim}{\mathrm{w}-\lim}
\newcommand{\rs}{\mathrm{rs}}

\newcommand{\spec}{\mathrm{sp}}

\newcommand{\ind}{\mathrm{ind}}
\newcommand{\s}{\mathrm{s}}

\renewcommand\Re{\operatorname{Re}}
\renewcommand\Im{\operatorname{Im}}
\newcommand{\Tr}{\mathrm{Tr}\,}

\newcommand{\bbC}{\mathbb{C}}

\newcommand{\loc}{\mathrm{loc}}

\renewcommand{\bar}{\overline}

\def\lbra {{[\![}}
\def\rbra {{]\!]}}

\def\R{\mathbb R}
\def\C{\mathbb C}

\def\P{\mathbb{P}}

\def\B{\mathcal B}
\def\d{\mathrm d}
\def\x{\mathsf x}

\newcounter{smallroman}
\newenvironment{romanenumerate}
{\begin{list}{{\normalfont\textrm{(\roman{smallroman})}}}
  {\usecounter{smallroman}\setlength{\itemindent}{0cm}
   \setlength{\leftmargin}{5ex}\setlength{\labelwidth}{4ex}
   \setlength{\topsep}{0.75\parsep}\setlength{\partopsep}{0ex}
   \setlength{\itemsep}{0ex}}}
{\end{list}}

\begin{document}

\title[1-dimensional Schr\"odinger operators with complex potentials]
{1-dimensional Schr\"odinger operators\\
  with complex potentials \\[2mm]
  {\tiny\today}
}

\author[J. Derezi\'nski]{Jan Derezi\'nski}
\address[Jan Derezi\'nski]{Department of Mathematical Methods in
  Physics, Faculty of Physics, University of Warsaw, Pasteura 5,
  02-093 Warszawa, Poland} 
\email{jan.derezinski@fuw.edu.pl}

\author[V. Georgescu]{Vladimir Georgescu} \address[Vladimir
Georgescu]{D\'epartement de Math\'ematiques, Universit\'e de
  Cergy-Pontoise, 95000 Cergy-Pontoise, France}
\email{vladimir.georgescu@math.cnrs.fr}

\begin{abstract}
  We discuss realizations of $L:=-\partial_x^2+V(x)$ as closed
  operators on $L^2]a,b[$, where $V$ is complex, locally integrable
  and may have an arbitrary behavior at (finite or infinite) endpoints
  $a$ and $b$.  The main tool of our analysis are Green's operators,
  that is, various right inverses of $L$.
\end{abstract}

\maketitle

{\hypersetup{linkcolor=black} % RoyalBlue
  \tableofcontents
}

\section{Introduction}

The paper is devoted to operators of the form 
\begin{equation}
  L=-\partial_x^2+V(x)\label{sturm}
\end{equation}
on $]a,b[$, where $a<b$, $a$ can be $-\infty$ and $b$ can be
$\infty$. The potential $V$ can be complex, have low regularity, and a
rather arbitrary behavior at the boundary of the domain: we assume
that $V\in L_\loc^1]a,b[$.  We study realizations of $L$ as closed
operators on $L^2]a,b[$.

Operators of the form (\ref{sturm}) are commonly called {\em
  1-dimensional Schr\"odinger operators} or, shorter, {\em 1d
  Schr\"odinger operators.}  They are special cases of {\em
  Sturm-Liouville operators}, that is operators of the form
\begin{equation}
 -\frac{1}{w(x)}\partial_xp(x)\partial_x+\frac{q(x)}{w(x)}.\label{sturm-l}
\end{equation}
Note, however, that if $\frac{p(x)}{w(x)}$ is real, then under rather
weak assumptions on $w,p,q$, a simple unitary transformation reduces
(\ref{sturm-l}) to (\ref{sturm}). Therefore, it is not a serious
restriction to consider 1d Schr\"odinger operators instead of
Sturm-Liouville operators.

1d Schr\"odinger operators is a classic subject with a lot of
literature.  Most of the literature is devoted to  real $V$, when
$L$ can be realized as self-adjoint operator.  It is,
however, quite striking that the usual theory well-known from the real
(self-adjoint) case works almost equally well in the complex case.  In
particular, essentially the same theory for boundary conditions and
the same formulas for {\em Green's operators} (right inverses of (\ref{sturm}))
hold as in the real case. We will describe these topics in detail in
this paper.

A large part of the literature on 1d Schr\"odinger operators assumes
that potentials are $L^1$ near finite endpoints. Under this condition
one can impose the so called {\em regular boundary conditions}
(Dirichlet, Neumann or Robin).  In this case, it is natural to use the
so-called {\em Weyl-Titchmarsh function} and the formalism of the
 so-called {\em boundary triplets}, see e.g. \cite{BBMNP} and
references therein. We are interested in general boundary conditions,
such as those considered in \cite{BDG,DR,DR2}, where the above
approach does not directly apply. See the discussion at the end of Subsect. \ref{triplet}.

One of the motivations of the present work is the study of exactly
solvable Schr\"odinger operators, such as those given by the Bessel
equation \cite{BDG,DR}, or the Whittaker equation \cite{DR2}. Analysis
of those operators indicates that non-real potentials are as good from
the point of view of the exact solvability as real ones. It is also
natural to organize exactly solvable Schr\"odinger operators in
holomorphic families, whose elements are self-adjoint only in
exceptional cases. Therefore, a theory for 1d Schr\"odinger operators
with complex potentials and general boundary conditions provides a
natural framework for the study of exactly solvable Hamiltonians.

As we mentioned above, we suppose that $V\in L_\loc^1]a,b[$.  The
theory is much easier if $V\in L_\loc^2]a,b[$, because one could then
assume that the operator acts on $C_\mathrm{c}^2]a,b[$. Dealing with
potentials in $L_\loc^1$ causes of a lot of trouble---this is however
a rather natural assumption. We think that handling a more general case
forces us to better understand the problem.  Actually, one could
consider even more singular potentials: it is easy to generalize our
results to potentials $V$ being  Borel measures on $]a,b[$.

In the  preliminary Sect. \ref{s:bode} we study the inhomogeneous problem
given by the operator (\ref{sturm}) by basic ODE methods.  We
introduce some distinguished Green's operators: The {\em two-sided
  Green's operators} are related to boundary conditions on both sides.
The {\em forward} and {\em backward Green's operators} are related to
the Cauchy problem at the endpoints of the interval. These operators
belong to the most often used objects in mathematics. Usually they
appear under the guise of {\em Green's functions}, which are the
integral kernels of Green's operators.

Note that in the Hilbert space $ L^2]a,b[$ we have a natural
conjugation $f\mapsto \bar f$ and a bilinear form
$\langle f|g\rangle: =\int fg$. For an operator $T$ it is natural to
define its {\em transpose} $T^\#:=\bar{T^*}$, where the bar denotes
the complex conjugation.  We say that $T$ is {\em self-transposed} if
$T^\#=T$ (in the literature the alternative name {\em J-self-adjoint}
is also used). These concepts play an important role in the theory of
differential operators on $ L^2]a,b[$. Therefore, we devote
Sect. \ref{s:abs} to a general theory of operators in a Hilbert space
with a conjugation.  We briefly recall the theory of
restrictions/extensions of unbounded operators. The concept of
self-transposed operators turns out to be a natural alternative to the
concept of self-adjointness.  It is well-known that self-adjoint
operators are {\em well-posed} (possess non-empty resolvent set). Not every self-transposed operator is
well-posed, however they often are.

The remaining sections are devoted to realizations of $L$ given by
(\ref{sturm}) as closed operators on the Hilbert space $L^2]a,b[$. The
most obvious realizations are the {\em minimal} one $L_{\min}$ and the
{\em maximal} one $L_{\max}$.  We prove that these operators are
closed and densely defined.  Under the assumption $V\in L_\loc^1]a,b[$
the proof is quite long and technical but, in our opinion,
instructive. If we assumed $V\in L_\loc^2]a,b[$, the proof would be
easy.

At this point it is helpful to recall basic theory of 1d Schr\"odinger
operators for real potentials. One is usually interested in
self-adjoint extensions of the Hermitian operator $L_{\min}$.  They
are are situated ``half-way'' between $L_{\min}$ and $L_{\max}$. More
precisely, we have 3 possibilities:
\begin{enumerate}
\item $L_{\min}=L_{\max}$: then $L_{\min}$ is already self-adjoint.
\item The codimension of $\cD(L_{\min})$ in $\cD(L_{\max})$ is
  $2$: if $L_\bullet$ is a self-adjoint extension of $L_{\min}$, the
  inclusions
  $\cD(L_{\min})\subset \cD(L_\bullet)\subset \cD(L_{\max})$ are of
  codimension $1$.
\item The codimension of $\cD(L_{\min})$ in $\cD(L_{\max})$ is
  $4$: if $L_\bullet$ is a self-adjoint extension of $L_{\min}$, the
  inclusions
  $\cD(L_{\min})\subset \cD(L_\bullet)\subset \cD(L_{\max})$ are of
  codimension $2$.
\end{enumerate}
Note that in the literature it is common to use the theory of {\em
  deficiency indices}. The cases (1), (2), resp. (3) correspond to
$L_{\min}$ having the deficiency indices $(0,0)$, $(1,1)$ and
$(2,2)$. However, the deficiency indices do not have a straightforward
generalization to the complex case.

Let us go back to complex potentials.  Note that the Hermitian
conjugate of an operator $T$, denoted $T^*$, turns out to be less
useful than its transpose
$T^\#$. In
particular, the role of self-adjoint operators is taken up by
self-transposed operators.

By choosing a subspace of $\cD(L_{\max})$ closed in the graph topology
and restricting $L_{\max}$ to this subspace we can define a closed
operator.  Such operators will be called {\em closed realizations of
  $L$}.  We will show that in the complex case closed realizations of
$L$ possess a theory quite analogous to that of the real case.

We are mostly interested in realizations of $L$ whose domain contains
$\cD(L_{\min})$, that is, operators $L_\bullet$ satisfying
$L_\m\subset L_\bullet\subset L_\M$. Such realizations are defined by specifying boundary
conditions. Similarly as in the real case, boundary conditions are
given by functionals on $\cD(L_\M)$ that vanish on $\cD(L_\m)$. For
each of endpoints, $a$ and $b$, there is a space of functionals
describing boundary conditions.  We call the dimension of this space
the {\em boundary index at $a$}, resp.\ $b$, and denote it $\nu_a(L)$,
resp.\ $\nu_b(L)$. They can take the values $0$ or $2$ only.
Therefore, we have the following classification of operators $L$:
\begin{equation}
  \begin{array}{l}
    \text{(1)    $\dim\cD(L_\M)/\cD(L_\m)=0$, or $L_\m=L_\M$},\\
    \text{(2)    $\dim\cD(L_\M)/\cD(L_\m)=2$},\\
      \text{(3)   $\dim\cD(L_\M)/\cD(L_\m)=4$}.
  \end{array}
  \label{table1}\end{equation}
Let $\lambda\in\C$. It is natural to consider the space of solutions
of $(L-\lambda)u=0$ that are square integrable near $a$, resp.\
$b$. We denote these spaces by $\cU_a(\lambda)$,
resp. $\cU_a(\lambda)$.  We will prove that
  \begin{align}
    & \nu_a(L)=0 \Longleftrightarrow
      \dim\cU_a(\lambda)\leq1 \ \forall\lambda\in\C
      \Longleftrightarrow
      \dim\cU_a(\lambda)\leq1 \text{ for some } \lambda\in\C,
      \label{eq:1} \\
    & \nu_a(L)=2 \Longleftrightarrow \dim\cU_a(\lambda)=2 \
      \forall\lambda\in\C \Longleftrightarrow
      \dim\cU_a(\lambda)=2 \text{ for some } \lambda\in\C. 
      \label{eq:2}
\end{align}
The most useful realizations of $L$ are well-posed. Not all $L$ possess such realizations. One
can classify such $L$'s as follows. If $L$ possesses a well-posed
realization $L_\bullet$, then one of the following conditions holds:
\begin{equation} \label{table2}
  \begin{array}{l}
    \text{(1)    $\dim\cD(L_\M)/\cD(L_\bullet)=0$, or $L_\bullet=L_\M$}\\
    \text{(2)    $\dim\cD(L_\M)/\cD(L_\bullet)=1$},\\
    \text{(3)   $\dim\cD(L_\M)/\cD(L_\bullet)=2$}.
  \end{array}
\end{equation}
There is a strict correspondence between (1), (2) and (3) of
(\ref{table1}) and (1), (2) and (3) of (\ref{table2}).  In cases (1)
and (2) from Table (\ref{table2}) we describe all realizations with
nonempty resolvent set and their resolvents.  We prove that if
$L_\bullet$ is such a realization, then we can find
$u\in\cU_a(\lambda)$ and $v\in\cU_b(\lambda)$ with the Wronskian equal
to $1$, so that the integral kernel of $(L_\bullet-\lambda)^{-1}$ can
then be easily expressed in terms of $u$ and $v$.

The case (3) is much richer. We describe all realizations of $L$ that
 have {\em separated boundary conditions} (given by independent boundary conditions at $a$ and
$b$). If in addition they are self-transposed, then essentially the
same formula as in (1) and (2) gives $(L_\bullet-\lambda)^{-1}$. There
are however two other separated realizations of $L$, which are denoted
$L_a$ and $L_b$, with boundary conditions only at $a$, resp.\
$b$. They are not self-transposed, in fact, they satisfy
$L_a^\#=L_b$. Their resolvents are given by what we call {\em forward}
and {\em backward Green's operators}, which incidentally are cousins
of the {\em retarded} and {\em advanced Green's functions}, well-known
from the theory of the wave equation.

In the last section we discuss potentials with a negative imaginary
part. We show that under some weak conditions they define dissipative
1d Schr\"odinger
operators. We also describe Weyl's limit point--limit cricle method
for such potentials. For real potentials, this method allows us to
determine the dimension of $\cU_a(\lambda)$ for $\Im(\lambda)>0$: if
$a$ is limit point, then $\dim\cU_a(\lambda)=1$; if $a$ is limit
circle then $\dim\cU_a(\lambda)=2$.  The picture is more complicated
if the potential is complex: there are examples where the endpoint $a$
is limit point and $\cU_a(\lambda)$ is two-dimensional.

1d Schr\"odinger operators is one of the most classic topics in
mathematics.  Already in the first half of 19 century Sturm and
Liouville considered second order differential operators on a finite
interval with various boundary consitions. The theory was extended to
a half-line and a line in a celebrated work by Weyl.

2nd order ODE's and 1d Schr\"odinger operators are considered in many
textbooks, including Atkinson \cite{At}, Coddington-Levinson
\cite{CL}, Dunford-Schwartz \cite{DS2,DS3}, Naimark \cite{Nai}, Pryce
\cite{Pry}, de Alfaro-Regge \cite{DAR}, Reed-Simon \cite{RS2}, Stone
\cite{S}, Titchmarsh \cite{Ti}, Teschl \cite{Te},
Gitman-Tyutin-Voronov \cite{GTV}.  However, in the literature complex
potentials are rarely studied in detail, and if so, then one does not
pay attention to nontrivial boundary conditions. The monograph by
Edmunds-Evans \cite{EE} is often considered as one of the most
up-to-date source for results on this subject. Many results presented
in our article have their counterpart in the literature, especially in
\cite{EE}. Let us try to make a more detailed comparison of our paper
with the literature

Most of the material of Sect. \ref{s:bode} is standard. However, we
have not seen a separate discussion of semi-regular boundary
condition, as described in Prop. \ref{prop:pro} (2) and (3). The
definitions of the canonical bisolution $G_\leftrightarrow$, various
Green's operators $G_{u,v}$ $G_\leftarrow$, $G_\rightarrow$ and
relations between them (\ref{moja-7})--(\ref{mojj}) are implicit in
many works on the subject, however they are rarely separately
emphasized.

The material of our Sect. \ref{s:abs} on Hilbert spaces with
conjugation is to a large extent contained in Chap. 3 Sects 3 and 5 of
\cite{EE}. It is based on previous results of Vishik \cite{Vi},
Galindo \cite{Gal} and Knowles \cite{Kn}. However, our presentation
seems to be somewhat different. It shows in particular that the
existence of a self-transposed extension follows almost trivially from
a basic theory of symplectic spaces described in Appendix
\ref{a1}. Another special feature of our Sect. \ref{s:abs} is a
discussion of properties of right inverses of an unbounded
operator.

The deepest result described in our paper is probably Theorem
\ref{th:proofconj} about the characterization of boundary conditions
by square integrable solutions. This result is actually not contained
in \cite{EE}. It is based on Everitt and Zettl \cite[Theorem 9.1]{EZ},
and uses \cite[Theorem 9.11.2]{At} of Atkinson and \cite[Theorem
5.4]{Rac} of Race.

The study of Green's operators contained in Sect. \ref{sect-spectrum}
is probably to a some extent new.

A separate subject that we discuss are potentials with negative
imaginary part studied by means of the so-called Weyl limit
circle/limit point method. Here the main reference is Sims \cite{Si},
see also \cite{BCEP,EE}.

The present manuscript grew out of the Appendix of \cite{BDG} devoted
to 1d Schr\"odinger operators with the potential
$\frac{1}{x^2}$. \cite{BDG} and its follow-up papers \cite{DR,DR2}
illustrated that 1d Schr\"odinger operators with complex potentials
and unusual boundary conditions appear naturally in various
situations. Motivated by these applications, we decided to write an
exposition of basic general theory of 1d Schr\"odinger operators.

We decided to make our exposition as complete and self-contained as
possible, explaining things that are perhaps obvious to experts, but
often difficult to many readers.  We use freely the modern operator
theory---this is not the case of a large part of literature, which
often sticks to old-fashioned approaches.  We also use the terminology
and notation which, we believe, is as natural as possible in the
context we consider.  For instance, we prefer the name
``self-transposed'' to ``$J$-self-adjoint'' found in a large part of
the literature. We believe that our treatment of the subject is quite
different from the one found in the literature.  One of the main tools
that we have found useful, rarely appearing in the literature, are
right inverses, which in the context of 1d Schr\"odinger operators we
call by the traditional name of {\em Green's operators}.

\section{Basic ODE theory} \label{s:bode}
\protect\setcounter{equation}{0}

\subsection{Notations}

Recall that $a<b$, $a$ can be $-\infty$ and $b$ can be $\infty$.  The
notation $[a,b]$ stands for the interval including the enpoints $a$
and $b$, while $]a,b[$ for the interval without endpoints. $[a,b[$ and
$]a,b]$ have the obvious meaning.

In some cases one could use the notation involving either $[a,b]$ or
$]a,b[$ without a change of the meaning. For instance,
$L^p\big([a,b]\big)=L^p\big(]a,b[\big)$. For esthetic reasons, we try
to use a uniform notation---we usually write $L^p]a,b[$, dropping the
round bracket for brevity.

In other cases, the choice of either $[a,b]$ or $]a,b[$ influences the
meaning of a symbol. For instance, $C[a,b]\subsetneq C]a,b[$. For
example, $f\in C[-\infty,b]$ implies that
$\lim\limits_{x\to-\infty}f(x)=:f(-\infty)$ exists.
    
$f\in L_\loc^p]a,b[$ iff for any $a<a_1<b_1<b$, we have
$f\Big|_{]a_1,b_1[}\in L^p]a_1,b_1[$.  $f\in L_\loc^p[a,b[$ if in
addition $f\in L^1[a,a_1[$ and similarly for $L_\loc^p]a,b]$.
      
$f\in L_\mathrm{c}^p]a,b[$ iff $f\in L^p]a,b[$ and $\supp f$ is a
compact subset of $]a,b[$.  The analogous meaning has the subscript
$\mathrm{c}$ in different situations.

$\oplus$ will mean the topological direct sum of two spaces.
          
\subsection{Absolutely continuous functions}

We will denote $f'$ or $\partial f$ the derivative of a distribution
$f$.  We will denote by $AC]a,b[$ the space of {\em absolutely
  continuous functions} on $]a,b[$, that is, distributions on $]a,b[$
whose derivative is in $L_\loc^1]a,b[$.  This definition is equivalent
to the more common definition: for every $\epsilon>0$ there exists
$\delta>0$ such that for any family of non-intersecting intervals
$[x_i,x_i']$ in $a,b[$ satisfying $\sum_{i=1}^n|x_i'-x_i|<\delta$ we
have $\sum_{i=1}^n|f(x_i')-f(x_i)|<\epsilon$.
  
We have $AC]a,b[\subset C]a,b[$. If $f,g\in AC]a,b[$, then
$fg\in AC]a,b[$ and the Leibniz rule holds:
\begin{equation}  \label{leibnitz1}
  (fg)'= f'g+fg' .
\end{equation}
$AC^n]a,b[$ will denote the space of distributions whose $n$th
derivative is in $AC]a,b[$.

\begin{lemma}
  Let $f_n\in AC]a,b[$ be a sequence such that for any $a<a_1<b_1<b$,
  $f_n\to f$ uniformly on $[a_1,b_1]$ and $f_n'\to g$ in
  $L^1[a_1,b_1]$. Then $f\in AC]a,b[$ and $g=f'$.
\label{prel}\end{lemma}

We will denote by $AC[a,b]$ the space of functions on $[a,b]$ whose
(distributional) derivative is in $L^1]a,b[$. Clearly,
$AC[a,b]\subset C[a,b]$.  If $f\in AC[a,b]$, then
\begin{align}
  \int_{a}^{b} f'(x)\d x&=f(b)-f(a).  \label{leibnitz2}
\end{align}
Note that $a$ can be $-\infty$ and $b$ can be $\infty$.

Obviously, if $f\in AC]a,b[$ and $a<a_1<b_1<b$ then
$f\Big|_{[a_1,b_1]}$ belongs to $AC[a_1,b_1]$.

    \subsection{Choice of functional-analytic setting}
\label{Choice of funtional-analytic setting}

Throughout the section, we assume that $V\in L_\loc^1]a,b[$ and we
consider the differential expression
\begin{equation} \label{setting}
  L:=-\partial^2+V .
\end{equation}
Sometimes we restrict our operator to a smaller interval,
say $]c,d[$, where $a\leq c<d\leq b$. Then 
(\ref{setting}) restricted to $]c,d[$ is denoted $L^{c,d}$.

Eventually, we would like to study operators in $L^2]a,b[$ associated
to $L$, which in many respects seems the most natural setting for one
dimensional Schr\"odinger operators. This introductory section,
however, is devoted mostly to the {\em equation} $Lf=g$. We postpone
considerations related to {\em operator realizations of $L$} to
the later sections.

Suppose that we choose $L_\loc^1]a,b[$ as the target space for
(\ref{setting}), which seems to be a rather general function space.
Note that if $f\in L_\loc^\infty]a,b[$, the product $fV$ is well
defined and belongs to $L_\loc^1]a,b[$. Moreover, this is the best we
can do if $V$ is an arbitrary locally integrable function, i.e.\ we
cannot replace $L_\loc^\infty$ by a larger space.  Then, if we
consider $L_\loc^\infty$ as the initial space for (\ref{setting}) and
we require that the target space for (\ref{setting}) is
$L_\loc^1]a,b[$, we are forced to work with functions
$f\in L_\loc^\infty]a,b[$ such that $Lf \in L_\loc^1]a,b[$. But then
$f''\in L_\loc^1]a,b[$, and hence $f\in AC^1]a,b[$.
  
Therefore it is natural to consider (\ref{setting}) as an an operator
$L:AC^1]a,b[\,\to L_\loc^1]a,b[$, which we will do throughout this
paper.  Restrictions of $L$ to subspaces of $AC^1]a,b[$ which are sent
into $L^2]a,b[$ by $L$ are the objects of main interest in our study.

We equip $L^1_\loc]a,b[$ with the topology of local uniform
convergene, i.e.\ a sequence $\{f_n\}$ converges to $f$ if and only if
$\lim\limits_{n\to\infty}\|f_n-f\|_{L^1(J)}=0$ for any compact
$J\subset]a,b[$. Clearly this is a complete space. It is convenient to
think of $L$ as an operator in $L^1_\loc]a,b[$ with domain
$AC^1]a,b[$.  Then $L$ is densely defined and later on we will prove
that it is is closed (see Corollary \ref{cor1}). \label{p:locunif}

\subsection{The Cauchy problem}

For $ g\in L_\loc^1]a,b[$ we consider the problem
\begin{equation}   \label{mojtaba1}
  Lf=g.
\end{equation}

\begin{proposition} \label{mojo}
  Let $a<d<b$. Then for any $p_0$, $p_1$ there exists a unique
  $f\in AC^1]a,b[$ satisfying (\ref{mojtaba1}) and
\begin{equation}
  f(d)=p_0,\quad f'(d)=p_1.
\end{equation}
\end{proposition}
    
\proof Define the operators $Q_d$ and $T_d$ by their integral kernels
\begin{align}\label{Qd}
Q_d(x,y):=\begin{cases}
 (y-x)V(y) ,&x<y<d,\\
(x-y)V(y),&x>y>d,\\
0&\text{ otherwise; }
\end{cases}\\\label{Td}
T_d(x,y):=\begin{cases}
 (y-x) ,&x<y<d,\\
(x-y),&x>y>d,\\
0&\text{ otherwise.}
\end{cases}
\end{align}
The Cauchy problem can be rewritten as $F(f)=f$ where $F$ is a map on
$C]a,b[$ given by
\begin{equation}
  F(f)(x):=p_0+p_1(x-d)+Q_df(x) +T_dg(x).
\end{equation}
If $a\leq a_1<d<b_1\leq b$ and we view $Q_d$ as an operator on
$C[a_1,b_1]$ with the supremum norm, then
\begin{equation}
  \|Q_d\|\leq \max\Big\{  \int_{a_1}^d|V(y)(y-a_1)|\d y,\,
  \int_d^{b_1}|V(y)(y-b_1)|\d y\Big\}.\label{preq}
\end{equation}
If the interval $[a_1,b_1]$ is finite, the operator $T_d$ is bounded
from $L^1]a_1,b_1[$ into $C[a_1,b_1]$.

Thus, by choosing a sufficiently small interval $[a_1,b_1]$ containing
$d$, we can make $F$ well defined and contractive on $C[a_1,b_1]$.
($F$ is contractive iff $\|Q_d\|<1$).  By Banach's Fixed Point Theorem
(or the convergence of an appropriate Neumann series) there exists
$f\in C[a_1,b_1]$ such that $f=F(f)$. Then note that we have
\[
  f'(x)=F(f)'(x) =p_1+\int_d^xV(y)f(y)\d y+\int_d^xg(y)\d y
\]
hence $f'\in AC[a_1,b_1]$ and $f\in AC^1[a_1,b_1]$.

Thus for every $d\in]a,b[$ we can find an open interval containing $d$ on
which there exists a unique solution to the Cauchy problem.  We can
cover $]a,b[$ with intervals $]a_j,b_j[$ containing $d_j$
with the analogous property.  This allows us to extend the solution
with initial conditions at any $d\in]a,b[$ to the whole $]a,b[$.  \qed

\subsection{Regular and semiregular endpoints}

One dimensional Schr\"odinger operators possess the simplest theory
when $-\infty<a<b<\infty$ and $V\in L^1]a,b[$. Then we say that
\emph{$L$ is a regular operator}. Most of the classical
Sturm-Liouville theory is devoted to such operators. More generally,
the following standard terminology will be convenient.

\begin{definition}\label{df:reg}
  The endpoint $a$ is called \emph{regular}, or $L$ is called
  \emph{regular at $a$}, if $a$ is finite and $V\in L_\loc^1[a,b[$
  (i.e. $V$ is integrable around $a$).  Similarly for $b$. Hence $L$
  is \emph{regular} if both endpoints are regular.
\end{definition}

1d Schr\"odinger 
operators satisfying the following conditions are also relatively well
behaved:

\begin{definition}\label{df:semireg}
  The endpoint $a$ is called \emph{semiregular} if $a$ is finite and
  $(x{-}a)V\in L_\loc^1[a,b[$ (i.e. $(x{-}a)V$ is integrable around
  $a$).  Similarly for $b$.
\end{definition}

\begin{proposition}\label{prop:pro}
  Let $g\in L_\loc^1[a,b[$.
  \begin{enumerate}
  \item Let $a$ be a regular endpoint.  Let $p_0$, $p_1$ be
    given. Then there exists a unique $f\in AC^1[a,b[$ satisfying
    $Lf=g$ and $f(a)=p_0, f'(a)=p_1$.
  \item Let $a$ be a semiregular endpoint. Then all solutions $f$ have
    a limit at $a$.
  \item Let $a$ be a semiregular endpoint.  Let $p_1$ be given. Then
    there exists a unique $f\in AC^1[a,b[$ satisfying $Lf=g$ and
    $f(a)=0, f'(a)=p_1$.
\end{enumerate}
\end{proposition}

\begin{proof} (1) is proven as Prop.\ \ref{mojo}, choosing $d=a$; the
  operators \eqref{Qd} and \eqref{Td} are now
  \begin{align}\label{Qa}
Q_a(x,y):=\begin{cases}
(x-y)V(y),&x>y>a,\\
0&\text{ otherwise; }
\end{cases}\\\label{Ta}
T_a(x,y):=\begin{cases}
(x-y),&x>y>a,\\
0&\text{ otherwise.}
\end{cases}
\end{align}

To prove (2) we choose $d$ inside $]a,b[$ such that
$\int_a^dV(y)|y-a|\d y<1$. This guarantees that the operator $Q_d$ is
contractive on $C[a,d]$.

To prove (3) we modify the proof of Prop.\ \ref{mojo}. We chose $d=a$
and use the Banach space
$|x-a|^{-1}C]a,b_1[:= \{f\in C]a,b_1[ \ \mid \|f\|:=
\sup\frac{|f(x)|}{|x-a|}<\infty\}$, where $a<b_1<b$ is finite and will
be chosen later.  The Cauchy problem can be rewritten as $F(f)=f$
where $F$ is given by ${|x-a|^{-1}C]a,b_1[}$ given by
\begin{equation}
  F(f)(x):=p_1(x-a)+Q_af(x) +T_ag(x).
\end{equation}
$Q_a$ is an operator on $|x-a|^{-1}C[a,b_1]$ with the norm
\begin{equation}
  \|Q_a\|\leq  \int_{a}^{b_1}|V(y)(y-a)|\d y.\label{preqa}
\end{equation}
The operator $T_a$ is bounded from $L^1]a,b_1[$ into
$|x-a|^{-1}C[a,b_1]$.  Therefore, $F$ is a well-defined map on
$|x-a|^{-1}C[a,b_1]$.  Then we argue similarly as in the proof of
Prop. \ref{mojo}.  For $b_1$ close enough to $a$ the map $F$ is
contractive and we can apply Banach's Fixed Point Theorem. From
\[
  f'(x)=F(f)'(x) =p_1+\int_a^x|y-a|V(y)|y-a|^{-1}f(y)\d
  y+\int_a^xg(y)\d y
\]
we see that $f'(a)=p_1$.
\end{proof}

  An example of a potential with a finite point which is not
  semiregular is $V(x)=\frac{c}{x^2}$ on $]0,\infty[$. For its theory
  see \cite{BDG,DR}.
      
\subsection{Wronskian}

\begin{definition} The {\em Wronskian} of two differentiable functions
  $u,v$ is
\begin{equation}
  W(u,v;x)= W_x(u,v)=u(x)v'(x)-u'(x)v(x).
\end{equation}
\end{definition}

\begin{proposition}
  Let $u,v\in AC^1]a,b[$.
  Then the {\em Lagrange identity} holds:
  \begin{equation} \label{mojtaba}
    \partial_x W(u,v;x)=-(Lu)(x) v(x)+u(x) (Lv)(x).
\end{equation}
  Consequently, if $Lu=Lv=0$, then $W(u,v)$ is a constant function. 
\end{proposition}

\proof Since $u,u',v,v'\in AC]a,b[\,$, the Wronskian can be
differentiated and a simple computation yields (\ref{mojtaba}). \qed

\begin{definition}
  The set of solutions in $AC^1]a,b[$ of the homogeneous equation
  $Lf=0$ will be denoted $\Ker( L)$.
\end{definition}
$\Ker(L)$  is a two dimensional complex space and the map
$W:\Ker( L)\times\Ker( L)\to\bbC$ is bilinear and antisymmetric.
$u,v\in\Ker( L)$ are linearly independent if and only if
$W(u,v)\neq0$.  If
$u_2=\alpha u_1+\beta v_1, v_2=\gamma u_1+\delta v_1$ then
\[
W(u_2,v_2)=W(\alpha u_1,\delta v_1)+W(\beta v_1,\gamma u_1)
= (\alpha\delta-\beta\gamma) W(u_1,v_1).
\]
hence if $W(u_1,v_1)=1$ then $W(u_2,v_2)=1$ if and only if
$\alpha\delta-\beta\gamma=1$, and in this case a simple computation gives
\begin{equation} \label{comment}
  u_2(x)v_2(y)-u_2(y)v_2(x)=u_1(x) v_1(y)-u_1(y) v_1(x),\quad x,y\in]a,b[.
\end{equation}
Thus the function
\begin{equation}
  G_\leftrightarrow(x,y):=u(x) v(y)-u(y) v(x)\label{canon}\end{equation}
is independent of
the choice of the solutions $u,v$ of the homogeneous equation $Lf=0$
if they satisfy $W(u,v)=1$.
(\ref{canon}) can be interpreted as the integral kernel of an operator 
$G_\leftrightarrow: L_{\mathrm{c}}^1]a,b[\to AC^1]a,b[$,  and will be called the {\em canonical bisolution of $L$}.
       It satisfies
       \begin{equation} LG_\leftrightarrow=0,\quad
G_\leftrightarrow L=0,\quad G_\leftrightarrow(x,y)=
- G_\leftrightarrow(y,x).\end{equation}
         
\subsection{Green's operators}

The expression ``Green's function'' is commonly used to denote the
integral kernel of a right inverse of a differential operator, usually
of a second order.  We will use the expression ``Green's operator''
for a right inverse of $L$.

\begin{definition} \label{green1} An operator
  $G_\bullet: L_{\mathrm{c}}^1]a,b[\to AC^1]a,b[$ is called a {\em
    Green's operator of $L$} if
\begin{equation}
  LG_\bullet g=g,\quad g\in  L_{\mathrm{c}}^1]a,b[. \label{green11}
\end{equation} 
\end{definition}

Note that we do not require that $ G_\bullet L=\one$. Note also that
$G_\leftrightarrow$ is not Green's operator---it is a
bisolution. However, it is so closely related to various Green's
operators that its symbol contains the same letter $G$.

There are many Green's operators. If $G_\bullet$ is a Green's
operator, $ u,v$ are two solutions of the homogeneous equation, and
$\phi,\psi\in L_\loc^\infty]a,b[$ are arbitrary, then
\[
  G_\bullet+|u\rangle\langle \phi|+|v\rangle\langle \psi|
\]
is also a Green's operator.  Recall that if $E,F$ are vector spaces,
$g$ belongs to the dual of $E$, and $f\in F$, then $\ketbra{f}{g}$ is
the linear map $E\to F$ defined by $e\mapsto g(e) f$.

Let us define some distinguished Green's operators.  Let $u ,v$ be two
solutions of the homogeneous equation such that
$W(v ,u)=1$.
We easily check that the operators
$G_{u,v}$, $G_a$ and $G_b$ defined below are Green's operators in the
sense of Def. \ref{green1}:

\begin{definition}
  {\em Green's operator associated to $u$ at $a$ and $v$ at $b$},
  denoted $G_{u,v}$, is defined by its integral kernel
  \[G_{u,v}(x,y):=\begin{cases}u(x)v (y),&x<y,\\
    v (x)u(y),&x>y.\end{cases}\]
  \label{two-side} \end{definition}
 Operators of the form $G_{u,v}$ will be sometimes called {\em
   two-sided Green's operators}.
\begin{definition}  {\em Forward Green's operator} $G_\rightarrow $
  has the integral kernel
  \begin{equation}G_\rightarrow (x,y):=\begin{cases}0,&x<y,\\
  v (x)u(y)-u(x)v (y),&x>y.\end{cases}\label{qeq0}\end{equation}
 \end{definition}
 \begin{definition}  \label{df:Ga}
   {\em Backward Green's operator} $G_\leftarrow $
  has the integral kernel
  \[G_\leftarrow (x,y):=\begin{cases}
  u (x)v(y)-v(x)u (y),&x<y,\\
  0,&x>y.\end{cases}\]
 \end{definition}
 By the comment after \eqref{comment}, the operators $G_\rightarrow $ and $G_\leftarrow $
 are independent of the choice of $u,v$.  For $a<b_1<b$, $G_\rightarrow $ maps
 $L_{\mathrm{c}}^1]b_1,b[$ into functions that are zero on $]a,b_1]$.
 Similarly, for $a<a_1<b$, $G_\rightarrow $ maps $L_{\mathrm{c}}^1]a,a_1[$ into
 functions that are zero on $[a_1,b[$.

 Note also some formulas for differences of two kinds of Green's
 operators:
 \begin{align}
   G_{u,v}-G_\rightarrow &=|u\rangle\langle v|, \label{moja-7}
\\ 
   G_{u,v}-G_\leftarrow &=|v\rangle\langle u| ,      \label{mojta}
\\
G_\rightarrow -G_\leftarrow &=|v\rangle\langle u|-|u\rangle\langle v|=
G_\leftrightarrow,\label{mojj}\\
   G_{u,v}-G_{u_1,v_1}&=|u\rangle\langle v|-|u_1\rangle\langle v_1|,
 \end{align}

 The following definition introduces another class of Green's
 operators in the sense of Def. \ref{green1}, which are
 generalizations of forward and backward Green's operators.

 \begin{definition}  \label{greend}
 {\em Green's operator associated to $d\in]a,b[$} is defined by the
 integral kernel       
  \[G_d(x,y):=\begin{cases}
   u (x)v(y)-v(x)u (y),&x<y<d,\\
   v (x)u(y)-u(x)v(y),&x>y>d,\\
   0,&\text{ otherwise.}
  \end{cases}\]
\end{definition}

As in the case of $G_\rightarrow $ and $G_\leftarrow $, these operators are independent of
the choice of $u,v$.  Note that if $a<a_1<d<b_1<b$, then $G_d$ maps
$L_{\mathrm{c}}^1]a,a_1[$ on functions that are zero on $[a_1,b[$, and
$L_{\mathrm{c}}^1] b_1,b[$ on functions that are zero on $]a,b_1]$.

 The proof of Prop.  \ref{mojo} suggests how to construct $G_d$
 without knowing $v,u$ using the operators $Q_d$ and $T_d$ defined
 there. We have, at least formally,
\begin{equation}
  G_d=(\one-Q_d)^{-1}T_d.
\label{preq1}  \end{equation}
If we choose $a\leq a_1<d<b_1\leq b$ with a finite $[a,b]$ and
\begin{equation}
  \max\Big\{  \int_{a_1}^d|V(x)(x-a_1)|\d x,\,
  \int_d^{b_1}|V(x)(x-b_1)|\d x\Big\}<1,\label{preq2}
\end{equation}
then
(\ref{preq1}) is given by a convergent Neumann series in the
sense of an operator from $L^1]a_1,b_1[$ to $C[a_1,b_1]$.

\begin{remark}
  The 1-dimensional Schr\"odinger equation can be interpreted as the
  Klein Gordon equation on a $1+0$ dimensional spacetime (no spacial
  dimensions, only time). The operators $G_\leftrightarrow$,
  $G_\rightarrow$ and $G_\leftarrow$ have important generalizations to
  globally hyperbolic spacetimes of any dimension---they are then
  usually called the Pauli-Jordan, retarded, resp. advanced
  propagator, see e.g. \cite{DS}.
\end{remark}

\subsection{Some estimates}

The following elementary estimates will be useful later on.

\begin{lemma}\label{lm:ode-est}
Let $J$ be an open interval  of length
$\nu<\infty$ and $f\in L^1(J)$ with
 $f''\in L^1(J)$. Then $f$ and $f'$ are
continuous functions on the closure of $J$ and if $a$ is an end
point of $J$ then
\begin{equation}\label{eq:ode-est}
|f(a)|+\nu |f'(a)|\leq C\int_J \left(\nu|f''(x)| + \nu^{-1} |f(x)| \right) \d x,
\end{equation}
where $C$ is a real number independent of $f$ and $J$.
\end{lemma}

\begin{proof} By a scaling argument we may assume $\nu=1$. It suffices to
assume that $f$ is a distribution on $\R$ such that $f''=0$ outside
$J$.  Let $\theta:\R\to\R$ be of class $C^\infty$ outside of $0$ and
such that $\theta(x)=0$ if $x\leq0$, $\theta(x)=x$ if $0<x<1/2$,
$\theta(x)=0$ if $x\geq1$. Define $\eta$ by $\theta''=\delta-\eta$
where $\delta$ is the Dirac measure at the origin. Clearly $\eta$ is
of class $C^\infty$ with support in $[1/2,1]$. For any distribution
$f$ we have
\[
f=\delta *f=\theta'' * f + \eta * f= \theta * f'' + \eta * f \
\text{ hence also } \
f'= \theta' * f'' + \eta' * f.
\]
This clearly implies \eqref{eq:ode-est} for $\nu=1$ and $a$ the right
endpoint of $J$.  \end{proof}

\begin{lemma}\label{lm:ode-v}
  Assume that $\ell:=\sup_J\int_J|V(x)|\d x<\infty$ where $J$ runs
  over all the intervals $J\subset ]a,b[$ of length $\leq1$.  Then there
  are numbers $C,\nu_0>0$ such that
\begin{equation}\label{eq:ode-main}
\|f\|_{L^{\infty}(J)} + \nu \|f'\|_{L^{\infty}(J)} \leq
C\nu\|Lf\|_{L^1(J)} + C\nu^{-1}\|f\|_{L^1(J)}
\end{equation}
for all $f\in L_\loc^\infty]a,b[$, all $\nu\leq\nu_0$, and all intervals
$J\subset ]a,b[$ of length $\nu$.
\end{lemma}

\begin{proof} Note that for a continuous $f$ we have
$f''\in L^1_{\text{loc}}$ if and only $Lf\in L^1_{\text{loc}}$ and
then $f'$ is absolutely continuous. We take $\nu_0\leq1$ and strictly
less than half the length of $]a,b[$. If $\nu\leq\nu_0$, then
\eqref{eq:ode-est} gives for $x$ such that $]x,x+\nu[\subset ]a,b[$:
\begin{eqnarray*}
|f(x)|+\nu |f'(x)| &\leq& C\int_x^{x+\nu}
\left(\nu|Lf| + \nu|Vf| + \nu^{-1} |f(y)| \right) \d y  \\
&\leq&
C \|\nu|Lf| + \nu^{-1} |f| \|_{L^1(x,x+\nu)}  + 
C \ell\nu \|f\|_{L^\infty(x,x+\nu)} \\
&\leq&
C \nu \|Lf\|_{L^1(J)} + C \nu^{-1} \|f \|_{L^1(J)}  + 
C \ell\nu \|f\|_{L^\infty(J)}.
\end{eqnarray*}
If $x\in ]a,b[$ and $]x,x+\nu[\,\not\subset\, ]a,b[$ then
$]x-\nu,x[\,\subset\, ]a,b[$ and we have an estimate as above with
$]x,x+\nu[$ replaced by $]x-\nu,x[$. Hence
\[
\|f\|_{L^\infty(J)} + \nu \|f'\|_{L^\infty(J)} \leq
C \nu \|Lf\|_{L^1(J)} + C \nu^{-1} \|f \|_{L^1(J)}  + 
C \ell\nu \|f\|_{L^\infty(J)}.
\]
If $\nu_0$ is such that $C\ell\nu_0<1$ we get the required estimate.
\end{proof}

Recall (see Subsect. \ref{Choice of funtional-analytic setting}) that
$L^1_\loc]a,b[$ is equipped with the topology of local uniform
convergence and that we think of $L$ as an operator in $L^1_\loc]a,b[$
with domain $AC^1]a,b[$.  The next result says that this operator is
closed.

\begin{corollary} \label{cor1} Let $\{f_n\}$ be a sequence in
  $AC^1]a,b[$ such that the sequences $\{f_n\}$ and $\{Lf_n\}$ are
  Cauchy in $L^1_\loc]a,b[$.  Then the limits
  $f:=\lim\limits_{n\to\infty}f_n$ and
  $g:=\lim\limits_{n\to\infty}Lf_n$ exist in $ L_\loc^1]a,b[$ and we
  have $f\in AC^1]a,b[$ and $Lf=g$.
\end{corollary}

\proof The estimate (\ref{eq:ode-main}) implies that on every compact
interval $J$ we have uniform convergence of $f_n$ to $f$ (and also of
$f_n'$ to $f'$).  Therefore, $Vf_n\big|_J\to Vf\big|_J$ in $L^1(J)$
for any such $J$. Hence, $-f_n''=Lf_n-Vf_n$ converges in $L^1(J)$ to
$g-Vf$. Therefore, by Lemma \ref{prel}, $-f''=g-Vf$.  We know that
$g-Vf\in L_\loc^1]a,b[$, hence $f\in AC^1]a,b[$. \qed

\section{Hilbert spaces with conjugation} \label{s:abs}
\protect\setcounter{equation}{0}

\subsection{Bilinear scalar product} \label{ss:bil}

Let $\cH$ be a Hilbert space equipped with a {\em scalar product}
$(\cdot|\cdot)$.  One says that $J$ is a {\em conjugation} if it is an
antilinear operator on $\cH$ satisfying $J^2=\one$ and
$(Jf|Jg)=\bar{(f|g)}$. In a Hilbert space with a conjugation $J$ an
important role is played by the {\em natural bilinear form}
\begin{equation}
  \langle f|g\rangle:=(Jf|g).
\end{equation}
In our paper we usually consider the Hilbert space $\cH=L^2]a,b[$,
which has the obvious conjugation $f\mapsto \bar f$.  Its scalar
product and its bilinear form are as follows
\begin{align}\label{sesqui}
  ( f|g)&:=\int_a^b \overline{f(x)}g(x)\d x,\\
\label{bili}
\langle f|g\rangle&:=\int_a^b f(x)g(x)\d x=(\bar f| g).
\end{align}
Thus we use round brackets for the sesquilinear scalar product and
angular brackets for the bilinear form. Note that in some sense the
latter plays a more important role in our paper (and in similar
problems) than the former. See e.g. \cite{D,DR2}, where the same
notation is used.

If $\cG\subset L^2]a,b[$, we will write
\begin{align}
  \cG^\perp&:=\{f\in L^2]a,b[\ |\ ( f|g)=0,\ g\in\cG\},\\
  \cG^\per&:=\{f\in L^2]a,b[
            \ |\ \langle f|g\rangle=0,\ g\in\cG\}
            =\overline{\cG^\perp}. \label{eq:per}
\end{align}

\subsection{Transposition of operators}

If $T$ is an operator, then $\cD(T)$, $\Ker(T)$ and $\Ran(T)$ will
denote the domain, the nullspace (kernel) and the range of $T$.
$\bar T$ denotes the {\em complex conjugation of $T$}. This means,
\begin{equation} \label{eq:conjug}
 \cD(\bar{T}):=\{\bar {f}\mid f\in\cD(T)\},\quad  \overline{ T}f:=\overline{T\overline{ f}},\quad f\in\cD(\bar{T}).
\end{equation}
Suppose that $T$ is densely defined. We say that $u\in\cD(T^\#)$ if
\begin{equation}
  \langle u|Tv\rangle=  \langle w|v\rangle, \quad u\in\cD(T),
\end{equation}
for some $w\in\cH$. Then we set $T^\#u:=w$. The operator $T^\#$ is
called the {\em transpose of $T$}. Clearly, if $T^*$ denotes the usual
{\em Hermitian adjoint} of $T$, then $T^\#=\overline{T^*}=\bar{T}^*$.

If $T$ is a bounded linear operator  with
\begin{equation*}
\big(T f\big)(x):=\int_a^b T(x,y)f(y)\d y,
\end{equation*}
then
\begin{align}
\big(T^* f\big)(x)&=\int_a^b\overline{ T (y,x)}f(y)\d y,\\
\big(T^\# f\big)(x)&=\int_a^b T (y,x)f(y)\d y,\label{eq:trans}\\
\big(\overline{ T} f\big)(x)&=\int_a^b \overline{T (x,y)}f(y)\d y.
\end{align}
An operator $T$ is {\em self-adjoint} if $T=T^*$. We will say that it
is {\em self-transposed} if $ T^\#=T.$ It is useful to note that a
holomorphic function of a self-transposed operator is self-transposed.

\begin{remark}

  {\rm It would be natural to call an operator satisfying
    $T\subset T^\#$ {\em symmetric}. The natural name for an operator
    satisfying $T\subset T^*$ would then be {\em
      Hermitian}. Unfortunately and confusingly, in a large part of
    mathematical literature the word {\em symmetric} is reserved for
    an operator satisfying the latter condition.}
  
\end{remark}

Many properties of the transposition have their exact analogs for the
Hermitian conjugation and are proven similarly.  Below we describe
some of them.

\begin{proposition}
  Let $T$ be a densely defined operator.
  \begin{enumerate}
  \item $T^\#$ is closed.
  \item If $T$ is closed, then $T=T^{\#\#}$.
  \item Let $T\subset S$. Then $S^\#\subset T^\#$ and
  \begin{equation}
    \dim\cD(S)/\cD(T)=\dim\cD(T^\#)/\cD(S^\#).
  \end{equation}
  \end{enumerate}
\end{proposition}

\begin{lemma}\label{lm:abst}
  Let $S,T$ be linear operators on a Hilbert space $\mathcal{H}$ such
  that:\\[1mm]
  {\rm(1)} $\langle Sf|g\rangle=\langle f|Tg\rangle $ for all
  $f\in \mathcal{D}(S)$ and $g\in \mathcal{D}(T)$,\\
  {\rm(2)} $T$ is surjective,\\
  {\rm(3)} $\Ran(S)^\per\subset \Ker( T)$.\\[1mm]
  Then $S$ is densely defined.
\end{lemma}

\begin{proof}
  We must show:
  \[
    \langle f|h\rangle=0,\ \forall f\in \mathcal{D}(S)\Rightarrow h=0.
  \]
  Since $T$ is surjective, there is $g\in\mathcal{D}(T)$ such that
  $h=Tg$ and then we get
  $0=\langle f|h\rangle=\langle f|Tg\rangle =\langle Sf|g\rangle$ by
  (1) for all $f\in \mathcal{D}(S)$. Thus $g\in \Ran(S)^\per$ and (3)
  gives $Tg=0$ hence $h=0$.
\end{proof}

Here is a version of the closed range theorem \cite[Sect.\ VII.5]{Y},
which we will use in \S\ref{ss:L2G}.

\begin{theorem}\label{th:crt}
  If $T$ is a closed densely defined operator in $\cH$, then the
  following assertions are equivalent:\\
  {\rm(1)} $\Ran(T)$ is closed, \quad
  {\rm(2)} $\Ran(T^\#)$ is closed, \\
  {\rm(3)} $\Ran (T)=\Ker (T^\#)^\per$,  \quad
  {\rm(4)}   $\Ran( T^\#)=\Ker( T)^\per$.
\end{theorem}

\begin{remark}

  {\rm In a large part of literature \cite{EE,Kn} a different
    terminology and notation is used.  If $T$ is an operator, then
    $JT^*J$ is called the {\em $J$-adjoint of $T$}; an operator $T$
    satisfying $T=JT^*J$ is called {\em $J$-self-adjoint}, etc.  In
    the context of our paper $Jf=\bar f$. Moreover, $(Jf|g)$ and
    $JT^*J$ are denoted $\langle f|g\rangle$, resp. $T^\#$. Our
    notation and terminology stresses the naturalness of the bilinear
    product $\langle\cdot|\cdot\rangle$ and of the transposition
    $T\mapsto T^\#$.  Therefore, we prefer them instead of the
    notation and terminology of e.g. \cite{EE,Kn}, which puts $J$ in
    many places.}

\end{remark}

\subsection{Spectrum}

Let $T$ be a closed operator.  We say that $G$ is an {\em inverse} of
$T$ if $G$ is bounded, $GT=\one$ on $\cD(T)$, and $TG=\one$ on $\cH$.
An inverse of $T$, if it exists, is unique. If $T$ possesses an
inverse, we say that it is {\em invertible}.

We will denote by $\rs(T)$ the {\em resolvent set},
that is the set of $\lambda\in\C$ such that $T-\lambda$ is invertible.
The spectrum of $T$ is $\spec(T):=\C\backslash\rs(T)$.

Let $T$ be densely defined.  Then $T$ is invertible iff $T^\#$ is,
$T^{-1\#}=T^{\#-1}$ and $\spec(T)=\spec(T^\#)$.

We say that a closed operator $T$ is {\em Fredholm} if
$\dim\Ker(T)<\infty$ and $\dim\cH/\Ran(T)<\infty$.  If $T$ is a
Fredholm operator, we define its {\em index}
\begin{equation}
  \ind(T)=\dim\Ker(T)-\dim\cH/\Ran(T).\label{fred}
\end{equation}
If a closed operator $T$ is densely defined, then we have an
equivalent, more symmetric definition: $T$ is Fredholm if $\Ran(T)$ is
closed (equivalently, $\Ran(T^\#)$ is closed), $\dim\Ker(T)<\infty$
and $\dim\Ker(T^\#)<\infty$.  Besides,
\begin{equation}
  \ind(T)=\dim\Ker(T)-\dim\Ker(T^\#).\label{fred1a}\end{equation}
One can introduce various varieties of the {\em essential
  resolvent set} and {\em essential spectrum}, \cite{EE}.  One of them is
\begin{equation}
  \rs_{\mathrm{F}0}(T):=\{\lambda\in\C\mid T-\lambda\quad\text{is Fredholm of index $0$}\},
  \quad \spec_{\mathrm{F}0}(T):=\C\backslash\rs_{\mathrm{F}0}(T).\label{fred2}
  \end{equation}
  Clearly, $\rs(T)\subset\rs_{\mathrm{F}0}(T)$.

\subsection{Restrictions of closed operators}
\label{ss:L2G}

In this subsection we fix a closed operator on a Hilbert space
$\cH$. For consistency with the rest of the paper, this operator will
be denoted by $L_\M$.  Note that $\cD(L_\M)$ can be treated as a
Hilbert space with the scalar product
\begin{equation}  \label{scala}
  (f|g)_L:=(f|g)+(L_\M f|L_\M g).
\end{equation}
We will investigate closed operators $L_\bullet$ contained in $L_\M$.
Obviously, $L_\bullet\subset L_\M$ if and only if
$L_\bullet-\lambda\subset L_\M-\lambda$, where $\lambda$ is a
  complex number. Therefore, many of the statements in this and next
subsections have obvious generalizations, where $L_\M$ is replaced
with $L_\M-\lambda$. For simplicity of presentation, we keep
$\lambda=0$.

\begin{proposition} $ $
  \begin{enumerate}
  \item We have a 1-1 correspondence between closed subspaces
    $\cL_\bullet$ of $\cD(L_\M)$ and closed operators
    $L_\bullet\subset L_\M$ given by $\cL_\bullet=\cD(L_\bullet)$.
  \item If $L_\M$ is Fredholm and $L_\bullet\subset L_\M$ is
      closed, then $L_\bullet$ is Fredholm if and only if
    $\dim\cD(L_\M)/\cD(L_\bullet)<\infty$, and then
    \begin{equation} \label{eq:dimrel}
      \dim\cD(L_\M)/\cD(L_\bullet)=\ind(L_\M)-\ind(L_\bullet)
    \end{equation}
  \item If $\Ran(L_\M)=\cH$,
    and $L_\bullet\subset L_\M$, then $\ind(L_\bullet)=0$ iff
    \begin{equation}
      \dim\cD(L_\M)/\cD(L_\bullet)=\dim\Ker(L_\M).
    \end{equation}
\end{enumerate}\end{proposition}
    
\begin{proof}
  (1) is obvious.

  Since $L_\bullet$ is a restriction of $L_\M$, we
  have $\Ker(L_\bullet)\subset \Ker(L_\M)$ and, $\Ker(L_\bullet)$
  being a finite dimensional subspace of the Banach space
  $\cD(L_\bullet)$, there is a closed subspace of $\cD_\bullet$ of
  $\cD(L_\bullet)$ such that
  $\cD(L_\bullet)=\Ker(L_\bullet) \oplus \cD_\bullet$. We have
  $\cD_\bullet+\Ker(L_\M)\subset\cD(L_\M)$ with
  $\cD_\bullet\cap\Ker(L_\M)=0$ hence, since $\Ker(L_\M)$ is finite
  dimensional, there is a closed subspace $\cD_\M$ of $\cD(L_\M)$
  which contains $\cD_\bullet$ and such that
  $\cD(L_\M)=\Ker(L_\M) \oplus \cD_\M$.  Then we clearly get
  \begin{align*}
    \dim\cD(L_\M)/\cD(L_\bullet) &=\dim \Ker(L_\M)/\Ker(L_\bullet)
    + \dim \cD_\M/\cD_\bullet \\
    &=\dim \Ker(L_\M)-\dim\Ker(L_\bullet) + \dim \cD_\M/\cD_\bullet .
  \end{align*}
  The map $L_\M:\cD_\M\to\Ran(L_\M)$ is bijective and its restriction
  to $\cD_\bullet$ has $\Ran(L_\bullet)$ as image hence by the closed
  map theorem it is a homeomorphism, so
  \[
    \dim\cD_\M/\cD_\bullet=\dim\Ran(L_\M)/\Ran(L_\bullet)
    =\dim\cH/\Ran(L_\bullet)-\dim\cH/\Ran(L_\M).
\]
The last two relations imply \eqref{eq:dimrel}.

Finally, (3) is an immediate consequence of (2).
\end{proof}

The following concept is useful in the study of invertible operators
contained in $L_\M$:

\begin{definition}\label{rightinverse}
  We say that $G_\bullet$ is a {\em right inverse of $L_\M$} if
  it is a bounded operator on $\cH$ such that
  $\Ran (G_\bullet)\subset \mathcal{D}(L_\M)$ and
  \begin{equation}\label{eq:311}
    L_\M G_\bullet=\one.
  \end{equation}
\end{definition}

Note that $L_\M$ can have many right inverses or none.

\begin{proposition}\label{pr:lsurj}
  The following conditions are equivalent:\\
  {\rm(1)} $L_\M $ has right inverses,\\
  {\rm(2)} $\Ran(L_\M)=\cH$,\\
If in addition $L_\M$ is densely defined, then (1) and (2) are equivalent to\\
  {\rm(3)} $L_\M ^\#:\cD(L_\M ^\#)\to \cH$ is injective and $\Ran(L_\M^\#)$ is closed. \\
  {\rm(4)} $\Ker (L_\M ^\#)=\{0\}$ and $\Ran (L_\M ^\#)=\Ker( L_\M )^\per$.\\
Under these conditions  we have
  a  bijective correspondence  
  between\\
  {\rm(a)} right inverses $G_\bullet$ of $L_\M$;\\  
{\rm(b)} invertible operators $L_\bullet$ contained in $L_\M$;\\
{\rm(c)} closed subspaces
  $\mathcal{L}_\bullet$ of $\cD(L_\M )$ such that
$\cD(L_\M )=\mathcal{L}_\bullet\oplus \Ker( L_\M) $.\\
This correspondence is  given by $G_\bullet=L_\bullet^{-1}$ and $\cD(L_\bullet)=\cL_\bullet$.
\end{proposition}

\begin{proof}
  If $G_\bullet$ is a right inverse for $L_\M $ then $L_\M $ is
  surjective due to \eqref{eq:311}. Reciprocally, assume $L_\M $ is
  surjective.  Let $\cL_\bullet$ be the orthogonal complement of the
  closed subspace $\Ker( L_\M) $ in the Hilbert space
  $\cD(L_\M)$. Then $\cD(L_\M )=\cL_\bullet\oplus \Ker( L_\M )$. Now
  $L_\bullet= L_\M \big|_{\cL_\bullet}$ is a bijective map
  $\cL_\bullet\to \cH$ and $G_\bullet:= L_\bullet^{-1}$ is a right
  inverse of $L_\M $.  This proves $(1)\Leftrightarrow(2)$.  The
  equivalence with (3) and (4) follows by the closed range theorem
  (see Theorem \ref{th:crt}).
\end{proof}

\begin{proposition}\label{piuy5}
  Let $G_\bullet$ be a right inverse of $L_\M $.  Then
  \begin{enumerate} \item $\Ker(  G_\bullet)=\{0\}$.
  \item $G_\bullet$ is bounded from $\cH$ to $\cD(L_\M )$.
  \item $P_\bullet:=G_\bullet L_\M $ is a bounded projection in the
    space $\cD(L_\M )$ such that
    \[
      \Ran (P_\bullet) =\Ran( G_\bullet),\quad \Ker( P_\bullet) =\Ker( L_\M ).
    \]
  \item $\Ran (G_\bullet)$ is closed in $\cD(L_\M )$.
  \item $\cD(L_\M )=\Ran (G_\bullet)\oplus\Ker( L_\M )$.
  \end{enumerate}
\end{proposition}

\begin{proof}
(1) is obvious and
\begin{equation}\label{eq:GL}
  \|G_\bullet f\|_L^2=  \|L_\M G_\bullet f\|^2+\|G_\bullet f\|^2\leq
  \big(1+\|G_\bullet\|^2\big)\|f\|^2 
\end{equation}
implies (2). Since $L_\M : \cD(L_\M ) \to \cH$ is
  bounded, $P_\bullet $ is bounded on $\cD(L_\M )$.  Then
\[
  P_\bullet^2=G_\bullet( L_\M  G_\bullet) L_\M =G_\bullet
  L_\M =P_\bullet
\]
hence $P_\bullet$ is a projection.

It is obvious that $\Ran (P_\bullet)\subset \Ran( G_\bullet)$.                
If $g\in \cH$, then
\[
  G_\bullet g=G_\bullet L_\M  G_\bullet g=P_\bullet  G_\bullet g.
\]
Hence $\Ran (G_\bullet)\subset \Ran( P_\bullet)$. This shows that
$\Ran( G_\bullet)= \Ran( P_\bullet)$.
               
It is obvious that $\Ker (L_\M )\subset\Ker( P_\bullet)$. If
$0=P_\bullet f$, then
\[
  0=L_\M P_\bullet f=(L_\M G_\bullet) L_\M f= L_\M f.
\]
Hence $\Ker (P_\bullet)\subset L_\M $. This shows that
$\Ker( P_\bullet)= \Ker (L_\M )$.
  
Thus we have shown (3), which implies immediately (4) and (5).
\end{proof}

\begin{proposition} \label{piuy} Let $L_\bullet$ be a closed operator
  such that $L_\bullet\subset L_\M$ and $L_\bullet$ is invertible.
  Then
  \begin{equation}\label{dimen1}
    \cD(L_\M )=\cD(L_\bullet)\oplus\Ker (L_\M).
  \end{equation}
\end{proposition}

\proof $G_\bullet:=L_\bullet^{-1}$ is a right inverse of $L_\M $. Now
(\ref{dimen1}) is the same as Prop. \ref{piuy5}.(5).  \qed

\begin{proposition} \label{piuy7}

  Let $G_\bullet$ be a right inverse of $L_\M $. If
  $K:\cH\to\Ker (L_\M) $ is a linear continuous map, then
  $G_\bullet+K$ is also a right inverse of $L_\M $.  Conversely, all
  right inverses of $L_\M $ are of this form.
  
  In particular, suppose that $\Ker (L_\M )$ is $n$-dimensional and
  spanned by $u_1,\dots,u_n$. Then if $G_1,G_2$ are two right inverses
  of $L_\M $, there exist $\phi_1,\dots\phi_n\in\cH$ such that
  \begin{equation} \label{piuy7ab}
    G_1-G_2=\sum_{j=1}^n|u_j\rangle\langle \phi_j|
.  \end{equation}

\end{proposition}

\begin{proposition} \label{piuy6} Suppose that $G_\bullet$ is a
  bounded operator on $\cH$ and $\cD\subset\cH$ a dense subspace such
  that $G_\bullet \cD\subset \mathcal{D}(L_\M )$ and
  \begin{equation}
    L_\M  G_\bullet g=g,\quad g\in\cD.
  \end{equation}
  Then $G_\bullet$ is a right inverse of $L_\M $.
\end{proposition}

\proof Let $f\in \cH$ and $(f_n)\subset\cD$ such that
$f_n\underset{n\to\infty}{\to}f$. Then
$G_\bullet f_n\underset{n\to\infty}{\to}G_\bullet f$ and
$L_\M G_\bullet f_n=f_n\underset{n\to\infty}{\to} f$. By the
closedness of $L_\M $, $G_\bullet f\in\cD(L_\M )$ and
$L_\M G_\bullet f=f$.  \qed

\subsection{Nested pairs of operators}
\label{ss:L2GII-}

In this subsection we assume that $L_\m$ and $L_\M$ are two
densely defined closed operators. We assume that
they form a {\em nested pair},
\begin{equation}
  L_\m\subset L_\M .\label{nest1}
\end{equation}
Note that along with (\ref{nest1}) we have a second nested pair
\begin{equation}
  L_\M^\#\subset L_\m^\# .\label{nest2}
\end{equation}
In this subsection we do not assume that $L_\m^\#=L_\M$, so that the two
nested pairs can be different.

\begin{remark} The notion of a nested pair is closely related   to the
  notion of    {\em conjugate pair}, often introduced in the
  literature, e.g. in \cite{EE}. Two operators $A,B$ form a  conjugate
  pair if  $A\subset B^*$, and hence $B\subset A^*$. The pair of
  operators $L_\m$,  $ L_\M^*$ is an example of a conjugate pair.  
\end{remark}

\begin{proposition}\label{dire} $ $  
  \begin{enumerate}
  \item We have a direct decomposition
  \begin{equation}
    \cD(L_\M)=\cD(L_\m)\oplus\Ker(L_\m^* L_\M+1),
  \end{equation}
  where
\[ \Ker(L_\m^* L_\M+1)
    =\{u\in\cD(L_\M)\mid L_\M u\in\cD(L_\m^*) 
    \text{ and } L_\m^*L_\M u+u=0\}.
\]    
\item If in addition  $\Ran (L_\m^\#)=\cH$, then
  \begin{equation}
    \cD(L_\M)=\cD(L_\m)\oplus\Ker(L_\m^* L_\M).
  \end{equation}
  where
 $ \Ker(L_\m^* L_\M)
    =\{u\in\cD(L_\M)\mid L_\M u\in\cD(L_\m^*) 
    \text{ and } L_\m^*L_\M u=0\}$. 
\end{enumerate}
\end{proposition}

\begin{proof}

(1)  We will show that
\begin{equation}
  \Ker(L_\m^* L_\M+1)=\cD(L_\m)^{\perp_L},
\end{equation}
where the superscript $\perp_L$ denotes the orthogonal complement with
respect to the scalar product (\ref{scala}).  In fact, $u\in\cD(L_\m)$
and $v\in \cD(L_\m)^{\perp_L}$ if and only if
\begin{equation}
  0=(v|u)+(L_\M v|L_\M u)=(v|u)+(L_\M v|L_\m u).
\end{equation}
This means $L_\M v\in\cD(L_\m^*)$ and
\begin{equation}
  0=(L_\m^*L_\M v+v| u).
\end{equation}
Since $\cD(L_\m)$ is dense, we obtain
\begin{equation}
  0=v+L_\m^*L_\M v.
\end{equation}

(2) $\Ran (L_\m^\#)=\cH$ implies that $\Ran( L_\m)$ is
closed. Therefore,
\begin{equation}\label{direct}
  \cH=\Ran(L_\m)\oplus\Ker(L_\m^*).
\end{equation}
Let $u\in\cD(L_\M)$. By (\ref{direct}), there exist $v\in\cD(L_\m)$
and $w\in\Ker(L_\m^*)$ such that
\begin{equation}
  L_\M u=L_\m v+w.
\end{equation}
Hence $L_\M(u-v)=w$. Therefore, $u-v\in\Ker (L_\m^*L_\M)$. Thus we
have proven that
\begin{equation}
  \cD(L_\M)=\cD(L_\m)+\Ker(L_\m^* L_\M).
\end{equation}
Suppose now that $u\in\cD(L_\m)\cap\Ker(L_\m^*L_\M)$. Thus
\begin{equation}
  0=L_\m^*L_\M u=L_\m^*L_\m u.
\end{equation}
This implies $L_\m u=0$. But $\Ran (L_\m^\#)=\cH$ implies
$\Ker(L_\m)=\{0\}$.  Hence $u=0$.
\end{proof}
  
Our main goal in this and the next subsection is to study closed
operators $L_\bullet$ satisfying $L_\m\subset L_\bullet\subset L_\M$.
Such operators are the subject of the following proposition.

\begin{proposition}  \label{propo1}
  Let $L_\bullet$ be a closed operator such that
  $L_\bullet\subset L_\M$. Then $L_\m\subset L_\bullet$ if and only if
  $L_\bullet^\#\subset L_\m^\#$.
\end{proposition}

Let us reformulate Proposition \ref{propo1} for invertible
$L_\bullet$, using right inverses as the basic concept:

\begin{proposition}\label{propo2}
  Let $G_\bullet$ be a right inverse of $L_\M$ and
  $L_\bullet^{-1}=G_\bullet$. Then $L_\m\subset L_\bullet$ if and only
  if $G_\bullet^\#$ is a right inverse of $L_\m^\#$.
\end{proposition}

The following proposition should be compared with Prop. \ref{piuy7}.

\begin{proposition}\label{piuy7a}
  Suppose that $G_1$ is a right inverse of $L_\M$ and $G_1^\#$ is a
  right inverse of $L_\m^\#$. Then $G_2$ is also a right inverse of
  $L_\M$ and $G_2^\#$ is a right inverse of $L_\m^\#$ if and only if
  \begin{equation} G_1-G_2=K,\end{equation} where $K$ is a bounded
  operator in $\cH$ with $\Ran(K)\subset\Ker (L_\M)$ and
  $\Ran(K^\#)\subset\Ker( L_\m^\#)$.

  In particular, let $\Ker( L_\M)$ and $\Ker( L_\m^\#)$ be finite
  dimensional. Chose a basis $(u_1,\dots u_n)$ of $\Ker( L_\M)$ and a
  basis $(w_1,\dots w_m)$ of $\Ker( L_\m^\#)$.  Then
  \[
    G_1-G_2=\sum_{i,j}\alpha_{ij}|u_i\rangle\langle w_j| \quad\text{
      for some matrix }\quad [\alpha_{ij}].
  \]
\end{proposition}

\subsection{Nested pairs consisting of an operator and its transpose}
\label{ss:L2GII}

In this subsection we assume that the pair of densely defined closed
operators $L_\m$, $L_\M $ satisfies
\begin{equation}
  L_\m^\#=L_\M,\quad L_\m\subset L_\M,
\end{equation}
Note that this is a special case of conditions of the previous
subsection---now the two nested pairs (\ref{nest1}) and (\ref{nest2})
coincide with one another.

We will use the terminology related to symplectic vector spaces
introduced in  Appendix
\ref{a1}.

\begin{lemma} \label{sympl1}
  Let $u,v\in\cD(L_\M)$. Consider
  \begin{equation}
    \langle L_\M u|v\rangle-    \langle u| L_\M v\rangle.
    \label{sympl}\end{equation}
  Then (\ref{sympl}) is zero if $u\in\cD(L_\m)$ or $v\in\cD(L_\m)$. If
  we fix $u$ and (\ref{sympl}) is zero for all $v \in\cD(L_\M)$, 
  then $u\in\cD(L_\m)$.  Besides,
  \begin{equation} \label{sympl2}
    |\langle L_\M u|v\rangle- \langle u|
    L_\M v\rangle| \leq \|u\|_L\|v\|_L.
  \end{equation}
\end{lemma}

If $\phi,\psi\in\cD(L_\M)/\cD(L_\m)$ are represented by
$u,v\in\cD(L_\M)$, we set
\begin{equation}
  \lbra\phi|\psi\rbra:=    \langle L_\M u|v\rangle-    \langle u| L_\M
  v\rangle. 
\end{equation}
By Lemma \ref{sympl1}, $\lbra\cdot|\cdot\rbra$ is a well defined
continuous symplectic form on $\cD(L_\M)/\cD(L_\m)$.

To every closed operator $L_\bullet$ such that
$L_\m\subset L_\bullet\subset L_\M$ we associate
a closed subspace  $\cW_\bullet$ of $\cD(L_\M)/\cD(L_\m)$ by
\[\cW_\bullet:=\cD(L_\bullet)/\cD(L_\m).\]

\begin{proposition} \label{sympl6}
$ $  \begin{enumerate}
\item The above correspondence is bijective.
\item If $L_\bullet$ is mapped to $\cW_\bullet$, then $L_\bullet^\#$
    is mapped to $W_\bullet^{\,\s\!\perp}$ (the symplectic orthogonal
    complement of $\cW_\bullet$).
\item Self-transposed operators are mapped to Lagrangian subspaces.
\end{enumerate}

\end{proposition}

The following result is quite striking and shows that in a certain
respect the concept of self-transposedness is superior to the concept
of self-adjointness. It is due to Galindo \cite{Gal}, with a
simplified proof given by Knowles \cite{Kn}, see also \cite{EE}.  It
is a generalization of a well-known property of real Hermitian
operators: they have a self-adjoint extension which commutes with the
usual conjugation.

\begin{theorem}\label{zorn}
  There exists a self-transposed operator $L_\bullet $ such that
  $L_\m\subset L_\bullet \subset L_\M$. Moreover,
  $\dim \cD(L_\M^\#)/\cD(L_\m)$ is even or infinite
  and
  \begin{equation}
    \dim\cD(L_\bullet )/\cD(L_\m)=
    \dim\cD(L_\M^\#)/\cD(L_\bullet )=
    \frac12\dim\cD(L_\M^\#)/\cD(L_\m).
  \end{equation}
\end{theorem}

\begin{proof}
  By Prop.  \ref{sympl4}, there exists a Lagrangian subspace
  $\cW_\bullet$ contained 
   in the symplectic space
  $\cD(L_\M)/\cD(L_\m)$. By Prop. \ref{sympl5} it is closed.  The
  corresponding operator $L_\bullet$ is self-transposed
  (Prop. \ref{sympl6}).
\end{proof}

\begin{proposition}\label{zorn-1}
  Suppose that $L_\m\subset L_\bullet \subset L_\M$ and
  \begin{equation}
    \dim\cD(L_\bullet )/\cD(L_\m)= \dim\cD(L_\M^\#)/\cD(L_\bullet )=
    \frac12\dim\cD(L_\M^\#)/\cD(L_\m)=1.
  \end{equation}
  Then $L_\bullet$ is self-transposed.
\end{proposition}

\begin{proof}
  All 1-dimensional subspaces in a 2-dimensional symplectic space are
  Lagrangian.
\end{proof}

Here is a version of Prop. \ref{dire} adapted to the present context:

\begin{proposition}\label{dire1} $ $  
\begin{enumerate}
\item
    We have a direct decomposition
  \begin{equation}
    \cD(L_\M)=\cD(L_\m)\oplus\Ker(\bar L_\M L_\M+1).
  \end{equation}
\item If in addition $\Ran (L_\M)=\cH$, then
  \begin{equation}
    \cD(L_\M)=\cD(L_\m)\oplus\Ker(\bar L_\M L_\M).
  \end{equation}
\end{enumerate}
  \end{proposition}

Here is a consequence of Prop. \ref{propo2}

\begin{proposition}\label{piuy00+}
  Suppose that $L_\bullet$ is invertible and
  $ L_\m\subset L_\bullet\subset L_\M$. Then there exists a
  self-transposed invertible $L_1$ such that
  $ L_\m\subset L_1\subset L_\M$.
\end{proposition}

\begin{proof} 
  $G_\bullet:=L_\bullet^{-1}$ is a right inverse of $L_\M$.  By
  Prop. \ref{propo2}, $G_\bullet^\#$ is also a right inverse of
  $L_\M $.  Therefore, also $G_1: =(G_\bullet+G_\bullet^\#)/2$ is 
a  right inverse, which in addition is self-transposed.  Now $L_1$ such
  that $L_1^{-1}=G_1$ has the required properties.
\end{proof}

Here is a version of Prop. \ref{piuy7a} adapted to the present
context:

\begin{proposition}\label{piuy7b}
  Suppose that $G_1$ is a right inverse of $L_\M$ and $G_1^\#$ is a
  right inverse of $L_\M$. Then $G_2$ is also a right inverse of
  $L_\M$ and $G_2^\#$ is a right inverse of $L_\M$ if and only if
  \begin{equation}
    G_1-G_2=K,
  \end{equation}
  where $K$ and $K^\#$ are bounded from $\cH$ to
  $\Ker (L_\M)$.
  
  In particular, let $\Ker( L_\M)$ be finite dimensional. Chose a
  basis $(u_1,\dots u_n)$ of $\Ker(L_\M)$.  Then
  \[
    G_1-G_2=\sum_{i,j}\alpha_{ij}|u_i\rangle\langle u_j| \quad\text{
      for some matrix }\quad [\alpha_{ij}].
  \]
\end{proposition}

\begin{theorem}\label{fred1}
  Suppose that $ L_\m\subset L_\bullet\subset L_\M$.
  \begin{enumerate}
  \item If $L_\bullet$ is Fredholm of index $0$, then
    \begin{equation}\label{dimen3}
      \dim\cD(L_\M)/\cD(L_\bullet)=\cD(L_\bullet)/\cD(L_\m)
      =\frac{1}{2}\dim\cD(L_\M)/\cD(L_\m)  .
    \end{equation}
  \item If $\dim\cD(L_\M)/\cD(L_\m)<\infty$ and \eqref{dimen3} holds,
    then $L_\bullet$ is Fredholm of index $0$.
\end{enumerate}
\end{theorem}

\begin{remark}
  Thm. \ref{fred1} implies that if
  $ L_\m\subset L_\bullet\subset L_\M$ and
  $\rs_{\mathrm{F}0}(L_\bullet)\neq\emptyset$, then $L_\bullet$ has to
  satisfy \eqref{dimen3} (see \eqref{fred2} for the definition of
  $\rs_{\mathrm{F}0}$). Thus the most ``useful'' operators (which
  usually means well-posed ones) are ``in the middle'' between the
  minimal and maximal operator.
\end{remark}

\section{Basic \texorpdfstring{$L^2$}{L2} theory of 1d Schr\"odinger operators} \label{s:l2ode}
\protect\setcounter{equation}{0}

\subsection{The maximal and minimal operator}

As before, we assume that $V\in L_\loc^1]a,b[$. Recall that $L$ is the differential expression
     \begin{equation}
  L:=-\partial^2+V .
\end{equation}
     In this section we present basic realizations of $L$ as closed operators on
     $L^2]a,b[$.

\begin{definition} \label{df:lmax} The \emph{maximal operator}
  $L_{\max}$ is defined by
\begin{align}\mathcal{D}(L_{\max})&:=\big\{f\in L^2]a,b[\,\cap\,
AC^1]a,b[ \ \mid  L f\in L^2]a,b[\big\},\label{lm:ode-v2}\\
  L_{\max}f&:=Lf,\quad f\in\mathcal{D}(L_{\max}).
\end{align}
We equip $\mathcal{D}(L_{\max})$ with the graph norm
\[
  \|f\|_{L}^2:=\|f\|^2+\|Lf\|^2.
\]
\end{definition}

\begin{remark} Note that $L^2]a,b[\subset L_\loc^1]a,b[$. Therefore,
  as explained in Subsect.  \ref{Choice of funtional-analytic
    setting}, $f\in L_\loc^\infty]a,b[$ and $Lf\in L^2]a,b[$ implies
  $f\in AC^1]a,b[$. Therefore, in (\ref{lm:ode-v2}) we can replace
  $AC^1]a,b[$ with $ L_\loc^\infty]a,b[$ {\rm(}or $C]a,b[$, or
  $C^1]a,b[\,${\rm)}. \end{remark}

Recall that $AC_\mathrm{c}^1]a,b[$ are once absolutely differentiable
functions of compact support.

\begin{definition}
We set
\[
  \mathcal{D}(L_{\mathrm{c}}):= AC_{\mathrm{c}}^1]a,b[\,\cap\,
  \mathcal{D}(L_{\max}).
\]
Let $L_{\mathrm{c}}$ be the restriction of $L_{\max}$ to
$\mathcal{D}(L_{\mathrm{c}})$.  Finally, $L_{\min}$ is defined as the
closure of $L_{\mathrm{c}}$.
\end{definition}

The next theorem is the main result of this subsection:

\begin{theorem}\label{perq2}
The operators $L_\m,L_\M$ have the following properties.
  \begin{enumerate}
  \item The operators $L_{\max}$ and $L_{\min}$ are closed, densely
    defined and  ${L_{\min}\subset L_{\max}}$.
  \item \label{dontknow0} 
    $L_{\max}^\#=L_{\min}$ and $L_{\min}^\#=L_{\max}$.
  \item Suppose that $f_1,f_2\in\mathcal{D}(L_{\max})$. Then there
    exist
    \begin{align}
      W(f_1,f_2;a)&:=\lim\limits_{d\searrow
                    a}W(f_1,f_2;d),\label{wronski-a}\\ 
      W(f_1,f_2;b)&:=\lim\limits_{d\nearrow
                    b}W(f_1,f_2;d),\label{wronski-b} 
      \end{align}
      and the so-called {\em Green's identity} (the integrated form of
      the Lagrange identity) holds:
      \begin{equation}\label{eq:lagrange1}
        \langle L_{\max}f_1|f_2\rangle-
        \langle f_1|L_{\max}f_2\rangle=W(f_1,f_2;b)-W(f_1,f_2;a).
      \end{equation}

    \item We set $W_d(f_1,f_2)= W(f_1,f_2;d)$ for any $d\in[a,b]$ and
      $f_1,f_2\in\mathcal{D}(L_{\max})$. Then for any $d\in[a,b]$ the
      map $W_d:\cD(L_{\max})\times \cD(L_{\max})\to\C$ is a continuous
      bilinear antisymmetric form, in particular
      \begin{equation}
        |W_d(f_1,f_2)|\leq C_d\|f_1\|_{L}\|f_2\|_{L}.
        \label{wron}
      \end{equation}
    \item\label{dontknow1}
      $ \mathcal{D}(L_{\min})$ coincides with
      \begin{equation}
        \big\{f\in\mathcal{D}(L_{\max})\ \mid\ W(f,g;a)=0 \text{ and }
        W(f,g;b)=0  \text{ for all }g\in\mathcal{D}(L_{\max})\big\}.
        \label{dontknow2}
      \end{equation}
    \item $\bar{L}_\m=\bar{L_\m}=L_\M^*$ and
       $\bar{L}_\M=\bar{L_\M}=L_\m^*$.
      \end{enumerate}
\end{theorem}

One of the things we will need to prove is the density of
$ \mathcal{D}(L_\mathrm{c})$ in $L^2]a,b[$.  This is easy if
$V\in L_\loc^2]a,b[$ (see Prop.\ \ref{pr:l2}), but with our assumptions
on the potential the proof is not so trivial, because the idea of
approximating an $f\in L^2(I)$ with smooth functions does not work:
$\mathcal{D}(L_{\max})$ may not contain any ``nice'' function, as the
example described below shows.

\begin{example} {\rm Let
    $V(x)=\sum_{\sigma} c_\sigma|\x-\sigma|^{-1/2}$ where $\sigma$
    runs over the set of rational numbers and $c_\sigma\in\R$ satisfy
    $c_\sigma>0$ and $\sum_\sigma c_\sigma<\infty$. Then
    $V\in L^1_\loc(\R)$ but $V$ is not square integrable on any
    nonempty open set. Hence there is no $C^2$ nonzero function in the
    domain of $L$ in $L^2(\R)$.}
\end{example}

Before proving Thm \ref{perq2}, we first state an immediate
consequence of Lemma \ref{lm:ode-v}:
    
\begin{lemma}\label{lm:embed}
{\rm(1)} Let $J$ be a finite interval whose closure is contained in
$]a,b[$. Then
\begin{align}
\big\| f\big|_J\big\|\leq C_J\|f\|_{L},\label{lem-1} \\
\big\| f'\big|_J\big\|\leq C_J\|f\|_{L}.\label{lem-2}
\end{align}
{\rm(2)} Let $\chi\in C^\infty]a,b[$ with $\chi'\in
C_\mathrm{c}^\infty]a,b[$.  Then 
\begin{equation}
  \|\chi f\|_{L}\leq
  C_\chi\|f\|_{L}.\label{chi}
\end{equation} 
\end{lemma}

As in the previous section, we fix $u ,v\in AC^1]a,b[$ that span
$\Ker (L)$ and satisfy $W(v ,u)=1$.

Our proof of Thm \ref{perq2} uses ideas from \cite[Theorem 10.11]{S}
and \cite[ Sect.\ 17.4]{Nai} and is based on an abstract result
described in Lemma \ref{lm:abst}.  The following lemma
about the regular case  (cf.\ Definition \ref{df:reg})
contains the
key arguments of the proof of (1) and (2) of Thm \ref{perq2}:

\begin{lemma}\label{lem-a}
  If $V\in L^1]a,b[$ and $a,b\in\R$,  then
  \begin{enumerate}\item $\Ker (L_{\max})=\Ker( L)$.
  \item $\Ran (L_{\max})=L^2]a,b[$.
  \item $\langle L_\mathrm{c}f|g\rangle= \langle f|L_{\max}g\rangle$,
    $f\in\cD(L_\mathrm{c})$, $g\in\cD(L_{\max})$.
  \item $\Ran (L_\mathrm{c})=L_\mathrm{c}^2]a,b[\,\cap\,\Ker( L)^\per$.
  \item $\Ran( L_\mathrm{c})^\per=\Ker( L)$.
  \item $\cD(L_\mathrm{c})$ is dense in $L^2]a,b[$.
    \end{enumerate}
\end{lemma}

\proof Clearly, $\Ker (L)=\Span(u,v)\subset AC^1[a,b]\subset
L^2]a,b[$. Therefore, $\Ker (L)\subset\cD(L_{\max})$. This proves (1).

Recall that in (\ref{qeq0}) we defined the forward Green's operator
$G_\rightarrow $.  Under the assumptions of the present lemma, it maps $L^2]a,b[$
into $AC^1[a,b]$. Therefore, for any $g\in L^2]a,b[$,
$\alpha,\beta\in\C$,
\[
  f=\alpha u+\beta v+G_\rightarrow  g
\]
belongs to $AC^1[a,b]$ and verifies $Lf=g$. Therefore, $f\in\cD(L_{\max})$.
Hence $L_{\max}$ is surjective. This proves (2).

To obtain (3) we integrate twice by parts. This is allowed by
(\ref{leibnitz1}) and (\ref{leibnitz2}),
since $f,g\in AC^1[a,b]$.

It is obvious that $\Ran (L_\mathrm{c})\subset L_\mathrm{c}^2]a,b[$.
$\Ran( L_\mathrm{c})\subset \Ker( L_{\max})^\per$ follows from (3).

Let us prove the converse inclusions. Let
$g\in L_\mathrm{c}^2]a,b[\,\cap\,\Ker( L)^\per$. Set
$f:=G_\rightarrow g$. Clearly, $Lf=g$. Using $\int_a^b g u=\int _a^b g v=0$ we
see that $f$ has compact support. Hence $f\in\cD(L_\mathrm{c})$. This
proves (4).

$L_\mathrm{c}^2]a,b[$ is dense in $L^2]a,b[$ and $\Ker( L)^\per$ has a
finite codimension. Therefore, by Lemma \ref{dense}, given below,
$L_\mathrm{c}^2]a,b[\,\cap\,\Ker( L)^\per$ is dense in $\Ker( L)$. This
implies (5).

By applying Lemma \ref{lm:abst} with $T:=L_{\max}$ and
$S:=L_\mathrm{c}$, we obtain (6). \qed

\begin{lemma} Let $\cH$ be a Hilbert space and $\cK$ a closed subspace
  of finite codimension. If $\cZ$ is a dense subspace of $\cH$, then
  $\cZ\cap\cK$ is dense in $\cK$.
\label{dense}\end{lemma}

\begin{proof}
  The lemma is obvious if the codimension is $1$. Then we apply
  induction.
\end{proof}

\medskip
    
\noindent{\bf Proof of Thm \ref{perq2}.}
It follows from Lemma \ref{lem-a} (6) that $\cD(L_{\mathrm{c}})$ is
dense in $L^2]a,b[$.  We have
\begin{equation}
  \label{supset} L_\mathrm{c}^\#\supset L_{\max}
\end{equation}
by integration by parts, as in the proof of (3), Lemma \ref{lem-a}.

Suppose that $h,k\in L^2]a,b[$ such that
\begin{equation}
  \langle L_\mathrm{c}f|h\rangle=\langle f|k\rangle,\quad
  f\in\cD(L_\mathrm{c}).
\end{equation}
In other words, $h\in\cD(L_\mathrm{c}^\#)$ and $L_\mathrm{c}^\#h=k$.
Choose $d\in]a,b[$. We set $h_d:=G_dk$, where $G_d$ is defined in
Def. \ref{greend}. Clearly, $Lh_d=k$. For $f\in\cD(L_\mathrm{c})$, set
$g:=L_\mathrm{c}f$.  We can assume that $\supp f\subset[a_1,b_1]$ for
$a<a_1<b_1<b$.  Now
\[
  \langle g|h_d\rangle= \langle L_\mathrm{c}f|h_d\rangle = \langle
  f|Lh_d\rangle= \langle f|k\rangle = \langle L_\mathrm{c}
  f|h\rangle = \langle g|h\rangle.
\]
By Lemma \ref{lem-a} (4) applied to $[a_1.b_1]$,
\begin{equation} \label{equo}
  h=h_d+\alpha u+\beta v
\end{equation}
on $[a_1,b_1]$. But since $a_1,b_1$ were arbitrary under the condition
$a<a_1<b_1<b$, (\ref{equo}) holds on $]a,b[$. Hence $Lh=k$. Therefore,
$h\in \cD(L_{\max})$ and $L_{\max}h=k$.  This proves that
\begin{equation} \label{subset}
  L_\mathrm{c}^\#\subset L_{\max}.
\end{equation}
From (\ref{supset}) and (\ref{subset}) we see that
$L_\mathrm{c}^\#= L_{\max}$.  In particular, $L_{\max}$ is closed and
$L_\mathrm{c}$ is closable. We have
\begin{equation}
  L_{\min}=L_\mathrm{c}^{\#\#}= L_{\max}^\# .
\end{equation}
This ends the proof of (1) and (2).

For $f,g\in \mathcal{D}(L_{\max})$ and $a<a_1<b_1<b$ we have
\begin{align}\label{eq:lagrange}
\int_{a_1}^{b_1}(Lf(x) g(x)-f(x) Lg(x))\d x & = \int_{a_1}^{b_1}(f(x)g'(x)-f'(x)g(x))'\d x 
\\
&= W(f,g;a_1)-W(f,g;b_1).
\notag\end{align}
The lhs of
(\ref{eq:lagrange}) clearly converges
as $a_1\searrow a$. Therefore, the limit (\ref{wronski-a}) exists.
Similarly, by taking $b_1\nearrow b$
we show that the limit (\ref{wronski-b}) exists.
Taking both limits we obtain (\ref{eq:lagrange1}). This proves (3).

If $d\in]a,b[$, then (\ref{wron}) is an immediate consequence of (\ref{lem-1}) and (\ref{lem-2}).
    We can rewrite (\ref{eq:lagrange}) as
    \begin{align}\label{eq:lagrange2}
W(f,g;a)=-\int_{a}^{d}\big((Lf)(x) g(x)-f(x) Lg(x)\big)\d x 
+ W(f,g;d).
\end{align}
Now both terms on the right of (\ref{eq:lagrange2}) can be estimated
by $C \|f\|_{L}\|g\|_{L}$. This shows (\ref{wron}) for
$d=a$.  The proof for $d=b$ is analogous.

    Let $L_w$ be $L$ restricted to
    (\ref{dontknow2}).
By (\ref{wron}),  (\ref{dontknow2}) is a closed subspace of $\cD( L_{\max})$. Hence, $L_w$ is closed.
    Obviously, $L_\mathrm{c}\subset L_w$. By 
    (\ref{eq:lagrange1}),
    $L_w\subset L_{\max}^\#$. By (2), we know that $ L_{\max}^\#=L_{\min}$.
But $L_{\min}$ is the closure of $L_\mathrm{c}$.    Hence $L_w=L_{\min}$. This proves (5). \qed

\begin{remark}{\rm Here is an alternative, more direct proof of the
    closedness of $L_{\max}$. Let $f_n\in \mathcal{D}(L_{\max})$ be a
    Cauchy sequence wrt the graph norm.  This means that $f_n$ and
    $Lf_n$ are Cauchy sequences wrt $L^2]a,b[$.  Let
    $f:=\lim\limits_{n\to\infty}f_n$,
    $g:=\lim\limits_{n\to\infty}Lf_n$.  Let $J$ be an arbitrary
    sufficiently small closed interval in $]a,b[$. We have
  \begin{align}
    \|f_n-f_m\|_{L^1(J)}&\leq\sqrt{|J|}        \|f_n-f_m\|_{L^2(J)}
                          ,\\
    \|Lf_n-Lf_m\|_{L^1(J)}&\leq\sqrt{|J|}        \|Lf_n-Lf_m\|_{L^2(J)}.
  \end{align}
  Hence $f_n$ satisfies the conditions of Cor.  \ref{cor1}. Hence
  $f\in AC^1]a,b[$ and $g=Lf$. Hence $f\in \mathcal{D}(L_{\max})$ and
  it is the limit of $f_n$ in the sense of the graph norm. Therefore,
  $\mathcal{D}(L_{\max})$ is complete. Hence $L_{\max}$ and $L_{\min}$
  are closed.}\end{remark}

\subsection{Smooth functions in the domain of
  \texorpdfstring{$L_\M$}{Lmax}} 

We point out a certain pathology of the operators $L_\M$ and $L_\m$ if
$V$ is only locally integrable.

\begin{lemma}\label{lm:conjugate}
  (1) The imaginary part of $V$ is locally square integrable if and
  only if $\cD(L_{\mathrm{c}})$ is stable under conjugation and in
  this case $\cD(L_\m)$ and $\cD(L_\M)$ are also stable under
  conjugation. \\[1mm]
  (2) If the imaginary part of $V$ is not square integrable on any
  open set, then for $f\in\cD(L_\M)$ we have $\bar{f}\in\cD(L_\M)$
  only if $f=0$. In other words, $\cD(L_\M)\cap\cD(\bar {L}_\M)=\{0\}$.  Hence $\cD(L_\M)$ does not contain any nonzero real
  function.
\end{lemma}

\begin{proof}
(1):
  Write $Lf=-f''+V_1f+iV_2f$ if $V=V_1+iV_2$ with $V_1,V_2$ real. Then
  if $V_2\in L^2_\loc]a,b[$ and $f\in AC^1_{\mathrm{c}}]a,b[$ we have
  $V_2f\in L^2]a,b[$ so $-f''+Vf\in L^2]a,b[$ if and only if
  $-f''+V_1f\in L^2]a,b[$ hence $-\bar{f}''+V_1\bar{f}\in L^2]a,b[$ so
  we get $-\bar{f}''+V\bar{f}\in L^2]a,b[$, thus $\cD(L_{\mathrm{c}})$
  is stable under conjugation. The corresponding assertion concerning
  $\cD(L_\m)$ follows by taking the completion, and that concerning
  $\cD(L_\M)$ follows by taking the transposition.

  Reciprocally, assume that $\cD(L_{\mathrm{c}})$ is stable under
  conjugation and let $x_0\in]a,b[$. Then there is
  $f\in\cD(L_{\mathrm{c}})$ such that $f(x_0)\neq0$ and we may assume
  that its real part $g=(f+\bar{f})/2$ does not vanish on a
  neighbourhood of $x_0$. Then $g\in \cD(L_{\mathrm{c}})$ hence
  $-g''+V_1g+iV_2g\in L^2]a,b[$ and so must be the imaginary part of
  this function hence $V_2$ is square integrable on a neighbourhood of
  $x_0$. 
  
(2):  Assume now that $V_2$ is not square integrable on any open set. If
  $f\in AC^1$ is real then $-f''+Vf\in L^2$ if and only if
  $-f''+V_1f\in L^2$ and $V_2f\in L^2$ and if $f\neq0$ then the second
  condition implies $f=0$. Finally, if $f\in\cD(L_\M)$ and
  $\bar{f}\in\cD(L_\M)$ then the functions $f+\bar{f}$ and $f-\bar{f}$
  will be zero by (1).
\end{proof}

\begin{remark}\label{re:conjugate}{\rm
    Clearly,  $L_\m^*=\bar{L}_\M$.
    %and $\cD(\bar{L}_\M)=\{\bar{f}\mid f\in\cD(L_\M)\}$.
    Thus, by Prop. \ref{lm:conjugate} (2), if the imaginary part of
    $V$ is not square integrable on any open set then
    $\cD(L_\m)\cap\cD(L_\m^*)=\{0\}$. On the other hand, if the
    imaginary part of $V$ is locally square integrable, then
    $\cD(L_\m)\subset\cD(L_\m^*)$.  }\end{remark}

If $V\in L_\loc^2$, many things simplify:

\begin{proposition}\label{pr:l2}
  If $V\in L^2_\loc]a,b[$ then $C_\mathrm{c}^\infty]a,b[$ is a dense
  subspace of $\mathcal{D}(L_{\min})$.
\end{proposition}

\begin{proof}
  Clearly $C_\mathrm{c}^\infty]a,b[\subset\mathcal{D}(L_\mathrm{c})$.
  Let $f\in C_\mathrm{c}]a,b[$.  Then $Lf\in L^2]a,b[$ if and only if
  $f''\in L^2]a,b[$. Fix some $\theta\in C_\mathrm{c}^\infty(\R)$ with
  $\int\theta=1$ and let $\theta_n(x):=n\theta(nx )$ with
  $n\geq1$. Then for $n$ large
  $f_n:=\theta_n*f\in C_\mathrm{c}^\infty]a,b[$ and has support in a
  fixed small neighbourhood of $\supp f$. Moreover, $f_n\to f$ in
  $C^1_\mathrm{c}]a,b[$, in particular $f_n\to f$ uniformly with
  supports in a fixed compact, which clearly implies $Vf_n\to Vf$ in
  $L^2]a,b[$. Moreover $f_n''\to f''$ in $L^2]a,b[$.
\end{proof}

\subsection{Closed operators contained in
\texorpdfstring{$L_{\max}$}{Lmax}}

If $\cD(L_\bullet)$ is a subspace of $\cD(L_{\max})$ closed in the
$\|\cdot\|_{L}$ norm, then the operator
\begin{equation}
  L_\bullet:=L_{\max}\Big|_{\cD(L_\bullet)}
\end{equation}
is closed and contained in $L_{\max}$. We can call such an operator
$L_\bullet$ a {\em closed realization of $L$}.

We will be mostly interested in operators $L_\bullet$ that satisfy
\begin{equation}
  \cD(L_{\min})\subset \cD(L_\bullet) \subset \cD(L_{\max})
\end{equation}
so that
\begin{equation}
  L_{\min}\subset L_\bullet \subset L_{\max}.
\end{equation}
They are automatically densely defined.
Note that we can use the theory developed in Subsect. \ref{ss:L2GII-} and \ref{ss:L2GII}.
In particular, as descibed in Prop. \ref{propo1},
one can easily check if a realization of $L$ contains $L_{\min}$ with the
help of the following criterion:

\begin{proposition}
  Suppose that $L_\bullet$ is a closed densely defined operator
  contained in $L_{\max}$.  Then $L_\bullet^\#$ is contained in
  $L_{\max}$ if and only if $L_{\min}\subset L_\bullet$.  In
  particular, if $L_\bullet^\#=L_\bullet$, then
  $L_{\min}\subset L_\bullet$.
\end{proposition}

The most obvious examples of such operators are given by one-sided
boundary conditions:
\begin{definition}\label{defno}
  Set
  \begin{align}  \label{proba}
    \mathcal{D}(L_a)&:=\{f\in\mathcal{D}(L_{\max})\
     \mid\ W(f,g;a)=0 \text{ for all }g\in \mathcal{D}(L_{\max})\},\\ 
    \mathcal{D}(L_b)&:=\{f\in\mathcal{D}(L_{\max})\ \mid\
                  W(f,g;b)=0 \text{ for all }g\in \mathcal{D}(L_{\max})\}.
  \end{align}
  Let $L_a$, resp.  $L_b$ be $L_{\max}$ restricted to
  $\cD(L_a)$, resp. $\cD(L_b)$.
\end{definition}

\begin{proposition}
  $L_a$ and $L_b$ are closed and densely defined operators satisfying
  \begin{align}
    L_a^\#=L_b,&\quad L_b^\#=L_a,\\
    L_{\min}\subset L_a\subset L_{\max},
               &\quad L_{\min}\subset L_b\subset L_{\max}.
  \end{align}
\end{proposition}

\section{Boundary conditions}
\protect\setcounter{equation}{0}

\subsection{Regular endpoints}\label{ss:bvf}

Recall that the endpoint $a$ is regular if it is finite and $V$ is
integrable close to $a$.

\begin{proposition}\label{re:reg}
  If $L$ is regular at $a$ then any function $f\in\cD(L_{\max})$
  extends to a function of class $C^1$ on the left closed interval
  $[a,b[$, hence $f(a)$ and $f'(a)$ are well defined, and for
  $f,g\in\cD(L_{\max})$ we have
  $W_a(f,g)=f(a)g'(a)-f'(a)g(a)$. Similarly if $L$ is regular at
  $b$. Thus if $L$ is regular then $\cD(L_{\max})\subset C^1[a,b]$ and
  Green's identity \eqref{eq:lagrange1} has the classical form
  \[
    \braket{L_{\max}f_1}{f_2}- \braket{f_1}{L_{\max}f_2}=
    \big(f_1(b)f_2'(b)-f_1'(b)f_2(b)\big)- \big
    (f_1(a)f_2'(a)-f_1'(a)f_2(a)\big).
  \]
\end{proposition}

Thus if $L$ is a regular operator then we have four continuous linear
functionals on $f\in \cD(L_{\max})$
\begin{align}
  f\mapsto f(a),&\quad   f\mapsto f'(a),\label{pom1}\\
  f\mapsto f(b),&\quad   f\mapsto f'(b),\label{pom2}
\end{align}
which  give a
  convenient description of closed operators $L_\bullet$ such that
  $L_{\min}\subset L_\bullet\subset L_{\max}$.
  In particular, $\cD(L_{\min})$ is the intersection of the kernels of
  (\ref{pom1}) and   (\ref{pom2}),
$\cD(L_a)$ is the intersection of the kernels of (\ref{pom1}) and
$\cD(L_b)$ is the intersection of the kernels of (\ref{pom2}).

\subsection{Boundary functionals}
 \label{triplet}
 It is possible to extend the strategy described above to the case of
 an arbitrary $L$ by using an abstract version of the notion of
 boundary value of a function. We shall do it in this section.

  The abstract theory of boundary value functionals goes back to
  J.~W. Calkin's thesis \cite{Cal} who used it for the classification
  of self-adjoint extensions of Hermitian operators. The theory was
  adapted to Hermitian differential operators of any order by Naimark
  \cite{Nai} and to operators with complex coefficients of class
  $C^\infty$ by Dunford and Schwarz in \cite[ch.\ XIII]{DS2}.  In this
  section we shall use this technique in the case of second order
  operators with potentials which are only locally integrable: this
  loss of regularity is a problem for some arguments in \cite{DS2}.

  Recall that $\cD(L_{\max})$ is equipped with the Hilbert space
  structure associated to the norm
  $\|f\|_L=\sqrt{\|f\|^2+\|Lf\|^2}$. Following \cite[\S XXX.2]{DS2},
  we introduce the following notions.

\begin{definition}\label{df:bvf} 
  A \emph{boundary functional for $L$} is any linear
    continuous form on $\cD(L_{\max})$ which vanishes on
    $\cD(L_{\min})$. A \emph{boundary functional at $a$} is a boundary
    functional $\phi$ such that $\phi(f)=0$ for all $f\in\cD(L_{\max})$
    with $f(x)=0$ near $a$;  \emph{boundary functionals at $b$} are
    defined similarly. $\cb(L)$ is the set of boundary functionals for $L$
    and $\cb_a(L),\cb_b(L)$ the subsets of boundary functionals at $a$ and
    $b$.
\end{definition}

$\cb(L)$ is a closed linear subspace of the topological dual
  $\cD(L_{\max})'$ of $\cD(L_{\max})$ and $\cb_a(L),\cb_b(L)$ are
  closed linear subspaces of $\cb(L)$. By using a partition of unity
  on $]a,b[$ it is easy to prove that
\begin{equation}\label{eq:Bab}
\cb(L)=\cb_a(L)\oplus\cb_b(L),
\end{equation}
a topological direct sum.

\begin{definition}\label{df:dimbv}
  We define
  \begin{align*}
   \text{ the {\em boundary index  for $L$ at $a$}},&\quad
  \nu_a(L):=\dim\cb_a(L),\\\text{ the
    {\em boundary index for $L$  at $b$}},&\quad \nu_b(L):=\dim\cb_b(L),\\
  \text{ and the
  {\em total boundary index for $L$}},&\quad 
  \nu(L):=\dim\cb(L)=\nu_a(L)+\nu_b(L).
  \end{align*}
\end{definition}

By definition, the subspace $\cb(L)\subset\cD(L_{\max})'$ is the
  polar set of the closed subspace $\cD(L_\m)$ of $\cD(L_\M)$. Hence it is
  canonically identified with the dual space of
  $\cD(L_{\max})/\cD(L_\m)$:
\begin{equation}\label{eq:polarB}
\cb(L)=\big(\cD(L_\M)/\cD(L_\m)\big)' .
\end{equation}
Clearly one may also define $\cb_a(L)$ as the set of continuous
  linear forms on $\cD(L_{\max})$ which vanish on the closed subspace
  $\cD(L_a)$, and similarly for $\cb_b(L)$. Thus
\begin{equation}\label{eq:polar}
\cb_a(L)=\big(\cD(L_{\max})/\cD(L_a)\big)' .
\end{equation}

\begin{definition}
For each $f\in\cD(L_\M)$ and $x\in[a,b]$,  we introduce the functional
\begin{equation}\label{eq:fata}
  \vec f_x:\cD(L_\M)\to\C
  \quad\text{defined by}\quad
  \vec f_x (g)=W_x(f,g).
\end{equation}
 By Theorem \ref{perq2} it is a well defined linear continuous form on $\cD(L_\M)$. \label{def:vec}\end{definition}

Remember that if $x\in]a,b[$, then we can write
\begin{equation}\label{eq:dirac}
  \vec{f}_x(g)=f(x)g'(x)-f'(x)g(x)=W_x(f,g) \quad
  \forall g\in C^1]a,b[ .
\end{equation}
If $x=a$, in general we cannot write (\ref{eq:dirac}) (unless $a$ is
regular). However we know that for all $x\in[a,b]$ (\ref{eq:fata})
depends weakly continuously on $x$. Thus in general
\begin{equation}\label{eq:bvflim}
  \wlim_{x\to a}\vec f_x=\vec f_a.
\end{equation}

It is easy to see that $\vec f_a\in\cB_a$, cf.\ \eqref{proba} for
example. We shall prove below that any boundary value functional at
the endpoint $a$ is of this form.

\begin{theorem}\label{th:bvf}
  {\rm(i)} $f\mapsto \vec f_a$ is a linear surjective map
  $\cD(L_\M)\to\cb_a(L)$.\\[1mm]
  {\rm(ii)} $W_a(f,g)=0$ for all $f,g\in\cD(L_\M)$ if and only
  if $\cb_a(L)=\{0\}$. \\[1mm]
  {\rm(iii)} $W_a(f,g)\neq0$ if and only if the functionals
  $\vec f_a, \vec g_a$ are linearly independent. \\[1mm]
  {\rm(iv)} If $W_a(f,g)\neq0$ then $\{\vec f_a, \vec g_a\}$ is a
  basis in $\cb_a(L)$; then $\forall h\in\cD(L_\M)$ we have
  \begin{equation}\label{eq:kodafgh}
    \vec h_a=cW_a(g,h)\vec f_a+cW_a(h,f) \vec g_a
    \quad\text{ with } c=-1/W_a(f,g). 
  \end{equation}
\end{theorem}

\begin{proof}
  Let $\mathcal{W}_a$ be the set of linear forms of the form
    $\vec f_a$, this is a vector subspace of $\cb_a(L)$ and we shall
    prove later that $\mathcal{W}_a=\cb_a(L)$. For the moment, note
    that $\mathcal{W}_a$ separates the points of
    $\cY_a:=\cD(L_{\max})/\cD(L_a)$, i.e.\ we have $W_a(f,g)=0$
    for all $f$ if and only if $g\in\cD(L_a)$, cf.\ \eqref{dontknow2}
    and \eqref{proba}. On the other hand, \eqref{eq:polar} implies
    that $\cb_a(L)=\{0\}$ is equivalent to $\cD(L_\M)=\cD(L_a)$ which
    in turn is equivalent to $W_a(f,g)=0$ for all $f,g\in\cD(L_\M)$ by
    \eqref{proba}. This proves (ii).

  For the rest of the proof we need \emph{Kodaira's identity}
  \cite[pp.\ 151--152]{Pry}, namely: if $f,g,h,k$ are $C^1$ functions
  on $]a,b[$ then
  \begin{equation}\label{eq:fghk}
    W(f,g)W(h,k) +W(g,h) W(f,k)+W(h,f) W(g,k)=0 ,
\end{equation}
with the usual definition $W(f,g)=fg'-f'g$.  The relation obviously
holds pointwise on $]a,b[$. If $f,g,h,k\in\cD(L_\M)$, then the relation
extends to $[a,b]$, in particular
\begin{equation}\label{eq:fghka}
  W_a(f,g)W_a(h,k) +W_a(g,h) W_a(f,k)+W_a(h,f) W_a(g,k)=0,
\end{equation}
and similarly at $b$. This implies \eqref{eq:kodafgh} if
$W_a(f,g)\neq0$ from which it follows that $\{\vec f_a, \vec g_a\}$
is a basis in the vector space $\mathcal{W}_a$, in particular
$\mathcal{W}_a$ has dimension $2$. But $\mathcal{W}_a\subset\cY_a'$
separates the points of $\cY_a$ hence $\mathcal{W}_a=\cY_a'=\cb_a(L)$
which proves the surjectivity of the map $f\mapsto \vec f(a)$. This
proves (i) and (iv) completely and also one implication in (iii). It
remains to prove that $\vec f_a, \vec g_a$ are linearly dependent if
$W_a(f,g)=0$.

We prove this but with a different notation which allows us to use what
we have already shown. Let $f$ such that $\vec f_a\neq0$. Then
$\vec f_a$ is part of a basis in $\mathcal{W}_a=\cb_a(L)$, hence
there is $g$ such that $\{\vec f_a, \vec g_a\}$ is a basis in
$\cb_a(L)$. Then $W_a(f,g)\neq0$ and we have \eqref{eq:kodafgh}. Thus
if $W_a(h,f)=0$ then $\vec h_a=cW_a(g,h)\vec f_a$, so
$\vec h_a,\vec f_a$ are linearly dependent.
\end{proof}

The space $\cB_a$ is naturally a symplectic space.  In fact, if
$\cB_a$ is nontrivial, then we can find $k,h$ with $W_a(k,h)\neq0$.
By the Kodaira identity,
\begin{align}\notag
  W_a(f,g)&=\frac{-W_a(f,k)W_a(g,h)+W_a(f,h)W_a(g,k)}{W_a(h,k)}\\
  &=\frac{-\vec f_a(k)\vec g_a(h)-\vec f_a(h)\vec g_a(k)}{W_a(h,k)}.
\end{align}
Thus if we set for $\phi,\psi\in\cB_a$ with $\vec f_a=\phi$,
$\vec g_a=\psi$,
\begin{equation}
  \lbra\phi|\psi\rbra_a:=W_a(f,g),\label{sigma}\end{equation}
  then $\lbra\cdot|\cdot\rbra_a$ is a well defined symplectic form on $\cB_a$.
Moreover, $f\mapsto \vec f_a$ maps the form $W_a$ onto $\lbra\cdot|\cdot\rbra_a$.  
If $\lbra\phi|\psi\rbra_a=1$, then by the Kodaira identity
\begin{equation}
  W(h,k)=\phi(h)\psi(k)-\psi(h)\phi(k).
\end{equation}

In the literature boundary functionals are usually described using the
notion of {\em boundary triplet}.  Let us make a comment on this
concept.  Suppose, for definiteness, that $\nu_a=\nu_b=2$.  Choose
bases
\begin{equation}
  \phi_a,\psi_a, \text{ of } \cB_a \text{ and } \phi_b,\psi_b \text{ of } \cB_b\label{basis}
\end{equation}
such that $\lbra\phi_a|\psi_a\rbra_a=\lbra\phi_b|\psi_b\rbra_b=1$.
We have the maps
\begin{align}
  \cD(L_{\max})\ni f\mapsto    \phi
(  f)&:=\big(\phi_a(f),\phi_b(f)\big)\in\bbC^2;\\ 
  \cD(L_{\max})\ni f\mapsto       \psi
 ( f)&:=\big(\psi_a(f),\psi_b(f)\big)\in\bbC^2.
\end{align}
Then we can rewrite Green's formula (\ref{eq:lagrange1}) as
\begin{equation}        \label{eq:lagrange1a}
  \langle L_{\max}f|g\rangle-
  \langle f|L_{\max}g\rangle=
  \langle \psi( f)|\phi( g)\rangle
  - \langle \phi( f)|\psi( g)\rangle.
\end{equation}
The triplet $(\bbC^2,\phi,\psi)$ is often called in the
literature a {\em boundary triplet}, see e.g. \cite{BBMNP} and
references therein. It can be used to characterize operators in
between $L_{\min}$ and $L_{\max}$.

Thus a boundary triplet is essentially a choice of a basis  (\ref{basis}) in the space of boundary functionals. Such a
choice is often natural: in particular this is the case of regular
boundary conditions, see (\ref{pom1}), (\ref{pom2}). In our paper we
consider rather general potentials for which there may be no natural
choice for (\ref{basis}). Therefore, we do not use the boundary
triplet formalism.

\subsection{Classification of endpoints and of realizations of
  \texorpdfstring{$L$}{L}}

The next fact is a consequence of Theorem \ref{th:bvf}.  One may
think of the assertion ``$\nu_a(L)$ can only take the values 0 or 2''
as a version of Weyl's dichotomy, cf.\ \S\ref{ss:weyl}.

\begin{theorem}\label{th:dimbv}
  $\nu_a(L)$ can be $0$ or $2$: we have
  $\nu_a(L)=0 \Leftrightarrow W_a=0$ and
  $\nu_a(L)=2\Leftrightarrow W_a\neq0$. Similarly for $\nu_b(L)$,
  hence $\nu(L)\in\{0,2,4\}$.
\end{theorem}

\begin{remark}\label{re:weyl}{\rm According to the terminology in
    \cite{DS2}, we might say that $L$ has \emph{no boundary values at
      $a$} if $\nu_a(L)=0$ and that $L$ has \emph{two boundary values
      at $a$} if $\nu_a(L)=2$.  }
\end{remark}

\begin{example} \label{ex:bcreg}

  {\rm If $L$ is regular at the endpoint $a$ then $\nu_a(L)=2$.  It is
    clear that $f\mapsto f(a)$ and $f\mapsto f'(a)$ are linearly
    independent and Theorem \ref{th:dimbv} implies that they form a
    basis in $\cb_a(L)$.}
\end{example}

\begin{example} \label{ex:bcsemireg}
  {\rm If $L$ is semiregular at $a$ then we also have $\nu_a(L)=2$.
    Indeed, $\dim\cU_a(\lambda)=2$ (Prop.\ \ref{prop:pro}) and this
    implies $\nu_a(L)=2$ by (the easy part of) Theorem
    \ref{th:proofconj}.  We also have the distinguished boundary
    functional $f\mapsto f(a)$, as shown in Prop. \ref{prop:pro}
    (2). If $u$ is the solution of $Lu=0$ satisfying $u(a)=0$,
    $u'(0)=1$, whose existence is guaranteed by Prop.  \ref{prop:pro}
    (3), then this functional coincides with $\vec u_a$.  However, in
    general, we do not have another, linearly independent
    distinguished boundary functional.}  
\end{example}

As a consequence of Theorem \ref{th:dimbv} we get the following
classification of 1d Schr\"odinger operators in terms of the boundary
functionals.

\begin{enumerate}
\item $\nu_a(L)=\nu_b(L)=0$. This is
  equivalent to $L_\m=L_a=L_b=L_\M$.
\item $\nu_a(L)=0,$
   $\nu_b(L)=2,$
 Then $\cD(L_\m)$ is a  subspace of codimension $2$ in
  $\cD(L_\M)$. This is equivalent to $L_a=L_\M$, and to $L_\m=L_b$.
\item
$\nu_a(L)=2,$
   $\nu_b(L)=0.$
 Then $\cD(L_\m)$ is a  subspace of codimension $2$ in
  $\cD(L_\M)$.  This is equivalent to $L_b=L_\M$, and to $L_\m=L_a$.
\item $\nu_a(L)=\nu_b(L)=2$.
  Then $\cD(L_\m)$ is a  subspace of codimension $4$ in
  $\cD(L_\M)$.
\end{enumerate} 

In case (2) the operators $L_\bullet$ with
$L_\m\subsetneq L_\bullet\subsetneq L_\M$ are defined by nonzero
boundary value functionals $\phi$ at $a$:
$\cD(L_\bullet)=\{f\in\cD(L_\M) \mid \phi(f)=0\}$.  Similarly in case
(3).

Consider now the case (4).  The domain of nontrivial realizations
$L_\bullet$ could be then of codimension $1,2$, or $3$ in $\cD(L_\M)$.
We will see that realizations of codimension $2$ are the most
important.

Each realization of $L$ extending $L_\m$ is defined by a subspace
$\cC_\bullet\subset\cB_a\oplus\cB_b$.
\begin{equation}
  \cD(L_\bullet):=\{f\in\cD(L_\M)\ \mid
  \phi(f)=0,\quad \phi\in\cC_\bullet\}\end{equation}
The space $\cC_\bullet$ is called the {\em space of boundary
  conditions for $L_\bullet$}. The dimension of $\cC_\bullet$
coincides with the codimension of $\cD(L_\bullet)$ in $\cD(L_\M)$. 

\begin{definition} \label{separated}
  We say that the boundary conditions $\cC_\bullet$ are {\em
    separated} if 
\begin{equation} 
  \cC_\bullet=\cC_\bullet\cap\cB_a\,\oplus\,\cC_\bullet\cap\cB_b.
\end{equation}
\end{definition}

For instance, $L_a$ and $L_b$ are given by separated boundary
conditions $\cB_a$, resp. $\cB_b$.

\begin{definition} \label{def:1}
  Let $\phi\in\cB_a$ and $\psi\in\cB_b$.
  Then the realization of $L$ with the boundary condition
  $\C\phi\oplus\C\psi$ will be denoted $L_{\phi,\psi}$.
\end{definition}

Clearly, $L_{\phi,\psi}$ has separated boundary conditions and depends
only on the complex lines determined by $\phi$ and $\psi$. More
explicitly,
\[
  \cD(L_{\phi,\psi})=\{f\in \cD(L_\M)\mid \phi(f)=\psi(f)=0\}.
\]
Recall that if $\phi\neq0$ then
$\phi(f)=0 \Leftrightarrow\exists c(f)\in\C$ such that
$\vec{f}_a=c(f)\phi$.  We abbreviate $L_\phi=L_{\phi,0}$ if $\psi=0$
and define similarly $L_\psi$ if $\phi=0$.  Thus $L_\phi$ involves no
boundary condition at $b$:
\begin{equation}\label{eq:Lphi}
\cD(L_{\phi})=\{f\in \cD(L_\M)\mid \phi(f)=0\} =\{f\in \cD(L_\M)\mid
  \exists c(f) \text{ such that } \vec{f}_a=c(f)\phi\}
\end{equation}
where the second equality holds if $\phi\neq0$. Note that
$L_{0,0}=L_\M$.

\subsection{Properties of boundary functionals}

The next proposition is a version of \cite[XIII.2.27]{DS2} in our
context.

\begin{proposition}\label{co:bvf}
  If $\phi\in\cb_a(L)$ then there are continuous functions
  $\alpha,\beta:\,]a,b[\to\C$ such that
  \begin{equation*}
    \phi(f)=\lim_{x\to a}\big(\alpha(x)f(x)+\beta(x)f'(x) \big)
    \quad\forall f\in\cD(L_\M). 
  \end{equation*}
  Reciprocally, if $\alpha,\beta$ are complex functions on $]a,b[$ and
  $\lim_{x\to a}\big(\alpha(x)f(x)+\beta(x)f'(x) \big)=:\phi(f)$
  exists $\forall f\in\cD(L_\M)$, then $\phi\in\cb_a(L)$.
\end{proposition}

\begin{proof}
 The first assertion follows from Theorem \ref{th:bvf}-(i)
and relations \eqref{eq:bvflim}, \eqref{eq:dirac} while the second one
is a consequence of Banach-Steinhaus theorem.
\end{proof}

Recall that for $d\in[a,b]$ the symbol $L^{a,d}$ denotes the operator
$-\partial^2+V$ on the interval $]a,d[$.

\begin{lemma}\label{lm:bvf}
  Let $d\in]a,b[$. Then
  \begin{equation}\label{eq:bvfd}
\dim\cb_a(L)=\dim\cb_a(L^{a,d}) .
  \end{equation}
\end{lemma}

\begin{proof}
  
  Since $d$ is a regular endpoint for $L^d$, the maximal operator
  $L^{a,d}_{\max}$ associated to $L^{a,d}$ has the property
  $\cD(L^{a,d}_{\max})\subset C^1]a,d]$. Thus the restriction map
  $R:f\mapsto f\big|_{]a,d[}$ is a surjective map
  $\cD(L_{\max})\to \cD(L^{a,d}_{\max})$ such that
  $R\cD(L_a)=\cD(L^{a,d}_a)$. If $\phi$ is a boundary value functional
  at $a$ for $L^{a,d}$ then clearly $\phi\circ R$ is a boundary value
  functional at $a$ for $L^{a,d}$ and the map $\phi\mapsto\phi\circ R$
  is a bijective map $\cb_a(L)\to\cb_a(L^{a,d})$.
\end{proof}

We note that the space $\cb(L)$ and its subspaces $\cb_a(L),\cb_b(L)$
depend on $L$ only through the domains $\cD(L_\M)$ and $\cD(L_\m)$. So
in order to compute them one can sometimes change the
potential and consider an
operator $L^U:=-\partial^2+U$ instead of $L:=-\partial^2+V$.
This is especially useful if $U$ is real: for example, $U$ could be
the real part of $V$, if its imaginary part is bounded.

\begin{proposition}\label{pr:relbdd}
  Let $U:\,]a,b[\,\to\C$ measurable such that
  $\|(U-V)f\|\leq \alpha\|Lf\|+\beta\|f\|$ for some real numbers
  $\alpha,\beta$ with $\alpha<1$ and all $f\in\cD(L_\M)$.  Then
  $\cD(L_\M)=\cD(L^U_\M)$ and $\cD(L_\m)=\cD(L^U_\m)$. Hence
  $\cb(L)=\cb(L^U)$ and
  \begin{equation}\label{eq:relbdd}
    \nu_a(L)=\nu_a(L^U), \quad \nu_b(L)=\nu_b(L^U) .
  \end{equation}
\end{proposition}

\begin{proof}
 We have
  \[
    (1-\alpha)\|Lf\|-\beta \|f\|\leq\|L^Uf\|
    \leq (1+\alpha)\|Lf\|+\beta\|f\|
\]
so the norms $\|\cdot\|_L$ and $\|\cdot\|_{L^U}$ are equivalent. Then
we use \eqref{eq:polarB}.
\end{proof}

\subsection{Infinite endpoints}

Suppose now that our interval is right-infinite.
We will show that if the potential stays bounded in average at infinity, then all elements of the maximal domain converge to zero at $\infty$ together with their derivative, which obviously implies that their Wronskian converges to zero.

\begin{proposition}
 Suppose that $b=\infty$ and
  \begin{equation}
    \limsup_{c\to\infty}\int_{c}^{c+1}|V(x)|\d x<\infty.
  \end{equation}
  Then 
  \begin{equation}
f\in\mathcal{D}(L_{\max})\quad\Rightarrow\quad    \lim_{x\to\infty}f(x)=0,\     \lim_{x\to\infty}f'(x)=0.
\label{mana}  \end{equation}
  Hence 
$\nu_b=0$.
  \label{thm010}
\end{proposition}

Of course, an analogous statement is true for $a=-\infty$ on
left-infinite intervals.

\noindent{\em Proof of Prop. \ref{thm010}.}
Let $\nu<\nu_0$ and let $J_n:=[a+n\nu,a+(n+1)\nu]$.
Then, using first (\ref{eq:ode-main}) and then the Schwarz inequality, we obtain
\begin{align*}
  \|f\|_{L^\infty(J_n)}+\nu\|f'\|_{L^\infty(J_n)}&\leq
  C_1\|Lf\|_{L^1(J_n)}+
    C_2\|f\|_{L^1(J_n)}\\
    &\leq C_1\sqrt\nu\|Lf\|_{L^2(J_n)}+
    C_2\sqrt\nu\|f\|_{L^2(J_n)}\,\underset{n\to\infty}\to\,0.
\end{align*}
This implies \eqref{mana}.
\qed

\section{Solutions square integrable near endpoints}
\protect\setcounter{equation}{0}

\subsection{Spaces \texorpdfstring{$\cU_a(\lambda)$}{Ua(lambda)} and
 \texorpdfstring{$\cU_b(\lambda)$}{Ub(lambda)}}

In this section we will show that one can compute the boundary indices
with the help of eigenfunctions of the operator $L$ which are square
integrable around a given endpoint.

\begin{definition} \label{defempty} If $\lambda\in\C$ then
  $\mathcal{U}_a(\lambda)$ is the set of $f\in AC^1]a,b[$ such that
  $(L-\lambda)f=0$ and $f$ is $L^2$ on $]a,d[$ for some, hence for all
  $d$ such that $a<d<b$. Similarly we define $\mathcal{U}_b(\lambda)$.
\end{definition}

\begin{proposition} \label{endpoint} If $a$ is a semiregular endpoint
  for $L$, then $\dim\cU_a(\lambda)=2$ for all $\lambda\in\C$.
  Besides, if $a$ is regular, we can choose $u,v\in\Ker( L-\lambda)$
  such that
  \begin{align}
    u(a)=1,&\quad u'(a)=0,\label{regu1}\\
    v(a)=0,&\quad v'(a)=1.\label{regu2}
  \end{align}
  Similarly for $b$.
\end{proposition}

\proof We apply Prop. \ref{prop:pro}.    \qed

\subsection{Two-dimensional
  \texorpdfstring{$\cU_a(\lambda)$}{Ua(lambda)}} \label{ss:weyl}

The next proposition contains the main technical fact about the
dimensions of the $\cU_a(\lambda)$.

\begin{proposition}\label{pr:perturb}
  Assume that all the solutions of $Lf=0$ are square integrable near
  $a$. If $f\in C^1]a,b[$ and $|Lf| \leq B|f|$ for some $B>0$, then
  $f$ is square integrable near $a$. In particular, if
  $U\in L^\infty]a,b[$ then all the solutions of $(L+U)f=0$ are square
  integrable near $a$.
\end{proposition}
    
\begin{proof}
  We may clearly assume that $b$ is a regular endpoint and
  $f\in C^1]a,b]$.  Let $G_\leftarrow$ be the backward Green's
    operator of $L$ (Definition \ref{df:Ga}).  If $Lf=g$, then
  $L(f-G_{\leftarrow} g)=0$. Therefore
\begin{equation}
  f(x)=\alpha u(x)+\beta v(x)+\int_x^b \big(u(x)v(y)-v(x)u(y)\big)
  g(y)\d y,\end{equation}
for some $\alpha,\beta$. Set $A:=\sqrt{|\alpha|^2+|\beta|^2}$ and
$\mu(x):=\sqrt{|u(x)|^2+|v(x)|^2}$. Then 
\begin{equation*}
  |f(x)|\leq A\mu(x)+\mu(x)\int_x^b\mu(y)|g(y)|\d y
        \leq\mu(x)\Big(A+B\int_x^b\mu(y)|f(y)|\d y\Big),
  \end{equation*}
  and the Gronwall Lemma applied to $|f|/\mu$ implies
\begin{equation}\label{pwo}
  |f(x)|\leq A\mu(x)\exp\Big(B\int_x^b\mu^2(y)\d y\Big).
\end{equation}
Clearly the right hand side of \eqref{pwo} is square integrable.
\end{proof}

The above proposition has the following important consequence.

\begin{proposition}\label{paio}
  If $\dim\cU_a(\lambda)=2$ for some $\lambda\in\C$ then
  $\dim\cU_a(\lambda)=2$ for all $\lambda\in\C$. Besides, if this is
  the case, then $\nu_a(L)=2$.
\end{proposition}

\subsection{The kernel of \texorpdfstring{$L_\M$}{Lmax}}

Let us describe the relationship between the dimension of the kernel
of $L_\M-\lambda$ and the dimensions of spaces $\cU_a(\lambda)$ and
$\cU_b(\lambda)$.

The first proposition is a corrolary of Prop. \ref{paio}:
\begin{proposition}
  The following statements are equivalent:
  \begin{enumerate}
  \item $\dim\Ker (L_\M-\lambda)=2$ for some $\lambda\in\C$.
  \item $\dim\Ker (L_\M-\lambda)=2$ for all $\lambda\in\C$.
  \item $\dim\cU_a(\lambda_a)=\dim\cU_b(\lambda_b)=2$ for some
    $\lambda_a,\lambda_b \in\C$.
  \item $\dim\cU_a(\lambda)=\dim\cU_b(\lambda)=2$ for all
    $\lambda\in\C$.
  \end{enumerate}
  Besides, if this is the case, then $\nu_a(L)=\nu_b(L)=2$.
\end{proposition}

The next two propositions are essentially obvious:

\begin{proposition}
  Let $\lambda\in\C$.  We have $\dim\Ker (L_\M-\lambda)=1$ if and only
  if one of the following statements is true:
  \begin{enumerate}
  \item $\dim\cU_a(\lambda)=\dim\cU_b(\lambda)=1$ and
    $\cU_a(\lambda)=\cU_b(\lambda)$.
  \item $\dim\cU_a(\lambda)=2$ and $\dim\cU_b(\lambda)=1$.
  \item $\dim\cU_a(\lambda)=1$ and $\dim\cU_b(\lambda)=2$.
  \end{enumerate}
\end{proposition}

\begin{proposition}
  Let $\lambda\in\C$ and $\cU_a(\lambda)\neq\{0\}$,
  $\cU_b(\lambda)\neq\{0\}$.  Then $\dim\Ker(L_\M-\lambda)=0$ if and
  only if $\dim\cU_a(\lambda)=\dim\cU_b(\lambda)=1$ and
  $\cU_a(\lambda)\neq\cU_b(\lambda)$.
\end{proposition}

\subsection{First order ODE's}

We will need some properties of vector-valued ordinary differential
equations.  We will denote by $B(\C^n)$ the space of $n\times n$
matrices.
    
The following statement can be proven by the same methods as Prop.
\ref{mojo}. Clearly, in the following proposition $\C^n$ can be easily
replaced by an arbitrary Banach space.

\begin{proposition} \label{cauchy} Let $u_0\in\C^n$ and the function
  $[a,b[\ni x\mapsto A(x)\in B(\C^n)$ be in
  $L_\loc^1\big([a,b[,B(\C^n)\big)$. Then there exists a unique
  solution in $AC\big([a,b[,\C^n\big)$ of the following Cauchy
  problem:
  \begin{equation}
    \partial_xu(x)=A(x)u(x),\quad u(a)=u_0.
  \end{equation}
  In particular, the dimension of the space of solutions of
  $\partial_xu(x)=A(x)u(x)$ is $n$.
        
  If $A \in L^1\big(]a,b[,B(\C^n)\big)$ then
  $u\in AC\big([a,b],\C^n\big)$, hence $u$ is continuous up to $b$.
\end{proposition}

The following theorem is much more interesting. It is borrowed from
Atkinson \cite[Th.\ 9.11.2]{At}. Note that in this theorem the finite
dimensionality of the space $\C^n$ seems essential.

\begin{theorem}\label{atkinson}
  Suppose that $A,B$ are functions $[a,b[\,\to B(\C^n)$ belonging to
  $L_\loc^1([a,b[,B(\C^n))$ satisfying $A(x)=A^*(x)\geq0$,
  $B(x)=B^*(x)$. Let $J$ be an invertible matrix satisfying $J^*=-J$
  and such that $J^{-1}A(x)$ is real. If for some $\lambda\in\C$ all
  solutions of
  \begin{equation}
    J\partial_x\phi(x)=\lambda A(x)\phi(x)+B(x)\phi(x) \label{at1}
  \end{equation}
  satisfy
  \begin{equation}
    \int_a^b\big(\phi(x)|A(x)\phi(x)\big)\d x<\infty\label{at2}
  \end{equation}
  then for all $\lambda\in\C$ all solutions of (\ref{at1}) satisfy
  (\ref{at2}).
\end{theorem}

\proof For $\lambda\in\C$, let
$[a,b[\ni x\mapsto Y_\lambda(x)\in B(\C^n)$ be the solution of
\begin{equation}
  J\partial_xY_\lambda(x)=\lambda
  A(x)Y_\lambda(x)+B(x)Y_\lambda(x),\quad Y_\lambda(a)=\one.
\end{equation}                    
Then the theorem is equivalent to the following statement: if for some
$\lambda\in\C$ we have
\begin{equation}
  \int_a^b\Tr Y_\lambda^*(x)A(x)Y_\lambda(x)\d x<\infty,\label{lamb}
\end{equation}
then for all $\lambda\in \C$ we have (\ref{lamb}). We are going to
prove this in the following.
 
First note that
\begin{equation}
  \bar{\Tr J^{-1}B(x)}=\Tr(J^{-1}B(x))^*=
  -\Tr B(x)J^{-1}=-\Tr J^{-1}B(x).
\end{equation}
Therefore, $\Tr J^{-1}B(x)\in\i\R$.  By the same argument
$\Tr J^{-1}A(x)\in\i\R$.  But $\Tr J^{-1}A(x)$ is real.  Hence
$\Tr J^{-1}A(x)=0$, and so for arbitrary $\lambda\in\C$ we have
$\Tr\big(\lambda A(x)+B(x)\big)\in\i\R$. Therefore,
\begin{equation}
  \partial_x\det Y_\lambda(x)
  =\Tr\big(\lambda A(x)+B(x)\big)\det Y_\lambda(x)
\end{equation}
implies
\begin{equation}
  | \det Y_\lambda(x)|= | \det Y_\lambda(a)|=1.\label{lamb1}
\end{equation}
Therefore, $Y_\lambda(x)$ is invertible for all $x\in]a,b]$.

Now let $\mu\in\C$ and assume that (\ref{lamb}) holds for
$\lambda=\mu$.  We have
\begin{equation}
  \partial_xY_\mu^*(x)JY_\mu(x)=
  (\mu-\bar\mu)Y_\mu^*(x)A(x)Y_\mu(x),\quad 
  Y_\mu^*(a)JY_\mu(a)=J.
\end{equation}
Hence
\begin{equation}
  Y_\mu^*(x)JY_\mu(x)=
  J+ (\mu-\bar\mu)\int_a^xY_\mu^*(y)A(y)Y_\mu(y)\d y.
\end{equation}
Using (\ref{lamb}) we see that $ Y_\mu^*(x)JY_\mu(x)$ is bounded
uniformly in $x\in[a,b[$.  By (\ref{lamb1}), its inverse is also
bounded uniformly in $x\in[a,b[$.  (Here we use the finiteness of the
dimension of $\C^n$!)

Set $Z_\lambda(x):=Y_\mu^{-1}(x)Y_\lambda(x).$ We have
\begin{align}
  \partial_xZ_\lambda=(\lambda-\mu)Y_\mu^{-1}J^{-1}AY_\lambda
 =(\lambda-\mu)\big(Y^*_\mu J Y_\mu)^{-1}Y_\mu^*AY_\mu Z_\lambda.
\end{align}
We have proven that $\big(Y^*_\mu J Y_\mu)^{-1}$ is uniformly
bounded. By (\ref{lamb}) the norm of $Y_\mu^*AY_\mu $ is in
$L^1]a,b[$. Hence by the second part of Prop. \ref{cauchy},
$\|Z_\lambda\|$ is uniformly bounded on $[a,b[$.  Now, by using
\begin{equation}
  Y_\lambda^*(x)A(x)Y_\lambda(x)=
  Z_\lambda^*(x)Y_\mu^*(x)A(x)Y_\mu(x)Z_\lambda(x)
\end{equation}
we see that (\ref{lamb}) for $\lambda=\mu$ implies (\ref{lamb}) for
all $\lambda$. \qed

\begin{remark}
  Set
  \begin{equation}
    J:=\begin{pmatrix} 0 & -1 \\ 1 & 0\end{pmatrix},\quad
    A:=\begin{pmatrix} 1 & 0 \\ 0 & 0\end{pmatrix},\quad
    B(x):=\begin{pmatrix} -V(x) & 0 \\ 0 & 1\end{pmatrix}.                 
  \end{equation}
  Then $Lf=\lambda f$ can be rewritten as (\ref{at1}), that is,
  \begin{equation}
    J\partial_x\phi(x)=\lambda A(x)\phi(x)+B(x)\phi(x),\label{at3}
  \end{equation}
  with $\phi=\begin{pmatrix} f \\ f'\end{pmatrix}$. Moreover
  \begin{equation}
    \int_a^b\big(\phi(x)|A(x)\phi(x)\big)=\int|f(x)|^2\d x,\label{at4}
  \end{equation}
  hence the condition (\ref{at2}) means that $f\in L^2]a,b[$.  Note
  also that the conditions of Thm \ref{atkinson} on $J$, $A$ and $B$
  are satisfied. Thm \ref{atkinson} therefore implies that if all
  solutions of $Lf=\lambda f$ are square integrable for one $\lambda$,
  they are square integrable for all $\lambda$. We thus obtain an
  alternative proof of Prop. \ref{paio}.
\end{remark}

\subsection{Von Neumann decomposition} \label{ss:VND}

Von Neumann's theory for the classification of self-adjoint extensions
of a Hermitian operators is well known, cf.\ \cite{S,DS2}.  In the
present subsection we will investigate how to adapt it to the case of
complex potentials.

First recall that $\cD(L_{\max})$ has a Hilbert space structure
inherited from its graph, which is a closed subspace of
$L^2]a,b[\,\oplus L^2]a,b[\,$, namely
\begin{equation}\label{eq:Lsproduct}
  (f|g)_{L}:=(Lf|Lg)+(f|g)=\langle \bar L\bar f|Lg\rangle +\langle\bar
  f|g\rangle.
\end{equation}
Therefore, by the Riesz representation theorem $\cD(\bar L_{\max})$
can be identified with the dual of $\cD(L_{\max})$:
\begin{equation} \label{ident}
  \cD(\bar L_{\max})\ni f\mapsto
  (\bar f| \cdot)_L \in \cD(L_{\max})'.
\end{equation}
Hence the space of boundary functionals $\B(L)\subset \cD(L_{\max})'$
can be viewed as a subspace of $ \cD(\bar L_{\max})$.

The following lemma follows from Prop. \ref{dire1} (1):

\begin{lemma}\label{lm:vonn}
 With the identification (\ref{ident}), there is a canonical linear isomorphism
  \begin{equation}\label{eq:identif}
    \cb(L) \simeq \Ker(L_\M\bar L_\M+1)
    =\{f\in\cD(\bar{L}_\M)\mid\bar{L}f\in\cD(L_\M) 
    \text{ and } L\bar{L}f+f=0\} .
  \end{equation}
\end{lemma}

Von Neumann's formalism is particularly efficient for real potentials
and gives more precise results than in the complex case, so for
completeness we begin with some comments on the real case.  Then we
explore what can be done for arbitrary complex potentials. The
differences between the real and complex case are significative, the
difficulties being related to the fact that in the complex case there
is no simple relation between the (geometric) limit point/circle
method and the dimension of the spaces $\cU_a(\lambda)$, cf.\
Subsect. \ref{ss:lplc}.

If $V$ is real then $\bar{L}=L$, $L_\m$ is Hermitian, and
$L_\M=L_\m^*$, hence
\begin{equation}\label{eq:identifS}
  \cb(L) \simeq  \Ker(L_\M^2+1) .
\end{equation}
Then by using the relation $L^2+1=(L-\i)(L+\i)$ it is easy to prove
that
\begin{equation}\label{eq:aBb}
  \cb(L)\simeq \Ker(L_\M-\i)+\Ker(L_\M+\i) .
\end{equation}
The last sum is obviously algebraically direct but also orthogonal for
the scalar product \eqref{eq:Lsproduct} hence we have an orthogonal
direct sum decomposition
\begin{equation}\label{eq:orthdec}
  \cD(L_\M)=\cD(L_\m)\oplus\cb(L)
  =\cD(L_\m)\oplus \Ker(L_\M-\i)\oplus\Ker(L_\M+\i) .
\end{equation}
The map $f\mapsto\bar{f}$ is a real linear isomorphism of
$\Ker(L_\M-\i)$ onto $\Ker(L_\M+\i)$ hence these spaces have equal
dimension $\leq2$ and so $\dim
\cb(L)=2\dim\Ker(L_\M-\i)\in\{0,2,4\}$. Of course, we have already
proved this in a much simpler way, but \eqref{eq:orthdec} also gives
via a simple argument the following: if $V$ is real then \label{p:real}
\begin{enumerate}
\item[{\rm(1)}] $\nu_a(L)=0 \Leftrightarrow \dim\cU_a(\lambda)=1$
  $\forall\lambda\in\C\setminus\R$;
\item[{\rm(2)}] $\nu_a(L)=2 \Leftrightarrow \dim\cU_a(\lambda)=2$
  $\forall\lambda\in\C$ .
\end{enumerate}
The simplicity of the treatment in the real case is due to the
possibility of working with \eqref{eq:aBb}, which involves only the
second order operators $L_\M\pm \i$, instead of \eqref{eq:identifS},
which involves the operator $L_\M^2$ of order $4$.  We do not have
such a simplification in the complex case where $L_\M\bar{L}_\M +1$ is
formally a fourth order differential operator with very singular
coefficients, since $V$ is only locally $L^1$.

Let us show how to generalize von Neumann's analysis to the complex
case.  We will follow \cite[Theorem 9.1]{EZ} which in turn is a
consequence of \cite[Theorem 9.11.2]{At}.  The nontrivial part of
Theorem \ref{th:proofconj} is due to D. Race \cite[Theorem 5.4]{Rac}.

We need to study the equation
\begin{equation}\label{eq4order}
  (L\bar{L}+\lambda)f=0.
\end{equation}
More precisely, by a solution of (\ref{eq4order}) we will mean
$f\in AC^1]a,b[$ such that $\bar{L}f\in AC^1]a,b[$ and
$L(\bar{L}f)+\lambda f=0$ holds.

Let us rewrite (\ref{eq4order}) as a 2nd order system of 2
equations. To this end, introduce
\begin{equation}
  Q:=\begin{pmatrix} 0&1 \\ 1 & 0 \end{pmatrix},\qquad
  \cI:=\begin{pmatrix} 1&0 \\ 0 & 0 \end{pmatrix},\qquad
  W:=
   \begin{pmatrix} 0 & V \\ \bar{V} & -1\end{pmatrix},
\end{equation}
\begin{equation}\label{eq:Lscript}
  \cL :=-Q\partial^2+W
  =\begin{pmatrix} 0 & L\\ \bar{L} & -1\end{pmatrix}.
\end{equation}
Consider the equation
\begin{equation}
  (\cL +\lambda\cI)F=0,\label{eq2order}
\end{equation}
on $AC^1\big(]a,b[,\C^2\big)$.
The equations (\ref{eq4order}) and  (\ref{eq2order}) are equivalent in the following sense:
\begin{lemma}\label{lm:reduction}
  The map $f\mapsto F:=\left(\begin{smallmatrix}f\\
      \bar{L}f \end{smallmatrix}\right)$ is an isomorphism of the
  space of solutions of ${(L\bar{L}+\lambda)f=0}$ onto the space of
  solutions of $(\cL +\lambda\cI)F=0$.
\end{lemma}

\begin{proof}
  \begin{sloppypar}
    It is immediate to see that if $(L\bar{L}f+\lambda) f=0$, then
    $\left(\begin{smallmatrix}f\\
        \bar{L}f \end{smallmatrix}\right)\in AC^1(]a,b[,\C^2)$
    and ${(\cL+\lambda\cI)\left(\begin{smallmatrix}f\\
          \bar{L}f \end{smallmatrix}\right)=0}$.
  \end{sloppypar}
  
  Reciprocally, if $(\cL +\lambda\cI)\left(\begin{smallmatrix} f_1\\
      f_2 \end{smallmatrix}\right)=0$, then $f_1,f_2\in AC^1]a,b[$ and
  $\begin{pmatrix} \lambda f_1 +L f_2 \\ \bar{L}
    f_1-f_2\end{pmatrix}=\begin{pmatrix} 0\\ 0 \end{pmatrix} $.  Thus
  $f_2=\bar{L}f_1$ and $(L \bar{L}+\lambda) f_1=0$.
\end{proof}

We still prefer to transform (\ref{eq2order}) to a 1st order system of
4 equations. To this end we introduce
\[
  J=\begin{pmatrix} 0 & -Q \\ Q & 0 \end{pmatrix},\quad
  A=\begin{pmatrix}\cI & 0\\ 0 & 0 \end{pmatrix}, \quad
  B=\begin{pmatrix} W & 0\\ 0 &- Q \end{pmatrix}
\]
Consider the equation
\begin{equation}\label{space}
  J\partial_x\phi=(\lambda A
  +B)\phi,
\end{equation}
where $\phi\in AC\big(]a,b[,\C^4\big)$.

\begin{lemma} \label{lm:reduction2} 
  The map $F\mapsto \phi:=\begin{pmatrix} F\\ -F' \end{pmatrix}$ is an
  isomorphism of the space of solutions of $(\cL +\lambda\cI)F=0$ onto
  the space of solutions of $J\partial_x\phi=(\lambda A +B)\phi$.
\end{lemma}

\proof It is immediate to see that if $(\cL+\lambda\cI)F=0$, then
  \[
    \begin{pmatrix} F\\ -F' \end{pmatrix}\in AC\big(]a,b[,\C^4\big)
    \quad\text{and}\quad
    (-J\partial_x+\lambda A+B)\begin{pmatrix} F\\
        -F' \end{pmatrix}=0.
  \]
  Reciprocally, if
  $(-J\partial_x+\lambda A+B)\begin{pmatrix} F\\ G \end{pmatrix}=0$,
  then
  \[
    F,G\in AC\big(]a,b[,\C^2\big) \quad\text{and}\quad
\begin{pmatrix}QG'+(W+\lambda\cI)F\\ -QF'+QG \end{pmatrix}=
\begin{pmatrix} 0\\ 0 \end{pmatrix}.
\]
Therefore, $G=F'$, so that $F\in AC^1\big(]a,b[,\C^2\big)$, and
$QF''+(W+\lambda\cI)F=0$.  \qed

\begin{lemma} \label{lm:EverittZettl}
  
  Suppose that $L$ is regular at $b$.  If all solutions of
  $(L\bar{L}+\lambda)f=0$ are square integrable for some $\lambda$,
  then all solutions of $(L\bar{L}+\lambda)f=0$ are square integrable
  for all $\lambda$.

  Moreover, if this is the case, then all solutions of $Lf=0$ are
  square integrable.

\end{lemma}

\proof By Lemmas \ref{lm:reduction} and \ref{lm:reduction2}, instead
of $(L\bar{L}+\lambda)f=0$ we can consider
$J\partial_x\phi=(\lambda A +B)\phi$, and the square integrability of
$f$ is equivalent to the integrability of
$\big(\phi(x)|A\phi(x)\big)= |\phi_1(x)|^2$, since $f=\phi_1$ under
the identification.  Note that $J,A$ are constant $4\times4$ matrices
with $J^*=-J$, $A^*=A$, $J^{-1}A$ is a real matrix, and $B(x)^*=B(x)$
belongs to $L_\loc^1]a,b]$.  Thus the equation (\ref{space}) satisfies
the assumptions of Thm \ref{atkinson}.  The theorem says that if for
some $\lambda$ all solutions $\phi$ of (\ref{space})
$\big(\phi(x)|A\phi(x)\big)$ is integrable, then this is so for any
$\lambda$.  This proves the first statement of the lemma.
  
Now suppose that all solutions of $(L\bar{L}+\lambda)f=0$ belong to
$L^2]a,b[$ for all $\lambda$. In particular, all solutions of
$L\bar{L}f=0$ are square integrable.  Since
$\bar{L}f=0\Rightarrow L\bar{L}f=0$, any solution of $\bar{L}f=0$ is
square integrable. Hence also any solution of $Lf=0$.  \qed

\begin{theorem}\label{th:proofconj}
  The following assertions are equivalent and true:
  \begin{align}
    &\nu_a(L)=0 \Longleftrightarrow
      \dim\cU_a(\lambda)\leq1 \ \forall\lambda\in\C
      \Longleftrightarrow
      \dim\cU_a(\lambda)\leq1 \text{ for some } \lambda\in\C,
      \label{eq:1conj} \\
    &\nu_a(L)=2 \Longleftrightarrow \dim\cU_a(\lambda)=2 \
      \forall\lambda\in\C \Longleftrightarrow
      \dim\cU_a(\lambda)=2 \text{ for some } \lambda\in\C. 
      \label{eq:2conj}
\end{align}
If $V$ is a real function, then
\begin{equation}\label{eq:LPsad}
\nu_a(L)=0 \Longleftrightarrow \dim\cU_a(\lambda)=1
\ \forall\lambda\in\C\setminus\R. 
\end{equation}
\end{theorem}

\begin{proof}
  The equivalences \eqref{eq:1conj} follow from \eqref{eq:2conj} by
  taking into account Theorem \ref{th:dimbv} and the fact that the
  dimension of $\cU_a(\lambda)$ is $\leq1$ if it is not $2$. Thus we
  only have to discuss \eqref{eq:2conj}.  The second equivalence from
  \eqref{eq:2conj} is a consequence of Proposition \ref{paio}.
  
  It is easy to see that $\nu_a(L)=2$ if $\dim\cU_a(\lambda)=2$ for
  some complex $\lambda$.  Indeed, let $u,v$ be solutions of the
  equation $(L-\lambda)f=0$ such that $W(u,v)=1$. Then if all the
  solutions of $(L-\lambda)f=0$ are square integrable near $a$, we get
  $W_a(u,v)=1$, hence $W_a\neq0$, so that $\nu_a(L)=2$.

  In what follows we consider the nontrivial part of the theorem: we
  assume $\nu_a(L)=2$ and show that $\dim\cU_a(0)=2$. Clearly we may
  assume that $b$ is a regular end, if not we replace it by any number
  between $a$ and $b$. Then $\nu(L)=2\Leftrightarrow\nu_a(L)=0$ and
  $\nu(L)=4\Leftrightarrow\nu_a(L)=2$ so we have to show that
  $\nu(L)=4\Rightarrow \dim\Ker (L_\M)=2$.  Since $\nu(L)=\dim \cb(L)$
  and $\Ker(L_\M\bar{L}_\M+1)\simeq\cb(L)$ by \eqref{eq:identif}, it
  suffices to prove
  \begin{equation}\label{eq:dim10}
    \dim \Ker(L_\M\bar{L}_\M+1)=4 \Rightarrow \dim\Ker (L_\M)=2  .
  \end{equation}
  By Prop. \ref{cauchy}, the space of solutions of the 1st order
  system (\ref{space}) is 4-dimensional. Therefore, Lemmas
  \ref{lm:reduction} and \ref{lm:reduction2} imply that the space of
  solutions of $L\bar{L}f+\lambda f=0$ is 4-dimensional.  Hence
  $\dim \Ker(L_\M\bar{L}_\M+1)=4$ implies that all solutions of
  ${(L\bar{L}+1)f=0}$ are square integrable.  Now by Lemma
  \ref{lm:EverittZettl} applied to $\lambda=1$, all solutions of
  $Lf=0$ are square integrable.
\end{proof}

\section{Spectrum and Green's operators}\label{sect-spectrum}
\protect\setcounter{equation}{0}

\subsection{Integral kernel of Green's operators}

Recall that in Def.  \ref{rightinverse} we introduced the concept of a
{\em right inverse of a closed operator}.  In the context of 1d
Schr\"odinger operators right inverses of $L_\M$ will be called {\em
  $L^2$ Green's operators}. Thus $G_\bullet$ is a $L^2$ Green's
operator if it is bounded, $\Ran (G_\bullet)\subset \cD(L_\M)$ and
$L_\M G_\bullet =\one$.

Let $G_\bullet$ be a Green's operator in the sense of Definition
\ref{green1}.  Clearly, $L_\mathrm{c}^2]a,b[$ is contained in
$L_\mathrm{c}^1]a,b[$. Besides, $L_{\mathrm{c}}^2]a,b[$ is dense in
$L^2]a,b[$. Therefore, if the restriction of $G_\bullet$ to
$L_{\mathrm{c}}^2]a,b[$ is bounded, then it has a unique extension to
a bounded operator on $L^2]a,b[$. This extension, which by
Prop. \ref{piuy6} is a $L^2$ Green's operator, will be denoted by the
same symbol $G_\bullet$.

Note that the pair $L_\M,L_\m$ satisfies $L_\m=L_\M^\#\subset L_\M$,
which are precisely the properties discussed in Subsect
\ref{ss:L2GII}.  Recall from that subsection that $L^2$ Green's
operators whose inverse contains $L_\m$ correspond to realizations of
$L$ that are between $L_\m$ and $L_\M$. The following proposition is
devoted to properties of such Green's operators.  

Recall that for any $x\in]a,b[$ we denote by $L^{a,x}$,
resp. $L^{x,b}$ the restriction of $L$ to $L^2]a,x[$,
resp. $L^2]x,b[$.  We also can define $L_\M^{a,x}$ and $L_\M^{x,b}$,
etc. Note that $x$ is a regular point of both $L^{a,x}$ and $L^{x,b}$
($V$ is integrable on a neighbourhood of $x$).

\begin{proposition}\label{pr:gkernel}
  Suppose that $L_\m\subset L_\bullet\subset L_\M$, $L_\bullet$ is
  invertible and $G_\bullet:=L_\bullet^{-1}$.  Then $G_\bullet$ is an
  integral operator whose integral kernel
  \[]a,b[\times]a,b[\ni(x,y)\mapsto G_\bullet(x,y)\in\C\]
  is a function separately continuous in $x$ and $y$ which has the
  following properties:\\ 
  (1) for each $a<x<b$ the function $G_\bullet(x,\cdot)$ restricted to
  $]a,x[$, resp. $]x,b[$ belongs to $\cD(L_\M^{a,x})$, resp.
  $\cD(L_\M^{x,b})$ and satisfies $LG_\bullet(x,\cdot)=0$ outside
  $x$. Besides, $G_\bullet(x,\cdot)$ and its derivative have limits at
  $x$ from the left and the right satisfying
  \begin{align*}
    G_\bullet(x,x-0)-G_\bullet(x,x+0)&=0,\\
    \partial_2G_\bullet(x,x-0)-\partial_2G_\bullet(x,x+0)&=1;
  \end{align*}
  (2) for each $a<y<b$ the function $G_\bullet(\cdot,y)$ restricted to
  $]a,y[$, resp. $]y,b[$ belongs to $\cD(L_\M^{a,y})$, resp.
  $\cD(L_\M^{y,b})$ and satisfies $LG_\bullet(\cdot,y)=0$ outside
  $y$. Besides, $G_\bullet(\cdot,y)$ and its derivative have limits at
  $y$ from the left and the right satisfying
  \begin{align*}
    G_\bullet(y-0,y)-G_\bullet(y+0,y)&=0,\\
    \partial_1G_\bullet(y-0,y)-\partial_1G_\bullet(y+0,y)&=1;
  \end{align*}
\end{proposition}

\begin{proof}
  We shall use ideas from the proof of Lemma 4 p.\ 1315 in \cite{DS2}.
  $G_\bullet$ is a continuous linear map
  $G_\bullet:L^2]a,b[ \to \cD(L_{\max})$ and for each $x\in]a,b[$ we
  have a continuous linear form $\varepsilon_x:f\mapsto f(x)$ on
  $\cD(L_{\max})$, hence we get a continuous linear form
  $\varepsilon_x\circ G_\bullet:L^2]a,b[\to\C$. Thus for each
  $x\in]a,b[$ there exists a unique $\phi_x\in L^2]a,b[$ such that
\[
  (G_\bullet f)(x)=\int_a^b \phi_x(y)f(y) \d y,
  \quad \forall f\in L^2]a,b[\,   .
\]
We get a map $\phi:\,]a,b[\to L^2]a,b[$ which is continuous, and even
locally Lipschitz, because if $J\subset\,]a,b[$ is compact and
$x,y\in J$ then
\begin{align*}
  \left|\int_a^b (\phi_x(z)-\phi_y(z))f(z) \d z \right|
  & = |(G_\bullet f)(x)-(G_\bullet f)(y)|
  \leq \|(G_\bullet f)'\|_{L^\infty(J)}|x-y| \\
  &  \leq C_1 \|G_\bullet f\|_{\cD(L_{\max})}|x-y| 
  \leq C_2 \|f\| |x-y|,
\end{align*}
hence $\|\phi_x-\phi_y\|\leq C_2|x-y|$. By taking $f=L_\bullet g$,
$ g\in\cD(L_\bullet)$, we get
\begin{equation}\label{eq:LG}
  g(x)=\int_a^b\phi_x(y)(L_\bullet g)(y) \d y.
\end{equation}
Set $\phi_x^a:=\phi_x\big|_{]a,x[}$ and $\phi_x^b:=\phi_x\big|_{]x,b[}$.
 (\ref{eq:LG}) can be rewritten as
\begin{equation}\label{eq:LG1}
  g(x)=\int_a^x\phi_x^a(y)(L_\bullet g)(y) \d y
+\int_x^b\phi_x^b(y)(L_\bullet g)(y) \d y.
\end{equation}
Since $G_\bullet^\#$ is also an $L^2$ Green's operator, we have
$L_{\min}\subset L_\bullet\subset L_{\max}$.  Assuming that
$g\in\cD(L_{\min})$ and $g(y)=0$ in a neighborhood of $x$, we can
rewrite (\ref{eq:LG1}) as
\begin{equation}\label{eq:LG2}
  0=\int_a^x\phi_x^a(y)(L_\m^{a,x} g)(y) \d y
+\int_x^b\phi_x^b(y)(L_\m^{a,x} g)(y) \d y.
\end{equation}
Such functions $g$ are dense in
$\cD(L_\m^{a,x})\oplus\cD(L_\m^{x,b})$.  Therefore, $\phi_x^a$ belongs
to $\cD(L_\M^{a,x})$ and $\phi_x^b$ belongs to $\cD(L_\M^{x,b})$.
Since $x$ is a regular end of both intervals $]a,x[$ and $]x,b[$ the
function $\phi_x$ and its derivative $\phi_x'$ extend to continuous
functions on $]a,x]$ and $[x,b[$. However, these extensions are not
necessarily continuous on $]a,b[$, i.e. we must distinguish the left
and right limits at $x$, denoted $\phi_x(x\pm0)$ and $\phi'_x(x\pm0)$.

We now take $g\in\cD(L_{\min})$ in \eqref{eq:LG}. By taking into
account (5) of Theorem \ref{perq2} and what we proved above we have
$W(\phi_x,g;a)=0$ and $W(\phi_x,g;b)=0$.  Denote $\phi_x^a$ and
$\phi_x^b$ the restrictions of $\phi_x$ to the intervals $]a,x[$ and
$]x,b[$.  Then by using Green's identity on $]a,x[$ and $]x,b[$ in
(\ref{eq:LG1}) we get
\[
  g(x) =-W(\phi^a_x,g;x) + W(\phi^b_x,g;x) .
\]
We may compute the last two terms explicitly because $x$ is a regular
end of both intervals:
\begin{align*}
  W(\phi^a_x,g;x) &=\phi_x(x-0)g'(x)-\phi_x'(x-0)g(x), \\
  W(\phi^b_x,g;x) &=\phi_x(x+0)g'(x)-\phi_x'(x+0)g(x).
\end{align*}  
Thus we get
\[
  g(x)=(\phi_x(x+0)-\phi_x(x-0))g'(x)+(\phi_x'(x-0)
  -\phi_x'(x+0))g(x).
\]
The values $g(x)$ and $g'(x)$ may be specified in an arbitrary way
under the condition $g\in\cD(L_{\min})$ so we get
$\phi_x(x+0)-\phi_x(x-0)=0$ and $\phi_x'(x-0) -\phi_x'(x+0)=1$. Thus
$\phi_x$ must be a continuous function which is continuously differentiable
outside $x$ and its derivative has a jump
$\phi_x'(x+0) -\phi_x'(x-0)=-1$ at $x$.

Thus $G_\bullet$ is an integral operator with kernel
$G_\bullet(x,y)=\phi_x(y)$. But $G_\bullet^\#$ is also an $L^2$
Green's operator and clearly $G_\bullet^\#$ has kernel
$G_\bullet^\#(x,y)=\phi_y(x)$.  Repeating the above arguments applied
to $G_\bullet^\#$ we obtain the remaining statements of the
proposition.
\end{proof}

Let us describe a consequence of the above proposition; we use the
notation of Definition \ref{defempty}.

\begin{proposition}\label{exto1}
  If there exists a realization of $L$ such that $\lambda\in\C$ is in
  its resolvent set, then $\dim\cU_a(\lambda)\geq1$ and
  $\dim\cU_b(\lambda)\geq1$.
\end{proposition}

\begin{proof}
  Suppose that $L$ possesses a realization with $\lambda\in\C$
  contained in its resolvent set. This means that $L-\lambda$
  possesses an $L^2$ Green's operator $G_\bullet$.  By Proposition
  \ref{piuy00+} it can be chosen to satisfy
  $G_\bullet=G_\bullet^\#$. Then Proposition \ref{pr:gkernel} implies
  that for any $x\in]a,b[$ the function
  $G_\bullet(x,\cdot)\in L^2]a,b[$ belongs to $L^2]a,b[$ and satisfies
  $LG_\bullet(x,\cdot)=0$ on $]a,x[$ and $]x,b[$. We will prove that
  there is $x$ such that $G_\bullet(x,\cdot)\big|_{]x,b[}\neq0$, which
  implies $\dim\cU_b(\lambda)\geq1$. In order to prove that
  $\dim\cU_a(\lambda)\geq1$ it suffices to show that there is $x$ such
  that $G_\bullet(x,\cdot)\big|_{]a,x[}\neq0$ and the argument is
  similar.

  If the required assertion is not true, then
  $G_\bullet(x,\cdot)\big|_{]x,b[}=0$ for any $x$, in other terms
  $G_\bullet(x,y)=0$ for all $a<x<y<b$. Since $G_\bullet$ is
  self-transposed, \eqref{eq:trans} gives
  $G_\bullet(x,y)=G_\bullet(y,x)\ \forall x,y$. Hence we will also
  have $G_\bullet(x,y)=0$ for $a<y<x<b$. But this means $G_\bullet=0$,
  which is not true.
\end{proof}

\subsection{Forward and backward Green's operators}\label{ss:osbc}

Let us study the $L^2$ theory of the forward Green's operator
$G_\rightarrow $.  Recall that if $u,v$ span $\Ker (L)$ with
$W(v,u)=1$, then $G_\rightarrow $ is given by
\begin{equation}
  G_\rightarrow g(x)=v(x)\int_a^xu(y)g(y)\d y-u(x)\int_a^xv(y)g(y)\d y.
\end{equation}

Note that elements of $\Ker (L)$ do not have to be square integrable.
We have $\Ker(L_\M)=\Ker(L)\cap L^2]a,b[$. In the following
proposition we consider the case $\Ker(L)=\Ker(L_\M)$:

\begin{proposition}\label{moj1}
  Assume $\dim \Ker (L_\M)=2$. Then
\begin{enumerate}
\item $G_\rightarrow $ is Hilbert-Schmidt. In particular, it is an
  $L^2$ Green's operator of $L$.
\item Let $L_a$ be the operator defined in Def. \ref{defno}.  Then $L_a$
  has an empty spectrum, $(L_a-\lambda)^{-1}$ is compact for avery
  $\lambda\in\C$, and we have $L_a^{-1}=G_\rightarrow $.
\item Every $f\in \mathcal{D}(L_{\max})$ has a unique decomposition as
\begin{equation}
  f=\alpha u+\beta v+f_a,\quad f_a\in G_\rightarrow
  L^2]a,b[ \, =\cD(L_a).\label{moj} 
\end{equation}
\item $G_\leftarrow $ has analogous properties. In particular, we have
\begin{equation}\label{eq:GaGb}
  G_\rightarrow ^\#=G_\leftarrow ,\quad   L_b^{-1}=G_\leftarrow .
\end{equation}
\end{enumerate}
\end{proposition}

\proof By hypothesis, $u,v\in L^2]a,b[$.  The Hilbert-Schmidt norm of
$G_\rightarrow $ is clearly bounded by $\sqrt2\|u\|_2\|v\|_2$. Then by
Proposition \ref{piuy} zero belongs to the resolvent set of $L_a$,
$L_a^{-1}=G_a$, and
\begin{equation}
  \cD(L_{\max})=\cD(L_a)\oplus\Ker (L_\M),
\end{equation} 
which can be restated as the decomposition (\ref{moj}). If
$\lambda\in\C$ and $V$ is replaced by $V-\lambda$ then the new
$G_\rightarrow $ will be the resolvent at $\lambda$ of $L_a$, which
proves the second assertion in (2). Finally, \eqref{eq:GaGb} is proved
by a simple computation.  \qed

\begin{proposition}
  $G_\rightarrow $ is bounded if and only if $\dim\Ker (L_\M)=2$ (so
  that the assumptions of Prop.  \ref{moj1} are valid).
  \end{proposition}

\begin{proof}
  Let $G_\rightarrow $ be bounded. Then so is
  $G_\rightarrow ^\#=G_\leftarrow $. Let us recall the identity
  (\ref{mojj}):
  \begin{equation}
    G_\rightarrow -G_\leftarrow 
    =|v\rangle\langle u|-|u\rangle\langle v|.\label{mojj1}
 \end{equation}
 But the boundedness of the rhs of (\ref{mojj1}) implies
 $v,u\in L^2]a,b[$.
\end{proof}

$G_\rightarrow$ is useful even if it not a bounded operator,
especially if $\dim\cU_a(0)=2$:

\begin{proposition} \label{moj1a} 
  Assume that $\dim\cU_a(0)=2$.  Then $G_\rightarrow $ extends as a
  map from $L^2]a,b[$ to $C^1]a,b[$ satisfying the bounds
  \begin{align}
    |G_\rightarrow
    g(x)|&\leq\Big(|u(x)|\|v\|_x+|v(x)|\|u\|_x\Big)\|g\|_x,\\ 
    |\partial_xG_\rightarrow
    g(x)|&\leq\Big(|u'(x)|\|v\|_x+|v'(x)|\|u\|_x\Big)\|g\|_x, 
  \end{align}
  where $\|g\|_x:=\Big(\int_a^x |g(y)|^2\d y\Big)^{\frac12}$.  If
  $\chi\in C_\mathrm{c}^\infty[a,b[$, $\chi=1$ around $a$, then every
  $f\in \mathcal{D}(L_{\max})$ has a unique decomposition as
  \begin{equation} \label{moja}
    f=\alpha\chi u+\beta\chi v+f_a,\quad f_a\in
    \mathcal{D}(L_a).
  \end{equation}
\end{proposition}

\begin{proof}
  Let $a<d<b$. Then we can restrict our problem to $]a,d[$. Now
  $\dim\cU_a(0)=\dim\cU_d(0)=2$.  Therefore, we can apply
  Prop. \ref{moj1}, using the fact that $G_\rightarrow $ restricted to
  $L^2]a,d[$ is an $L^2$ Green's operator of
  $L^{a,d}$.
\end{proof}

The main assertion of Theorem \ref{th:proofconj} is,
  technically speaking, that $\dim\cU_a(0)=2$ if $\nu_a(L)=2$.  We may
  state an improved version of this assertion as a boundary value
  problem and this is of a certain interest: it says that if
  $\nu_a(L)=2$ then the endpoint $a$ behaves almost as if it were a
  regular end (in the regular case one works with $L^1$ instead of
  $L^2$).  Note that since only the behavior near $a$ of the solutions
  matters, we may assume $b$ a regular endpoint.

\begin{proposition}\label{pr:bvproblem}
  Suppose that $\nu_a(L)=2$ and $b$ is a regular endpoint for
  $L$. Let $\phi,\psi\in\cB_a(L)$ be a pair of linearly
  independent boundary value functionals.
Then the linear continuous map
\begin{equation}\label{eq:bvproblem}
  \cD(L_\M)\ni f \mapsto (Lf,\phi(f),\psi(f))\in
  L^2]a,b[\times\C\times\C
\end{equation}
is bounded and invertible. In particular, for any $g\in L^2]a,b[$ and
any $\alpha,\beta\in\C$, there is a unique $f\in\cD(L_\M)$ such that
$Lf=g$, $\phi(f)=\alpha$, and $\psi(f)=\beta$.
\end{proposition}

    \begin{proof}
   By Proposition \ref{moj1}, the operator $L_a:\cD(L_a)\to L^2]a,b[$
    is bijective hence the map \eqref{eq:bvproblem} is
    injective. Since the map is clearly continuous, by the open
    mapping theorem it suffices to prove its surjectivity. Let
    $g\in L^2]a,b[$ and $\alpha,\beta\in\C$. Since $L_a$ is
    surjective, there is $h\in\cD(L_a)$ such that $Lh=g$.  Now it
    suffices to show that there is $k\in\Ker(L_\M)$ such that
    $\phi(k)=\alpha,\psi(k)=\beta$ because then $f=h+k\in\cD(L_\M)$
    will satisfy $Lf=g$, $\phi(f)=\alpha$, and
    $\psi(f)=\beta$. Clearly, it suffices to prove this just for one
    couple $\phi,\psi$.  Since $\Ker(L_\M)$ is two dimensional, there
    are $u,v\in\Ker(L_\M)$ with $W(u,v)=1$ and we may take
    $\phi=\vec u_a$ and $\psi=\vec v_a$ since, by Theorem
    \ref{th:bvf} the boundary value functionals
    $\vec u_a,\vec v_a\in\cb_a(L)$ are linearly independent. Then it
    suffices to take $k=-\beta u+\alpha v$.
\end{proof}

\subsection{Green's operators with two-sided boundary conditions}

Recall from Def. \ref{def:1} that if $\phi\in\cB_a$ and $\psi\in\cB_b$
be are nonzero functionals, then $L_{\phi,\psi}$ is the operator
$ L_{\phi,\psi}\subset L_\M$ with
\begin{align*}
\cD(L_{\phi,\psi})&:=
\{f\in\cD(L_\M)\ \mid
  \phi(f)= \psi(f)=0\}.\end{align*}
Note that $L_{\phi,\psi}^\#=L_{\phi,\psi}$.

Recall also that, if $u,v$ are solutions of the equation $Lf=0$ with
$W(v,u)=1$, we defined in Def. \ref{two-side} the two-sided Green's
operator $G_{u,v}$
 \[G_{u,v}g(x):=\int_x^bu(x)v (y)g(y)\d y+
 \int_a^x v (x)u(y) g(y)\d y.\]
Clearly, there exists a close relationship between realizations of $L$ of the form $L_{\phi,\psi}$ and Green's operators of the form 
$G_{u,v}$.

\begin{proposition}
  Suppose $\phi\in\cB_a$, $\psi\in\cB_b$ and
  $0\in \rs( L_{\phi,\psi})$. Then there exists $u\in\cU_a(0)$ and
  $v\in\cU_b(0)$ with $W(v,u)\neq0$ such that, in the notation of
  Def. \ref{def:vec},
  \begin{equation} \label{eq:aa}
    \phi=\vec u_a,\quad \psi=\vec v_b,
  \end{equation}
\end{proposition}

\begin{proof}
  Let us prove the existence of $u$.  Note that by
  Proposition \ref{exto1} we have $\dim\cU_a(0)\geq1$. Then, by
  Proposition \ref{pr:gkernel}, the operator $L_{\phi,\psi}^{-1}$ has
  an integral kernel $G_{\phi,\psi}(\cdot,\cdot)$ such that for any
  $a<c<b$ the restriction of $G_{\phi,\psi}(c,\cdot)$ to $]a,c[$
  belongs to $\cD(L^{a,c}_\M)$ and satisfies
  $LG_{\phi,\psi}(c,\cdot)=0$. If $f\in\cD(L_{\phi,\psi})$ the
  relation \eqref{eq:LG} gives 
\begin{equation*} 
  f(x)=\int_a^b G_{\phi,\psi}(x,y)(L_{\phi,\psi} f)(y) \d y
\end{equation*}
hence if $a<x<c$ and $f(x)=0$ for $x>c$ we have
\begin{equation} \label{eq:lac}
  0=\int_a^c G_{\phi,\psi}(c,y)(L_{\phi,\psi} f)(y) \d y.
\end{equation}
Denote $L^{a,c}_{\phi,c}$ the operator in $L^2]a,c[$ defined by $L$
and the boundary conditions $\phi(f)=0$ and $f(c)=f'(c)=0$. Clearly,
any function $f$ satisfying such conditions extends to a function in 
$\cD(L_{\phi,\psi})$ if we set $f(x)=0$ for $x>c$ hence \eqref{eq:lac} 
is equivalent to
\begin{equation*}
  \int_a^c G_{\phi,\psi}(c,\cdot) L^{a,c}_{\phi,c} f \d x=0
  \quad \forall f\in\cD(L^{a,c}_{\phi,c}) .
\end{equation*}
We noted above that $L_{\phi,\psi}^\#=L_{\phi,\psi}$ and by a simple
argument this implies
$(L^{a,c}_{\phi,c})^\#=L^{a,c}_{\phi,0}\equiv L^{a,c}_{\phi}$ hence
the preceding relation means
$G_{\phi,\psi}(c,\cdot)|_{]a,c[} \in \Ker(L^{a,c}_{\phi})$.  Now
recall that during the proof of Proposition \ref{exto1} we have seen
that $c$ may be chosen such that
$G_{\phi,\psi}(c,\cdot)|_{]a,c[}\neq0$.  Finally, if we fix such a $c$
and denote $u=G_{\phi,\psi}(c,\cdot)$ then we get a nonzero element
$u\in\cU_a(0)$ such that $\phi(u)=0$ which, since $u\neq0$, is
equivalent to $\phi=\alpha\vec u_a$.

In an analogous way we prove the existence of $v$.  Both are nonzero.
If $u$ is proportional to $v$, then they are eigenvectors of
$L_{\phi,\psi}$ for the eigenvalue $0$, which contradicts
$0\in \rs( L_{\phi,\psi})$.  Hence they are not proportional to one
another, so that $W(v,u)\neq0$.
\end{proof}

Note that in the above proposition we can have $\phi=0$ or $\psi=0$,
or both.  However, $u$ and $v$ are always non-zero.

Suppose now that we start from a two-sided Green's operator.
\begin{proposition}
  Let $G_{u,v}$ be bounded on $L^2]a,b[$. Then $u\in \mathcal{U}_a(0)$
  and $v\in\mathcal{U}_b(0)$.
\end{proposition}

\begin{proof}
  Let $a<d<b$.  If $G_{u,v}$ is bounded, then so is
  $\one_{]a,d]}( x)G_{u,v} \one_{]d,b]}(x)$, where $x$ denotes the
  operator of multiplication by the variable in $]a,b[$. But its
  integral kernel is
  \[
    u(x)\one_{]a,d]}(x)v (y) \one_{]d,b]}(y)
  \]
  where $x$ and $y$ denote the variables in $]a,b[$.  This is a rank
  one operator with the norm
  \[
    \Big(\int_a^d|u|^2(x)\d x\Big)^{\frac12} \Big(\int_d^b|v |^2(x)\d
    x\Big)^{\frac12} . \qedhere
  \]
\end{proof}

Until the end of this subsection we assume that $u\in\cU_a(0)$ and
$v\in\cU_b(0)$ and the functionals $\phi,\psi$ are given by
(\ref{eq:aa}).  Thus we have both Green's operator $G_{u,v}$ and the
operator $L_{\phi,\psi}$.

 Let
$\chi\in C^\infty]a,b[$ such that $\chi=1$ close to $a$ and $\chi=0$
close to $b$. Clearly,
\begin{equation}\label{eq:0}  \cD(L_{\phi,\psi})=\cD(L_{\min})+\Span\{ \chi u,(1-\chi) v\big).\end{equation}

We will show that  $G_{u,v}$ is bounded 
iff and only if $0\in\rs(L_{\phi,\psi})$.
However, it seems that there is no guarantee that 
 $G_{u,v}$ is bounded.

\begin{proposition}\label{khava-} 
\begin{align}
    G_{u,v}L_\mathrm{c}^2]a,b[&\subset \cD(L_{\phi,\psi}),\label{khav1} \\
    G_{u,v}L^2]a,b[&\subset AC^1]a,b[.\label{khav2}
  \end{align}
  Moreover, $G_{u,v}$ is bounded if and only if there exists $c>0$
  such that
\begin{equation}
      \|      L_{\phi,\psi}f\|\geq c\|f\|,\quad  f\in\cD(L_{\phi,\psi}).\label{khav3}
\end{equation}
If this is the case, then $0$ belongs to the resolvent set of $L_{\phi,\psi}$, we have
$G_{u,v}=L_{\phi,\psi}^{-1}$, $G_{u,v}^\#=G_{u,v}$ and
\begin{equation}\cD(L_{\phi,\psi})=G_{u,v}L^2]a,b[.\label{khav4}
\end{equation}
\end{proposition}

\proof
It is easy to see that
\[    G_{u,v}L_\mathrm{c}^2]a,b[\subset
    \cD(L_\mathrm{c})+\Span\big\{ \chi u,(1-\chi) v\big\}, \]
    which implies (\ref{khav1}).

Let $g\in L^2]a,b[$. For $a<x<b$ we compute:
\begin{align}
  \partial_xG_{u,v} g(x)
  =u'(x)\int_{x}^bv (y)g(y)\d y+ v' (x)\int_a^{x}u(y)g(y)\d y.
\label{mojk2}
\end{align}
Now,
$x\mapsto u'(x),v'(x),\int_{x}^bv (y)g(y)\d y,\, \int_a^{x}u(y)g(y)\d
y$ belong to $AC]a,b[$. Hence (\ref{mojk2}) belongs to
$AC]a,b[$. Therefore, (\ref{khav2}) is true.  Next, let
 \begin{equation}
      f=f_\mathrm{c}+\alpha \chi u+\beta(1-\chi)v,\quad
      f_\mathrm{c}\in\cD(L_\mathrm{c}).\label{piuy1}
\end{equation} 
We compute, integrating by parts, \label{p:379}
\begin{align}
      G_{u,v}L_{\phi,\psi}f(x)&=
      \int_a^b\Big(\big(-\partial_y^2+V(y)\big)G_{u,v}(x,y)\Big)f(y)\d y
      \\
      &+\lim_{y\to a}\big(G_{u,v}(x,y)f'(y)-\partial_yG_{u,v}(x,y)f(y)\big)\\
      &-\lim_{y\to b}\big(G_{u,v}(x,y)f'(y)-\partial_yG_{u,v}(x,y)f(y)\big)\\
      &=f(x)+v(x)W(u,f;a)-u(x)W(v,f;b)\quad =
      \quad f(x).\label{piuy2}
\end{align}
Moreover, functions of the form (\ref{piuy1}) are dense in
$\cD(L_{\phi,\psi})$. Therefore, if $G_{u,v}$ is bounded, then (\ref{piuy2})
extends to
\begin{equation}
  G_{u,v}L_{\phi,\psi}f=f,\quad f\in \cD(L_{\phi,\psi}).
      \label{khava}
\end{equation}
Hence $\|f\|=\|G_{u,v}L_{\phi,\psi}f\|\leq \|G_{u,v}\|\|L_{\phi,\psi}f\|$
  which gives \eqref{khav3}.

Assume that   $G_{u,v}$ is bounded in the sense of
$L^2]a,b[$. By Prop. \ref{piuy6}, $G_{u,v}$ is an $L^2$ Green's operator.
  By Prop. \ref{piuy5},  it is also bounded from $L^2]a,b[$ to
$\cD(L_{\max})$. Therefore (\ref{khav1}) extends then to
    \begin{equation}
      G_{u,v}L^2]a,b[\subset \cD(L_{\phi,\psi}),
    \end{equation}
    so that
\begin{equation}
  L_{\phi,\psi}      G_{u,v}g=g,\quad g\in L^2]a,b[.  \label{khavb}\end{equation}
By   (\ref{khava}) and (\ref{khavb}),   $G_{u,v}$ is a (bounded)
inverse of $L_{\phi,\psi}$ so that 
(\ref{khav3}) and  (\ref{khav4}) are true.

Now assume that   (\ref{khav3}) holds. By
(\ref{khav1}), we then have
\begin{equation}
  g=L_{\phi,\psi}G_{u,v}g,\quad g\in L_\mathrm{c}^2]a,b[.
\end{equation}
Hence,
\begin{equation}
 \| g\|=\|L_{\phi,\psi}G_{u,v}g\|\geq c\|G_{u,v} g\|
\end{equation}
on $L_\mathrm{c}^2]a,b[$, which is dense in $L^2]a,b[$.  Therefore,
$G_{u,v}$ is bounded.
\qed

\subsection{Classification of realizations possessing non-empty
  resolvent set} 

In applications well posed operators (possessing non-empty resolvent
set) are by far the most useful. The following theorem describes a
classification of realizations of $L$ with this property.

\begin{theorem}
  Suppose that $L_\bullet$ is a realization of $L$ with a non-empty
  resolvent set. Then exactly one of the following statements is true.
  \begin{enumerate}

  \item $L_\bullet=L_\M$.\\
    Then also $L_\m=L_\bullet$, so that $L$ possesses a unique
    realization. We have $\nu(L)=0$. \\
    If $\lambda\in\rs(L_\bullet)$, then $\dim\Ker(L_\M-\lambda)=0$,
    $\dim\cU_a(\lambda)=\dim\cU_b(\lambda)=1$ and
    $\cU_a(\lambda)\neq\cU_b(\lambda)$. If $u\in\cU_a(\lambda)$ and
    $v\in\cU_b(\lambda)$ with $W(v,u)=1$, then
    \[
      (L_\bullet-\lambda)^{-1}=G_{u,v}.
    \]
    $L_\bullet$ is self-transposed and has separated boundary
    conditions. (See Def. \ref{separated} for separated boundary
    conditions.)

  \item The inclusion $\cD(L_\bullet)\subset\cD(L_\M)$ is of
    codimension $1$.\\
    Then the inclusion $\cD(L_\m)\subset\cD(L_\bullet)$ is of
    codimension $1$ and we have $\nu(L)=2$. \\
    If $\lambda\in\rs(L_\bullet)$, then $\dim\Ker(L_\M-\lambda)=1$,
    $\dim\cU_a(\lambda)=2$ and $\dim\cU_b(\lambda)=1$, or
    $\dim\cU_a(\lambda)=1$ and $\dim\cU_b(\lambda)=2$.  We can find
    $u\in\cU_a(\lambda)$ and $v\in\cU_b(\lambda)$ with $W(v,u)=1$ such
    that
    \[(L_\bullet-\lambda)^{-1}=G_{u,v}.\]
    $L_\bullet$ is self-transposed and has separated boundary conditions.
  \item The inclusion $\cD(L_\bullet)\subset\cD(L_\M)$ is of
    codimension $2$.\\
    Then the inclusion $\cD(L_\m)\subset\cD(L_\bullet)$ is of
    codimension $2$.  We have $\nu(L)=4$. \\
    The spectrum of $L_\bullet$ is discrete and its resolvents are
    Hilbert-Schmidt.  For any $\lambda\in\C$ we have
    $\dim\Ker(L_\M-\lambda)=2$, $\dim\cU_a(\lambda)=2$ and
    $\dim\cU_b(\lambda)=2$. \\
    If in addition $L_\bullet$ is self-transposed, has separated
    boundary conditions, and $\lambda\in\rs(L_\bullet)$, then we can
    find $u\in\cU_a(\lambda)$ and $v\in\cU_b(\lambda)$ with $W(v,u)=1$
    such that
    \[(L_\bullet-\lambda)^{-1}=G_{u,v}.\] If, instead, $L_\bullet$ is
    not self-transposed and has separated boundary conditions, then it
    has empty spectrum and one of the following possibilities hold:
\begin{romanenumerate}
\item $L_\bullet=L_a$ and $(L_\bullet-\lambda)^{-1}$ is given by the forward Green's operator.
  \item $L_\bullet=L_b$ and $(L_\bullet-\lambda)^{-1}$ is given by the backward Green's operator.
\end{romanenumerate}
We have $L_a^\#=L_b$, and both (i) and (ii) are described
 in Prop. \ref{moj1a}.
  \end{enumerate}
\end{theorem}

\subsection{Existence of realizations with non-empty resolvent set}

$\C\backslash\R$ is contained in the resolvent set of all self-adjoint operators. The following proposition gives a generalization of this fact.

\begin{proposition}\label{pr:lpc}
  Let $V_{\mathrm{R}}$ and $V_{\mathrm{I}}$ be the real and imaginary
  part of $V$. Let $\|V_\mathrm{I}\|_\infty=:\beta<\infty$.
  Then
  \begin{equation}
    \{\lambda\in\C\mid|\Im\lambda|>\beta\}\label{hull1}
  \end{equation}
  is contained in the resolvent set of some realizations of $L$. All
  realizations of $L$ possess only discrete spectrum in (\ref{hull1}).
  \end{proposition}

\begin{proof}
  Let $L_\mathrm{R}:=-\partial_x^2+V_\mathrm{R}$. By Theorem
  \ref{perq2}, $L_{\mathrm{R},\min}$ is densely defined and
  $ L_{\mathrm{R},\min}^\# =L_{\mathrm{R},\max}\supset
  L_{\mathrm{R},\min}$.  By the reality of $V_\mathrm{R}$,
  $L_{\mathrm{R},\min}^*=L_{\mathrm{R},\min}^\#$. Therefore,
  $L_{\mathrm{R},\min}^*\supset L_{\mathrm{R},\min}$.  This means that
  $L_{\mathrm{R},\min}$ is Hermitian (symmetric).  Let us now apply
  the well-known theory of self-adjoint extensions of Hermitian
  operators.  Let $d_\pm:=\Ker \big(L_{\mathrm{R},\min}^*\mp\i)$ be
  the deficiency indices.  Using the fact that $L_{\mathrm{R},\min}$
  is real we conclude that $d_+=d_-$. Therefore, $L_{\mathrm{R},\min}$
  possesses at least one self-adjoint extension, which we denote
  $L_{\mathrm{R},\bullet}$.  By the self-adjointness of
  $L_{\mathrm{R},\bullet}$ we have
  $ \|( L_{\mathrm{R},\bullet}-\lambda)^{-1}\|\leq |\Im\lambda|^{-1}$
  for all $\lambda\not\in\R$.  Set
  $L_\bullet:= L_{\mathrm{R},\bullet}+\i V_\mathrm{I}$.  Clearly,
\begin{equation}\label{exto}
  L_{\max}\supset L_\bullet\supset L_{\min}.
\end{equation}
For $|\Im\lambda|>\beta$, $\lambda$ belongs to the resolvent set of
$L_\bullet$, and its resolvent is given by
\[(L_\bullet- \lambda)^{-1}= ( L_{\mathrm{R},\bullet}-\lambda)^{-1}\big(\one+\i V_\mathrm{I} ( L_{\mathrm{R},\bullet}-\lambda)^{-1}\big)^{-1}.
\qedhere\]
\end{proof}

Note that the above proposition can be improved to cover some singularities of $V_\mathrm{I}$. In fact, if
 there are
  numbers $\alpha,\beta$ with $0\leq\alpha<1$ such that
  \[\|V_{\mathrm{I}}f\|^2\leq \alpha^2\big(\|L_{\mathrm{R},\bullet} f\|^2+\beta^2\|f\|^2\big),\quad \forall f\in\cD(L_{\mathrm{R},\bullet}),\] then still
  \[\|V_\mathrm{I} ( L_{\mathrm{R},\bullet}-\lambda)^{-1}\|\leq\alpha<1,\]
  and the conclusion of Prop. \ref{pr:lpc} holds.

\subsection{``Pathological'' spectral properties} \label{ss:usp}

We construct now 1d Schr\"odinger operators whose realizations have an
empty resolvent set.  Such operators seem to be rather pathological
and not very interesting for applications.

\begin{proposition}\label{pr:example}
  There is $V\in L_\loc^\infty[0,\infty[$ such that if
  $L=-\partial^2+V$ then any operator $L_\bullet$ on $L^2]0,\infty[$
  with $L_{\min}\subset L_\bullet\subset L_{\max}$ has empty resolvent
  set, hence $\sigma(L_\bullet)= \C$.
\end{proposition}

\begin{proof}
  Let $I_n=]n^2-n,n^2+n[$ with $n\geq1$ integer. Then $I_n$ is an open
  interval of length $|I_n|=2n$ and $I_{n+1}$ starts with the point
  $n^2+n$ which is the endpoint of $I_n$. Thus $\cup_n I_n$ is a
  disjoint union equal to $]0,\infty[\,\setminus\{n^2+n\mid
  n\geq1\}$. Let $\P$ be the set prime numbers $\P=\{2,3,5,\dots\}$
  and for each prime $p$ denote $J_p=\cup_{k\geq1}I_{p^k}$. We get a
  family of open subsets $J_p$ of $]0,\infty[$ which are pairwise
  disjoint and each of them contains intervals of length as large as
  we wish. Now let $p\mapsto c_p$ be a bijective map from $\P$ to the
  set of complex rational numbers and let us define a function
  $V:[0,\infty[\,\to\C$ by the following rules: if $x\in J_p$ for
  some prime $p$ then $V(x)=c_p$ and $V(x)=0$ if
  $x\notin\cup_pJ_p$. Then $V$ is a locally bounded function whose
  range contains all the complex rational numbers. We set
  $L=-\partial^2+V(x)$ and we prove that the spectrum of any
  $L_\bullet$ with $L_{\min}\subset L_\bullet\subset L_{\max}$ is
  equal to $\C$. Since the spectrum is closed, it suffices to show
  that any complex rational number $c$ belongs to the spectrum of any
  $L_\bullet$. If not, there is a number $\alpha>0$ such that
  $\|(L_\bullet-c)\phi\|\geq\alpha\|\phi\|$ for any
  $\phi\in\cD(L_\bullet)$. If $r$ is a (large) positive number then
  there is an open interval $I$ of length $\geq r$ such that $V(x)=c$
  on $I$. Let $\phi\in C^\infty_{\mathrm{c}}(I)$ such that $\phi(x)=1$
  for $x$ at distance $\geq1$ from the boundary of $I$ and with
  $|\phi''|\leq\beta$ with a constant $\beta$ independent of $r$ (take
  $r>3$ for example). Then $\phi\in\cD(L_{\min})$ and
  $(L-c)\phi=-\phi''+V\phi-c\phi=-\phi''$ hence
  $\|\phi''\|=\|(L_\bullet-c)\phi\|\geq\alpha\|\phi\|$ so
  $\alpha\|\phi\|\leq 2\beta$ which is impossible because the left
  hand side is of order $\sqrt{r}$. One may choose $V$ of class
  $C^\infty$ by a simple modification of this construction.
\end{proof}

\section{Potentials with a negative imaginary part}
\protect\setcounter{equation}{0}

\subsection{Dissipative 1d Schr\"odinger  operators} \label{ss:dsto}

Recall that an operator $A$ is called {\em dissipative} if
\begin{equation} \label{eq:diss}
  \Im(f|Af)\leq0,\quad f\in\cD(A),\end{equation}
that is, if its numerical range is contained in
$\{\lambda\in\C\mid\Im\lambda\leq0\}$. It is called {\em maximal
  dissipative} if in addition its spectrum is contained in
$\{\lambda\in\C\mid\Im\lambda\leq0\}$.
The following criterion is well-known \cite{Kato}.

\begin{proposition}\label{pr:adis}
  Assume $A$ is closed, densely defined and dissipative. Then $A$ is
  maximal dissipative if and only if $-A^*$ is dissipative, and then
  $-A^*$ is maximal dissipative.
\end{proposition}

Let us now consider $L=-\partial_x^2+V(x)$ with $V\in L_\loc^1]a,b[$.

\begin{proposition}\label{pr:mindis}
  The operator $L_\m$ is dissipative if and only if $\Im V\leq0$.
  It is maximal dissipative if in addition $L_\m=L_\M$.
\end{proposition}

\begin{proof}
  
  If $f\in \cD(L_\M)$ then
  \begin{align*}
    & (f|L_\M f) =\int_a^b \big(\bar{f}'f'-(\bar{f}f')' +V\bar{f}f  \\
    &=\lim_{\substack{a_1\to a \\ b_1\to b}}
    \left(\bar{f}(a_1)f'(a_1) -\bar{f}(b_1)f'(b_1)
    +\int_{a_1}^{b_1} \left(|f'|^2+V|f|^2 \right) \right) 
  \end{align*}
  hence
  \begin{equation}\label{eq:diseq}
    \Im(f|L_\M f)  = 
    \lim_{\substack{a_1\to a \\ b_1\to b}}
    \left(\int_{a_1}^{b_1}\Im(V)|f|^2  + \Im(\bar{f}(a_1)f'(a_1))
      -\Im(\bar{f}(b_1)f'(b_1))   \right).
  \end{equation}
  Thus
  \begin{equation}\label{dissi}
    \Im(f|L_\m f) = \int_{a}^{b}\Im(V)|f|^2\geq0,\quad f\in
    \cD(L_{\mathrm{c}}).
  \end{equation}
  By continuity, (\ref{dissi}) extends to $f\in \cD(L_\m)$, which
  clearly implies that $L_\m$ is dissipative.  The same argument as
  above shows that $-\bar L_\m$ is dissipative

  If $L_\m=L_\M$, then $\bar L_\m=\bar L_\M$.  But $L_\m^*=\bar L_\M$.
  Hence $-L_\m^*$ is dissipative. This proves maximal dissipativity of
  $L_\m$.

  If $L_\m\neq L_\M$, then the spectrum of $L_\m$ is $\C$, so $L_\m$
  is not maximally dissipative.
\end{proof}

A convenient criterion for dissipativity holds if
$\cD(A)\subset \cD(A^*)$: then $A$ is dissipative if and only if
$\frac{1}{\i 2}(A-A^*)\geq 0$.  Unfortunately, an operator $A$ can be
dissipative even if $\cD(A)\cap\cD(A^*)=\{0\}.$

This is related to a certain difficulty when one tries to study
the dissipativity of Schr\"odinger operators with singular complex
potentials.  Suppose that we want to check that $L_\bullet$ is
dissipative using this criterion.  If
$L_\m\subset L_\bullet \subset L_\M$, then
$\bar{L}_\m\subset L_\bullet^*\subset \bar{L}_\M$.  But we may have
$\cD(\bar{L}_\M) \cap \cD(L_\M)=\{0\}$ (Lemma \ref{lm:conjugate}),
hence we could have $\cD(L_\bullet)\cap \cD(L_\bullet^*)=\{0\}$ which
is annoying. Indeed, although
$W_{x}(\bar{f},f)=2\i\Im(\bar{f}(x)f'(x))$, we cannot use in
\eqref{eq:diseq} the existence of the limits \eqref{wronski-a} and
\eqref{wronski-b} because in general $\bar f\not\in \cD(L_\M)$.

Let us describe the action of conjugation on boundary conditions.
The \emph{conjugate} of $\alpha\in\cb(L)$ is the boundary functional
$\bar\alpha\in\cb(\bar L)$ given by
\begin{equation}\label{eq:conjfunc}
\bar{\alpha}(f):=\bar{\alpha(\bar f)} .
\end{equation}
Clearly $\alpha\mapsto\bar\alpha$ is a bijective anti-linear map
$\cb(L)\to\cb(\bar L)$ which sends $\cb_a(L)$ into $\cb_a(\bar L)$ and
$\cb_b(L)$ into $\cb_b(\bar L)$.  Then if $g\in \cD(L_\M)$ is a
representative of $\alpha\in\cB_a$, so that
\begin{equation}
  \alpha(f)=W_a(g,f),\quad f\in\cD(L_\M),\label{repre1}
\end{equation}
then 
\begin{equation} \label{repre2}
  \bar\alpha(f)=W_a(\bar g,f) ,\quad f\in\cD(\bar{L}_\M) .
\end{equation}
Recall that $\cB_a$ and $\cB_b$ are equipped with symplectic forms
$\lbra\cdot|\cdot\rbra_a$ and $\lbra\cdot|\cdot\rbra_b$, see \eqref{sigma}.

Fix $\alpha\in\cB_a$ and $\beta\in\cB_b$. Consider $L_{\alpha,\beta}$,
the realizations of $L$ introduced in Definition \ref{def:1}. Recall
that it is the restriction of $L_\M$ to 
$\cD(L_{\alpha\beta})=\{f\in \cD(L_\M)\mid \alpha(f)=\beta(f)=0\}$.

\begin{proposition}\label{pr:hadpp}
$L_{\alpha \beta}^*=\bar{L}_{\bar\alpha\,\bar\beta}$.
\end{proposition}

\begin{proof}
  By Proposition \ref{khava-},
  $L_{\alpha \beta}^\#={L}_{\alpha\,\beta}$.  Clearly,
  $\bar{L_{\alpha \beta} } =\bar{L}_{\bar\alpha\,\bar\beta}$.
  Therefore, the proposition follows from
  $L_{\alpha \beta}^*=\bar{L_{\alpha\,\beta}^\#}$.
\end{proof}
  
In the following two subsections we will describe boundary conditions
that guarantee the dissipativity of $L_{\alpha,\beta}$.  We will
consider separately two classes of potentials: $V\in L_\loc^2]a,b[$
and $V\in L^1]a,b[$.

\subsection{Dissipative boundary conditions for locally square
  integrable potentials}

Note first the following sesquilinear version of Green's identity
(\ref{eq:lagrange1}).

\begin{lemma} \label{lm:ImVl} Suppose that
  $f,\bar{f},g\in\cD(L_\M)$. Then $\Im( V) f\in L^2]a,b[$ and
  \begin{align} \label{lsh}
    (L_\M f|g)-(f|L_\M g)
    =-2\i\int_a^b\Im(V)\bar fg+W_b(\bar f,g)-W_a(\bar f,g).
  \end{align}
\end{lemma}

\begin{proof}
  The left hand side of (\ref{lsh}) is
\begin{align}\label{lsh1}
  \langle \overline{L_\M f}|g\rangle-\langle\bar f|L_\M g\rangle
  &= \langle \overline{L_\M f}|g\rangle-\langle L_\M\bar f|g\rangle\\
  &+\langle L_\M\bar f|g\rangle-\langle\bar f|L_\M g\rangle.\label{lsh2}
\end{align}
Then we apply $\bar{L_\M}-L_\M=-2\i\Im(V)$ to (\ref{lsh1}) and Green's
identity (\ref{eq:lagrange1}) to (\ref{lsh2}). \end{proof}

For the rest of the argument we need the equality of the domains
$\cD(\bar{L}_\M)=\cD(L_\M)$ which, by Lemma \ref{lm:conjugate}, is
equivalent to $\Im V\in L_\loc^2]a,b[$.  Then we have
$\cB(L)=\cB(\bar L)$, hence $\alpha\mapsto\bar\alpha$ is a conjugation
in $\cB(L)$ which leaves invariant the subspaces $\cB_a(L)$ and
$\cB_b(L)$. Recall that in (\ref{sigma}) we equiped $\cB_a(L)$ in a
symplectic form $\lbra\cdot|\cdot\rbra_a$.  Note that
$\lbra\bar\alpha|\alpha\rbra_a$ is well defined for any
$\alpha\in\cB_a(L)$.

\begin{lemma} \label{lemo} If $\Im V\in L_\loc^2]a,b[$ and
  $\alpha\in\cB_a$ then the number $\lbra\bar\alpha|\alpha\rbra_a$ is
  purely imaginary and
  \begin{align*}
    \frac{1}{2\i}\lbra\bar\alpha|\alpha\rbra_a\geq0
    \ \Longleftrightarrow\ 
    \frac{1}{2\i} W_a(\bar f,f)\geq0 \ \forall
    f\in\cD(L_\M) \text{ with } \alpha(f)=0 .
  \end{align*}
\end{lemma}
\begin{proof} Let $g\in \cD(L_\M)$ be a representative of $\alpha$,
so that (\ref{repre1}) and  (\ref{repre2}) are true. Then
\[\lbra\bar\alpha|\alpha\rbra_a=W_a(\bar g,g)
  =\lim_{c\searrow a}\big(\bar g(c)g'(c)-\bar g'(c)g(c)\big),\] which
proves that $\lbra\bar\alpha|\alpha\rbra_a$ is purely imaginary.  Now,
by the Kodaira identity
\[W_a(\bar g,g)W_a(\bar f,f)= |W_a(g,f)|^2-|W_a(\bar g,f)|^2.\] But
$\alpha(f)=0$ means $W_a(g,f)=0$. Therefore,
$\lbra(\bar\alpha|\alpha\rbra_a W_a(\bar f,f)\leq0$. 
\end{proof}

\begin{theorem} \label{th:ImV}
  Let $\alpha\in\cB_a$ and $\beta\in\cB_b$.  If
  $\Im V \in L^2_\loc]a,b[$, we have
  \begin{equation} \label{eq:lemo1} L_{\alpha\beta} \text{ is
      dissipative } \Longleftrightarrow \Im V\leq0, \
    \frac{1}{2\i}\lbra\bar\alpha|\alpha\rbra_a\leq0,\text{ and }
    \frac{1}{2\i}\lbra\bar\beta|\beta\rbra_b\geq0.
 \end{equation}
 And then $L_{\alpha\beta}$ is maximal dissipative.
\end{theorem}

\begin{proof}

We consider only the case $\alpha\neq0,\beta\neq0$.  Lemma
\ref{lm:ImVl} gives  
\begin{align} \label{eq:lsh}
  \Im(f|L_\M f)
  =\int_a^b\Im(V)|f|^2+\frac{1}{2\i}W_a(\bar f,f)
   - \frac{1}{2\i}W_b(\bar f,f) \quad \forall f\in \cD(L_\M)
\end{align}
and this implies that $L_{\alpha\beta}$ is dissipative if and only if
\begin{align} \label{eq:lshdis}
  \frac{1}{2\i}W_a(\bar f,f) - \frac{1}{2\i}W_b(\bar f,f)
\leq \int_a^b\Im(-V)|f|^2 
  \quad \forall f\in \cD(L_{\alpha\beta}) .
\end{align}
If $L_{\alpha\beta}$ is dissipative, by taking
$f\in \cD(L_{\mathrm{c}})$ in \eqref{eq:lshdis} we get
$\Im(-V)\geq0$. Then by choosing $f\in \cD(L_{\alpha\beta})$ equal to
zero near $b$ we get
$\frac{1}{2\i}W_a(\bar f,f)\leq \int_a^b\Im(-V)|f|^2$. If we fix such
an $f$ and replace it in this estimate by $f\theta$ where
$\theta\in C^\infty(\R)$ with $0\leq\theta\leq1$ and $\theta(x)=1$ on
a neighborhood of $a$ the we get
$\frac{1}{2\i}W_a(\bar f,f)\leq \int_a^b\Im(-V)|f\theta|^2$. Since the
right hand side here can be made as small as we wish by taking
$\theta$ equal to zero for $x>d>a$ with $d$ close to $a$, we see that
we must have $\frac{1}{2\i}W_a(\bar f,f)\leq0$ and this clearly
implies the same inequality for any $f\in \cD(L_{\alpha\beta})$.  Then
we get $\frac{1}{2\i}\lbra\bar\alpha|\alpha\rbra_a\leq0$ by Lemma
\ref{lemo}. We similarly prove
$\frac{1}{2\i}\lbra\bar\beta|\beta\rbra_b\geq0$.

We proved the implication $\Rightarrow$ in \eqref{eq:lemo1} and
$\Leftarrow$ is clear by \eqref{eq:lshdis}. It remains to show the
maximal dissipativity assertion. Due to Propositions \ref{pr:adis} and
\ref{pr:hadpp} it suffices to prove that the operator
$-L^*_{\alpha\beta}=-\bar{L}_{\bar\alpha\,\bar\beta}$ is
dissipative. Observe first that the relation $\cD(\bar{L}_\M)=\cD(L_\M)$
implies
$\cD(\bar{L}_{\bar\alpha\,\bar\beta})=\cD(L_{\bar\alpha\,\bar\beta})$. Then
\eqref{eq:lsh} gives
\begin{align} \label{eq:lsh3}
  \Im(f|-\bar{L}_\M f)
  =\int_a^b\Im(V)|f|^2-\frac{1}{2\i}W_a(\bar f,f)
   +\frac{1}{2\i}W_b(\bar f,f) \quad \forall f\in \cD(L_\M)
\end{align}
hence instead of \eqref{eq:lshdis} we get the condition 
\begin{align*} 
  -\frac{1}{2\i}W_a(\bar f,f) + \frac{1}{2\i}W_b(\bar f,f)
\leq \int_a^b\Im(-V)|f|^2 
  \quad \forall f\in \cD(L_{\bar\alpha\,\bar\beta}) .
\end{align*}
As above we get $\frac{1}{2\i}W_a(\bar f,f) \geq0$ and
$\frac{1}{2\i}W_b(\bar f,f)\leq0$ for any
$f\in \cD(L_{\bar\alpha\,\bar\beta})$.  Thus, if $f\in \cD(L_\M)$ and
$\bar\alpha(f)=0$ then $\frac{1}{2\i}W_a(\bar f,f) \geq0$ and by Lemma
\ref{lemo} this means $\frac{1}{2\i}\lbra\alpha|\bar\alpha\rbra_a\geq0$
which is equivalent to
$\frac{1}{2\i}\lbra\bar\alpha|\alpha\rbra\leq0$. Similarly we get
$\frac{1}{2\i}\lbra\bar\beta|\beta\rbra_b\geq0$ and the last two
conditions are satisfied by the assumptions in the right hand side of
\eqref{eq:lemo1}. Hence $-L^*_{\alpha\beta}$ is dissipative.
\end{proof}

\subsection{Dissipative regular boundary conditions} \label{ss:regbc}

Suppose that the operator $L$ has a regular left endpoint at $a$.  As
we noted several times, for regular boundary conditions $\cB_a$ can be
identified with $\C^2$. Indeed,
\[\alpha(f)=\alpha_0f'(a)-\alpha_1f(a),\]
is a general form of a boundary functional, with
 $\alpha=(\alpha_0,\alpha_1)\in\C^2$
and $f\in\cD(L_\M)$.

The space $\cB_a$ is equipped with the symplectic form $\lbra \cdot|\cdot\rbra_a$, which coincides with the usual (two-dimensional) vector product:
\[\lbra\alpha|\beta\rbra_a=\alpha_0\beta_1-\alpha_1\beta_0=\alpha\times\beta.\]
Thus, if we write $\vec f_a:=\big(f(a),f'(a)\big)$,  an alternative notation for $\alpha(f)$ is
\[\alpha(f)=\alpha\times\vec f_a.\]

Note that there is no guarantee that $\cD(L_\m)$ and $\cD(L_\M)$ are
invariant wrt the complex conjugation. However the space
$\cB_a\simeq\C^2$ is equipped with the obvious complex conjugation:
\[
  \bar\alpha(f)=\bar\alpha_0f'(a)-\bar\alpha_1f(a)=\bar\alpha\times
  \vec f_a.
\]

\begin{lemma}\label{lm:easy}
  (1) $\alpha\times\beta=0$ if and only if the vectors
  $\alpha,\beta$ are collinear. \\[1mm]
  \phantom{xxxxxxxxxxxxx} (2) $\bar\alpha\times\alpha\in\i\R$ and
  $\bar\alpha\times\alpha =0$ if and only if $\alpha$ is
  proportional to a real vector.\\[1mm]
  \phantom{xxxxxxxxxxxxx} (3)
  $(\bar\alpha\times\alpha)(\bar\beta\times\beta)=|\alpha\times\beta|^2
  -|\bar\alpha\times\beta|^2$
\end{lemma}

\begin{proof}
(1):  If $\alpha_0\beta_1=\alpha_1\beta_0$ and $\beta\neq0$ then
  $\beta_k=0\Rightarrow\alpha_k=0$ and if $\beta_0\neq0\neq\beta_1$
  then $\alpha_0/\beta_0=\alpha_1/\beta_1$. (2): If
  $\bar\alpha\times\alpha=0$ we get $\bar\alpha=c^2\alpha$
  for some complex $c$ with $|c|=1$ which implies
  $(c\alpha)^*=c\alpha$. (3) follows by the Kodaira identity.
\end{proof}  

Here is a version of Lemma \ref{lm:ImVl} for the regular case.

\begin{lemma}\label{lm:ImVV}
  Let $V\in L^1]a,b[$.  Suppose that $f,g\in\cD(L_\M)$. Then
  \begin{align}
    (L_\M f|g)-(f|L_\M g)
    &=-2\i\int_a^b\Im(V)\bar fg+W_b(\bar f,g)-W_a(\bar f,g).
    \label{lhsl}
  \end{align}
  \end{lemma}

  Next we have a version of Thm \ref{th:ImV} for the regular case.
  Fix \emph{nonzero vectors} $\alpha,\beta$ and define
  $L_{\alpha\beta}$ by imposing the boundary conditions at $a$ and
  $b$:
\[
f(a)\alpha_1-f'(a)\alpha_0=0, \quad f(b)\beta_1-f'(b)\beta_0=0 .
\]
In this context it is quite easy to prove that
$L^*_{\alpha\beta}=\bar{L}_{\bar\alpha\,\bar\beta}$. 

\begin{theorem} \label{lm:ImV1} Suppose that $a,b$ are finite and
  $V\in L^1]a,b[$.  Then
  \[
    L_{\alpha\beta}\text{ is dissipative }\Leftrightarrow \Im
    V\leq0,\quad\Im(\bar\alpha_0\alpha_1)\leq0,\text{ and }
    \Im(\bar\beta_0\beta_1)\geq0.
  \]
  And in this case $L_{\alpha \beta}$ is maximal dissipative.
\end{theorem}

\begin{proof}
  The proof is similar to that of Theorem \ref{th:ImV}, but much
  simpler.  We use Lemma \ref{lm:ImVV} instead of Lemma \ref{lm:ImVl}
  and get the same relation \eqref{eq:lshdis} as necessary and
  sufficient condition for dissipativity. Then we use
  \[
    \frac{1}{2\i}\lbra\bar\alpha|\alpha\rbra_a =
  \frac{1}{2\i}(\bar\alpha_0\alpha_1-\bar\alpha_1\alpha_0)=
  \Im(\bar\alpha_0\alpha_1)
  \]
  and a similar relation for $\beta$. 
  Finally, when checking the dissipativity of $-L^*_{\alpha\beta}$,
  note that this operator is associated to the differential expression
  $\partial^2-\bar{V}$, which explains a difference of sign. 
 \end{proof}

\subsection{Weyl circle in the regular case}

In this subsection we fix a regular operator $L$ whose potential has
a negative imaginary part. We study solutions of $(L-\lambda)f=0$
for $\Im\lambda>0$ with real boundary conditions. In Theorem
\ref{th:weyl} we show that they define a certain circle in the complex plane called the {\em Weyl circle}.
This result  will be needed in the next subsection
\S\ref{ss:lplc}, where general boundary conditions
are studied.

We will use an argument essentially due to H.\ Weyl
in the real case, cf.\ \cite{CL,Nai,Pry} for example.  The Weyl circle for potentials
with semi-bounded imaginary part were first treated in \cite{Si}, see
\cite{BCEP} for more recent results.

Let us denote $U=\Im(\lambda-V)$ and
\begin{equation}\label{eq:newscp}
  \bracet{f}{g}_U=\int_a^b \bar{f}gU .
\end{equation}
We set $\|f\|_U^2=\bracet{f}{f}_U$ and note that if $U\geq 0$, then
$\bracet{\cdot}{\cdot}_U$ is a positive hermitian form and we denote
$\|\cdot\|_U$ is the corresponding seminorm.  Now if $f,g\in\cD(L_\M)$
and $Lf=\lambda f$, $Lg=\lambda g$ for some complex number $\lambda$
then (\ref{lsh}) can be rewritten as
\begin{equation}\label{eq:lagrangecl2}
  2\i\bracet{f}{g}_U= W_a(\bar f,g) -W_b(f,g) .
\end{equation}

\begin{theorem} \label{th:weyl} Assume that $\Im V\leq0$ and
  $\Im\lambda>0$.  Let $u,v$ be solutions of the equation
  $Lf=\lambda f$ with real boundary condition at $a$ and satisfying
  $ W(v,u)=1.$ If $w$ is a solution of $Lf=\lambda f$ with a real
  boundary condition at $b$, then there is a unique $m\in\C$ such that
  $w=mu+v$; this number is on the circle
  \begin{equation}\label{eq:weylcircle0}
  \int_a^b|mu+v|^2\Im(\lambda-V) =\Im m,
  \end{equation}
  which has
  \begin{equation}\label{eq:weylcircle}
  \text{center }
  c=\frac{\i/2-\bracet{u}{v}_U}{\|u\|^2_U}=\frac{W_b(\bar
    u,v)}{2\i\|u\|_U^2} \quad\text{and radius } r=\frac{1}{2\|u\|_U^2} .
\end{equation}
Conversely, let $m$ be a complex number on the circle
\eqref{eq:weylcircle0}, and define $w$ by $w=mu+v$. Then $w$ has a
real boundary condition at $b$ and $W(w,u)=1$.
\end{theorem}

\begin{proof}
  From Lemma \ref{lm:easy} (2) and the reality of the boundary
  conditions at $a$ we get
  \begin{equation}\label{mospilan}
    W_a(\bar u,u)=0,\quad W_a(\bar v,v)=0.
  \end{equation}
  This implies
  \begin{equation} \label{mos2}
    \|u\|_U^2=\frac{\i}{2}W_b(\bar u,u),\quad
    \|v\|_U^2=\frac{\i}{2}W_b(\bar v,v),
  \end{equation}
  due to \eqref{eq:lagrangecl2}. And if $w$ is as in the first part of
  the theorem then the same argument gives
  \begin{equation}
    \|w\|_U^2 =\frac{1}{2\i} W_a(\bar w,w).\label{mos1}
  \end{equation}
  Since $u,v$ are linearly independent solutions of $Lf=\lambda f$, if
  $w$ is another solution then we have $w=mu+nv$ for uniquely
  determined complex numbers $m,n$. Since $W(v,u)=1$ we see that
  $n=1$.
  
  Now fix $w=mu+v$.  Using (\ref{mospilan}) and $W_a(u,v)=-1$, we get
  \begin{equation}\label{eq:req}
    W(\bar w,w)_a
        =|m|^2W_a(\bar u,u)+\bar{m}W_a(\bar u,v)+
    mW_a(\bar v,u)+W_a(\bar v,v) = 2\i\Im m .
  \end{equation}
From (\ref{mos1}) and  \eqref{eq:req} we get
\begin{equation}\label{eq:uveqU}
  \|w\|_U^2 = \Im m .
\end{equation}
From this relation we get
\begin{equation}\label{eq:req2}
  \Im m=\|mu+v\|^2_U =|m|^2\|u\|^2_U
  + 2\Re\big(\bar{m}\bracet{u}{v}_U\big)+ \|v\|^2_U  
\end{equation}
and since $\Im m = 2\Re (\bar{m}\,\i/2)$ we may rewrite this as
\begin{equation} \label{eq:meq}
|m|^2\|u\|^2_U -2\Re\big(\bar{m}(\i/2-\bracet{u}{v}_U)\big)
    +\|v\|^2_U =0 .
\end{equation}
Clearly, $\|w\|_U>0$ hence $\Im m>0$ by
\eqref{eq:uveqU} so \eqref{eq:meq} may be rewritten
\begin{equation} \label{eq:meq2}
  |m|^2
  -2\Re\left(\bar{m}\frac{\i/2-\bracet{u}{v}_U}{\|u\|^2_U}\right)
    +\frac{\|v\|^2_U}{\|u\|^2_U} =0 .
\end{equation}
If $c\in\C$ and $d\in\R$ then
\( |m|^2-2\Re(\bar{m}c) +d=|m-c|^2 - (|c|^2-d) \). Hence there is $m$
such that $|m|^2-2\Re(\bar{m}c) +d=0$ if and only if $d\leq |c|^2$, and
then $|m|^2-2\Re(\bar{m}c) +d=0$ is the equation of a circle with
center $c$ and radius $\sqrt{|c|^2 -d}$.  Thus \eqref{eq:meq2} is the
equation of the circle with
\[
  \text{center}\quad  c=\frac{\i/2-\bracet{u}{v}_U}{\|u\|^2_U}
  \quad\text{and square of radius}\quad 
r^2=\frac{|\i/2-\bracet{u}{v}_U|^2 - \|u\|^2_U \|v\|^2_U}{\|u\|^4_U} .
\]
From \eqref{eq:lagrangecl2} we get
\(
  2\i\bracet{u}{v}_U= W_a(\bar u,v) -W_b(\bar u,v)=-1 -W_b(\bar u,v),
\)
hence $\i/2-\bracet{u}{v}_U  = W_b(\bar u,v)/2\i$. 
Then \eqref{mos2} implies 
\[
\|u\|^2_U \|v\|^2_U=-\frac{1}{4}W_b(\bar u,u) W_b(\bar v,v) ,
\]
hence
\[
|\i/2-\bracet{u}{v}_U|^2 - \|u\|^2_U \|v\|^2_U=
\big(|W_b(\bar u,v)|^2 +W_b(\bar u,u) W_b(\bar v,v) \big)/4 .
\]
But by the Kodaira identity
$W_b(\bar u, u)W_b(\bar v,v) =1-|W_b(\bar u,v)|^2$, hence we get
\[
|\i/2-\bracet{u}{v}_U|^2 - \|u\|^2_U \|v\|^2_U=1/4
\]
so \eqref{eq:meq2} is just the circle described by
\eqref{eq:weylcircle}.

To prove the reciprocal part of the theorem, consider a point $m$ on
this circle and let $w=mu+v$. Clearly $Lw=\lambda w$ and $W(w,u)=1$
and the computation \eqref{eq:req} gives us $W_a(\bar w,w)=2\i\Im
m$. We also have \eqref{eq:req2} because this just says that $m$ is on
the circle \eqref{eq:weylcircle}. Thus we have
\[
  \|w\|_U^2=\Im m = W_a(\bar w,w)/2\i,
\]
and then \eqref{eq:lagrangecl2} implies $W_b(\bar w,w)=0$. Therefore,
by Lemma \ref{eq:lagrangecl2} $w$ has a real boundary condition at
$b$. This proves the final assertion of the theorem.
\end{proof}

\subsection{Limit point/circle} \label{ss:lplc}

In this section we again assume that $\Im V\leq0$ and $\Im\lambda>0$.
We allow $b$ to be an irregular endpoint. We also assume that $a$ is a
regular endpoint. Thus we assume that $V\in L_\loc^1[a,b[$.  This
class of potentials has first been considered in \cite{Si}; see
\cite{BCEP} for more general conditions.
    
Using Theorem \ref{th:weyl}, we will obtain a classification of the
properties of $L$ around $b$ into three categories. This
classification can be called the {\em Weyl trichotomy}.  It replaces
the {\em Weyl dichotomy}, well known classification of irregular
endpoints for real potentials.
    
Note that this classifiction depends only on the behaviour of $V$
close to $b$. In particular, the assumption of regularity of $a$ is
made only for convenience. If $a$ is also irregular, then the analysis
should be done separately on intervals $]a,a_1]$ and $[b_1,b[$.

Let $u,v$ be solutions of $Lf=\lambda f$ on $]a,b[$ with real boundary
conditions at $a$ and such that $W(v,u)=1$.
\begin{definition}
  For any $d\in]a,b[$ we define
  \begin{align*}\text{the \emph{Weyl circle}}\quad
    \mathscr C_d&:=\Big\{m\in\C\mid
                  \int_a^d|mu+v|^2\Im(\lambda-V) =\Im m\Big\},\\
    \text{the \emph{open Weyl disk}}\quad
    \mathscr C_d^\circ&:=\Big\{m\in\C\mid
                        \int_a^d|mu+v|^2\Im(\lambda-V) <\Im m\Big\},\\
    \text{the \emph{closed Weyl disk}}\quad
    \mathscr C_d^\bullet&:=\Big\{m\in\C\mid
    \int_a^d|mu+v|^2\Im(\lambda-V) \leq\Im m\Big\}\,=
    \, \mathscr C_d^\circ\cup \mathscr C_d.
  \end{align*}
\end{definition}
Thus the Weyl circle is given by the condition \eqref{eq:weylcircle0}
with $b$ replaced by $d$. Since the left hand side of
\eqref{eq:weylcircle0} growth like $|m|^2$ when $m\to\infty$, it
follows that $\mathscr C_d^\circ$ is inside $\mathscr C_d$ .  If
$d_1<d_2$ then
\[
  \int_a^{d_2}|mu+v|^2\Im(\lambda-V) \leq\Im m \Rightarrow
\int_a^{d_1}|mu+v|^2\Im(\lambda-V) <\Im m.
\]
Hence $\mathscr C^\bullet_{d_2}\subset\mathscr C^\circ_{d_1}$ strictly
if $d_1<d_2<b$.
\begin{definition}
  We set
  \begin{align*}
    \mathscr C_b^\bullet&:=\bigcap_{d<b}\mathscr C^\bullet_d,\\
    \mathscr C_b&:=\text{the boundary of }\mathscr C_b^\bullet.
  \end{align*}
\end{definition}
It follows that either $\mathscr C_b^\bullet=\mathscr C_b$ is a point,
or $\mathscr C_b^\bullet$ is a disk and $\mathscr C_b$ is a circle of
radius $>0$.
\begin{definition} We say that $b$ is \emph{limit point} if $\mathscr C_b$
is a point. We say that $b$ is \emph{limit circle} if $\mathscr C_b$
is a circle of a positive radius. 
\end{definition}

\begin{lemma}
  Let $m\in \mathscr C_b^\bullet$.  Then
  \begin{equation}\label{mos4}
    \int_a^b|w|^2\Im(\lambda-V)\leq\Im(m).
  \end{equation}
  If $b$ is limit point then $\int_a^b|u|^2\Im(\lambda-V)=\infty$.
\end{lemma}

\begin{proof}
  For any $d\in]a,b[$, we have
  $ \mathscr C_b^\bullet\subset \mathscr C_d^\bullet$. Therefore,
  \begin{equation}\label{mos5}
    \int_a^d|w|^2\Im(\lambda-V)\leq\Im(m).
  \end{equation}
  Then we take the limit $d\nearrow b$. If $b$ is limit point then the
  radius of the Weyl circle $\mathscr C_d$ tends to zero as $d\to b$
  hence $\lim_{d\to b}\int_a^d|u|^2\Im(\lambda-V)=\infty$ by the last
  relation in \eqref{eq:weylcircle}.
\end{proof}

The above lemma implies immediately the following theorem:

\begin{theorem}
If $b$ is limit circle, then all solutions of $(L-\lambda)f-0$ satisfy
\begin{equation} \label{mos3}
  \int_a^b|w|^2\Im(\lambda-V)\quad\text{ is bounded.}
\end{equation}
If $b$ is limit point, then there exists only one (modulo a complex
factor) solution of $(L-\lambda)f=0$ satisfying (\ref{mos3}).
\end{theorem}

Note that $\Im(\lambda-V)\geq\Im\lambda>0$. Therefore, (\ref{mos3})
implies the square integrability of $w$.

Thus, for potentials with a negative imaginary part instead of
Weyl's dichotomy we have three possibilities for solutions
of $(L-\lambda)f=0$ (we consider solutions modulo a complex factor):
\begin{enumerate}
\item limit point case, only one solution satisfies \eqref{mos3},
  only one solution is square integrable;
\item limit point case, only one solution satisfies \eqref{mos3},
  all solutions are square integrable;
\item limit circle, all solutions satisfy \eqref{mos3},
    and hence all solutions are square integrable.
\end{enumerate}
If $V$ is real, then the case (2) is absent and we have the usual
Weyl's dichotomy.  Then $L$ is limit point at $b$ iff for any
$\lambda$ there is at most one solution of $Lf=\lambda f$ which is
square integrable near $b$. But this is not the case if $V$ is
complex.

We emphasize that \emph{the limit point/circle terminology is
  interpreted here in the geometric sense described above} (based on
Theorem \ref{th:weyl}).

There exist examples of (2) in the literature.  \emph{In the limit
  point case, it is possible that we have only one nonzero solution
  satisfying \eqref{mos3}, whereas all solutions are square integrable
  with respect to the Lebesgue measure.} Indeed, Sims \cite[p.\
257]{Si} has shown that this happens in simple examples like
$V(x)= x^6-3\i x^2/2$ on $]1,\infty[$. See also the discussion in
\cite{BCEP}.

\begin{remark}
  We also note that if $V$ is real then for any non-real $\lambda$
  there is at least one nonzero solution of $Lf=\lambda f$ which is
  square integrable near $b$.  But it does not seem to be known
  whether for arbitrary complex $V$ there is $\lambda$ such that
  $Lf=\lambda f$ has a nonzero solution which is square integrable
  near $b$.
\end{remark}

\appendix

\section{Symplectic spaces}  \label{a1}
\protect\setcounter{equation}{0}

Let $\cV$ be a vector space. A bilinear form $\lbra\cdot|\cdot\rbra$
on $\cV$ is called {\em symplectic} if it is antisymmetric and
nondegenerate, i.e.  $ \lbra\phi|\psi\rbra=- \lbra\phi|\psi\rbra$ and
for any $\phi\in\cV$ there exists $\psi\in\cV$ such that
$ \lbra\phi|\psi\rbra\neq0$. For a subspace $\cW$ in $\cV$ we define
its {\em symplectic orthogonal complement}
\[
  \cW^{\,\s\!\perp}:=\{\phi\mid\lbra \phi|\psi\rbra=0,\quad
  \psi\in\cW\}.
\]
$\cW$ is called {\em isotropic} if $\cW\subset\cW^{\,\s\!\perp}$ and
Lagrangian if $\cW=\cW^{\,\s\!\perp}$.

The following proposition is well-known and easy to prove.  Perhaps
the only nontrivial point is (4) for infinite-dimensional $\cV$, where
the usual induction argument needs to use the Zorn Lemma.

\begin{proposition}  $ $
   \begin{enumerate}
  \item
    If $\cV$ is a symplectic space, then $\dim\cV$ is even or infinite.
 \item If $\cW$ is a subspace of  $\cV$, then
   $\dim\cW=\dim\cV/\cW^\sperp.$
   \item
  If $\cW$ is a Lagrangian subspace, then $\dim\cW=\frac12\dim\cV$.
\item There exist Lagrangian subspaces in $\cV$.
  \end{enumerate}
  \label{sympl4}
\end{proposition}

Assume in addition that $\cV$ is a Banach space with norm $\|\cdot\|$.
We say that the symplectic form is continuous if
\begin{equation}
  |  \lbra\phi|\psi\rbra|\leq c\|\phi\|\|\psi\|.
\end{equation}

\begin{proposition} \label{sympl5} If $\cV$ is
    a Banach space and the symplectic form is continuous, then:
  \begin{enumerate}
  \item If $\cW$ is a subspace of $\cV$, then $\cW^\sperp$ is closed.
  \item Every Lagrangian subspace of $\cV$ is closed.
  \end{enumerate}
\end{proposition}

\begin{proof} (1) is obvious. It implies (2). \end{proof}

Let $\cV'$ denote the dual space of $\cV$. We have the map
\begin{equation}
  \cV\ni v\mapsto v':=\lbra v|\cdot\rbra\in\cV'.
\end{equation}
If $\dim\cV$ is finite, then this is an isomorphism.
$\cV'$ is then equipped with a symplectic form
\begin{equation}\label{symplo}
\lbra v'|w'\rbra:=\lbra v|w\rbra.
\end{equation}

%\%\%\%\%\%\%\%\%\%\%\%\%\%\%\%\%\%\%\%\%\%\%\%\%\%\%\%\%\%\%\%\%\%\%


\begin{thebibliography}{999}
  
\bibitem{At} Atkinson, F.V.: {\em Discrete and Continuous
      Boundary Value Problems}, Academic Press, 1964.

  
\bibitem{BBMNP} Boitsev, A.A., Brasche, J.F., Malamud, M.M.,
  Neidhardt, H., Popov, I.Yu.: Boundary triplets, tensor products and
  point contacts to reservoirs, preprint 2017

\bibitem{BCEP} {Brown, B.~M., McCormack, D.~K.~R., Evans, W.~D.,
  Plum, M.: On the spectrum of the second-order differential operators
  with complex coefficients, Proc.\ R.\ Soc.\ Lond.\ A {\bf 455},
  1234--1257, 1999.}
  
\bibitem{BDG} Bruneau, L., Derezi\'nski, J., Georgescu, V.:
  {Homogeneous Schr\"odinger operators on half-line}, Ann. Henri
  Poincar\'e {\bf 12} no. 3 (2011), 547--590.

\bibitem{CL} Coddington, E.A., Levinson, N.: Theory of Ordinary
  Differential Equations, McGraw-Hill, 1955

\bibitem{Cal} Calkin, J.W.: {Abstract Symmetric Boundary
    Conditions}, Trans. Amer. Math. Soc. 45, 369--442, 1939.

\bibitem{DAR} De Alfaro, V., Regge, T.: Potential scattering, North
  Holland Publishing Co. Amsterdam 1965
  
\bibitem{D} Derezi\'{n}ski, J.: Homogeneous rank one perturbations,
  Annales Henri Poincare 18 (2017) 3249-3268.

\bibitem{DR} Derezi\'nski, J., Richard, S.: {On Schr\"odinger
    operators with inverse square potentials on the half-line},

\bibitem{DR2} Derezi\'nski, J., Richard, S.: On radial Schrodinger
  operators with a Coulomb potential, arxiv.

\bibitem{DS} Derezi\'nski, J., Siemssen, D.:, An Evolution Equation
  Approach to the Klein–Gordon Operator on Curved Spacetime, to appear
  in Pure and Applied Analysis.

\bibitem{DS2} {Dunford, N., Schwartz J.~T.: {\em Linear Operators
    Part II: Spectral Operators, Chap. XIII Ordinary Differential
    Operators}, Wiley Interscience, 1967.}
  
\bibitem{DS3} Dunford, N., Schwartz J.~T.: {\em Linear Operators
    vol. III, Spectral Operators, Chap. XX.1.1 Spectral Differential
    Operators of Second Order}, Wiley Interscience, 1971.

\bibitem{EE} Edmunds, D.E., Evans, W. D., {\em Spectral Theory
      and Differential Operators}, Oxford University Press, 2nd ed.,
    2018.

\bibitem{EZ} Everitt, W. N., Zettl, A.: Generalized symmetric
    ordinary differential expressions. I. The general theory, Nieuw
    Arch. Wisk. (3) {\bf27}(3), 363-397, 1979.
  
\bibitem{Gal} Galindo, A.: On the existence of J-self-adjoint
    extensions of J-symmetric operators with adjoint, Comm. Pure and
    Appl. Math. {\bf15}, 423–425, 1962.
  
\bibitem{GTV} Gitman, D.M., Tyutin, I.V., Voronov, B.L.:
  \emph{Self-adjoint extensions in quantum mechanics.  General theory
    and applications to Schr\"odinger and Dirac equations with
    singular potentials}, Progress in Mathematical Physics {\bf 62},
  Birkh\"auser/Springer, New York, 2012.

\bibitem{Kato} Kato, T.: {\em Perturbation Theory for Linear
    Operators}, second edition, Springer-Verlag, Berlin 1976.

\bibitem{Kn} Knowles, I.: On J-self-adjoint extensions of
    J-symmetric operators, Proc. Amer. Math. Soc. {\bf79}, 42–44,
    1980.
  
\bibitem{Nai} Naimark, M.~A.: {\em Linear Differential Operators Vol.\
    II (1952)}, Harrap, London, 1968.

\bibitem{Pry} Pryce, J.~D.: {\em Numerical Solution of Sturm-Liouville
    Problems}, Clarendon Press, 1993.

\bibitem{Rac} Race, D.: {Theory of J-self-adjoint Extensions of
    J-Symmetric Operators}, J.\! Diff.\! Eq.\! {\bf 57}, 258-274,\! 1985.
  
\bibitem{RS2} Reed, M., Simon, B.: {\em Methods of Modern Mathematical
    Physics, II. Fourier Analysis, Self-Adjointness}, London, Academic
  Press 1975.

\bibitem{Si} Sims, A.~R.: Secondary Conditions for Linear
    Differential Operators of the Second Order, J.\ Math.\ Mech.\
    {\bf 6}, 247--285 (1957).  
  
\bibitem{S} Stone, M.~H.: {\em Linear Transformations in Hilbert Space
    and their Applications to Analysis}, Amer.\ Math.\ Soc.\ 1932.

\bibitem{Te} Teschl, G.: {\em Mathematical Methods in Quantum
  Mechanics; With Applications to Schr\"odinger Operators}, AMS
  Graduate Studies in Mathematics 99, 2009.

\bibitem{Ti} Titchmarsh, E.C.: {\em Eigenfunction expansions associated
    with second order differential equations, Vol. I}, second edition,
  Oxford University Press, 1962.

\bibitem{Vi} Vishik, M. L: On general boundary problems for elliptic
  differential equations. Trudy Moscov. Mat. Obsc. 1 (1952),
  187–246. English translation Amer. Math. Soc. Transl., Ser. 2, 24
  (1963), 107–172.

  
\bibitem{Y} Yosida, K.: {\em Functional Analysis}, fifth
  edition, Springer-Verlag, 1978.

  
\end{thebibliography}
\end{document}